\PassOptionsToPackage{table}{xcolor}
\documentclass[a4paper, 11pt]{article}

\title{Logical Characterizations of Recurrent Graph Neural Networks with Reals and Floats} 

\date{} 					

\usepackage{tikz}
\usepackage{multicol}

\usepackage{arydshln}
\usepackage{makecell}

\usepackage{colonequals}
\usepackage{amsmath, amssymb, amsthm}
\usepackage{mathtools}
\usepackage{thm-restate}
\usepackage{booktabs}
\usepackage[raggedright]{titlesec}

\PassOptionsToPackage{numbers, compress}{natbib}

\usepackage{authblk}

\author[1]{Veeti Ahvonen}
\author[1]{Damian Heiman}
\author[1]{Antti Kuusisto}
\author[2]{Carsten Lutz}

\affil[1]{Tampere University, Finland}
\affil[2]{Leipzig University, Germany}

\usepackage{graphicx} 

 \usepackage[textwidth=151mm, hcentering, vmargin={24mm,32.5mm}, foot=20mm]{geometry}

\usepackage{hyperref}
\usepackage{verbatim}

\usepackage[utf8]{inputenc}
\usepackage[T1]{fontenc}

\usepackage{enumitem}
\usepackage{tikz}
\usepackage{lipsum}
\usepackage{pdfpages}
\usepackage{multicol} 
\usepackage{xspace} 
\usepackage{stmaryrd}

\usepackage{makecell}
\usepackage{mathtools}
\usepackage{colonequals}

\usepackage{marginnote}

\usepackage{bm}

\theoremstyle{plain}
\newtheorem{theorem}{Theorem}[section]
\newtheorem{lemma}[theorem]{Lemma}
\newtheorem{proposition}[theorem]{Proposition}
\newtheorem{corollary}[theorem]{Corollary}

\theoremstyle{definition}
\newtheorem{remark}[theorem]{Remark}
\newtheorem{example}[theorem]{Example}
\newtheorem{definition}[theorem]{Definition}

\newcommand{\N}{\mathbb N}
\newcommand{\Z}{\mathbb Z}

\newcommand{\R}{\mathbb R}

\newcommand{\bb}{\mathbf{b}}

\newcommand{\bd}{\mathbf{d}}

\newcommand{\bu}{\mathbf{u}}
\newcommand{\bv}{\mathbf{v}}
\newcommand{\bw}{\mathbf{w}}
\newcommand{\bx}{\mathbf{x}}
\newcommand{\by}{\mathbf{y}}

\newcommand{\cA}{\mathcal{A}}
\newcommand{\cB}{\mathcal{B}}
\newcommand{\cC}{\mathcal{C}}

\newcommand{\cG}{\mathcal{G}}

\newcommand{\cK}{\mathcal{K}}
\newcommand{\cL}{\mathcal{L}}
\newcommand{\cM}{\mathcal{M}}
\newcommand{\cN}{\mathcal{N}}

\newcommand{\cP}{\mathcal{P}}
\newcommand{\cQ}{\mathcal{Q}}
\newcommand{\cR}{\mathcal{R}}
\newcommand{\cS}{\mathcal{S}}
\newcommand{\cT}{\mathcal{T}}

\newcommand{\cX}{\mathcal{X}}

\newcommand{\fC}{\mathfrak{C}}

\newcommand{\fG}{\mathfrak{G}}

\newcommand{\fP}{\mathfrak{P}}

\newcommand{\fR}{\mathfrak{R}}

\newcommand*{\abs}[1]{\lvert#1\rvert}

\newcommand{\MSO}{\mathrm{MSO}}
\newcommand{\FO}{\mathrm{FO}}
\newcommand{\dom}{\mathrm{dom}}
\newcommand{\MSC}{\mathrm{MSC}}
\newcommand{\GMSC}{\mathrm{GMSC}}
\newcommand{\GMSCG}{\mathrm{GMSC+G}}
\newcommand{\GMSCA}{\mathrm{GMSC+A}}
\newcommand{\GMSCAG}{\mathrm{GMSC+AG}}

\newcommand{\GMSCN}{\mathrm{GMSC[1]}}
\newcommand{\GMSCGN}{\mathrm{GMSC[1]+G}}
\newcommand{\GMSCAN}{\mathrm{GMSC[1]+A}}
\newcommand{\GMSCAGN}{\mathrm{GMSC[1]+AG}}

\newcommand{\GML}{\mathrm{GML}}
\newcommand{\GMLG}{\mathrm{GML+G}}

\newcommand{\AGG}{\mathrm{AGG}}
\newcommand{\READ}{\mathrm{READ}}
\newcommand{\COM}{\mathrm{COM}}

\newcommand{\CMPA}{\mathrm{CMPA}}
\newcommand{\CMPAG}{\mathrm{CMPA+G}}
\newcommand{\CMPAA}{\mathrm{CMPA+A}}
\newcommand{\CMPAAG}{\mathrm{CMPA+AG}}
\newcommand{\FCMPA}{\mathrm{FCMPA}}
\newcommand{\FCMPAG}{\mathrm{FCMPA+G}}

\newcommand{\FCMPAAG}{\mathrm{FCMPA+AG}}

\newcommand{\GNN}{\mathrm{GNN}}
\newcommand{\GNNG}{\mathrm{GNN[\R] + G}}
\newcommand{\GNNA}{\mathrm{GNN[\R] + A}}
\newcommand{\GNNAG}{\mathrm{GNN[\R] + AG}}
\newcommand{\GNNF}{\mathrm{GNN[F]}}
\newcommand{\GNNFG}{\mathrm{GNN[F] + G}}
\newcommand{\GNNFA}{\mathrm{GNN[F] + A}}
\newcommand{\GNNFAG}{\mathrm{GNN[F] + AG}}

\newcommand{\VGML}{\omega\text{-}\mathrm{GML}}

\newcommand{\ordo}{\mathcal{O}}

\usepackage{array,multirow}
\begin{document}

\maketitle

\begin{abstract} 
In pioneering work from 2019, Barceló and coauthors identified logics that precisely match the expressive power of constant iteration-depth graph neural networks (GNNs) relative to properties definable in first-order logic. In this article, we give exact logical characterizations of recurrent GNNs in two scenarios: (1) in the setting with floating-point numbers and (2) with reals. For floats, the formalism matching recurrent GNNs is a rule-based modal logic with counting, while for reals we use a suitable infinitary modal logic, also with counting. These results give exact matches between logics and GNNs in the recurrent setting without relativising to a background logic in either case, but using some natural assumptions about floating-point arithmetic. Applying our characterizations, we also prove that, relative to graph properties definable in monadic second-order logic (MSO), our infinitary and rule-based logics are equally expressive. This implies that recurrent GNNs with reals and floats have the same expressive power over MSO-definable properties and shows that, for such properties, also recurrent GNNs with reals are characterized by a (finitary!) rule-based modal logic. In the general case, in contrast, the expressive power  with floats is weaker than with reals. In addition to logic-oriented results, we also characterize recurrent GNNs, with both reals and floats, via distributed automata,  drawing links to distributed computing models.  
\end{abstract}

\section{Introduction}

Graph Neural Networks (GNNs) 
\cite{gori, scarcelli, zonghan-survey} 
have proven to be highly useful for processing graph data in numerous applications that span
a remarkable range of disciplines
including bioinformatics~\cite{biosurvey}, recommender systems \cite{recommendersurvey}, traffic forecasting \cite{trafficforecastingsurvey}, 
and a multitude of others. The success of GNNs in applications has 
stimulated lively research into their theoretical properties
such as expressive power. A landmark result is due to 
Barceló et al. \cite{DBLP:conf/iclr/BarceloKM0RS20} which was among
the first to characterize the expressive power of GNNs in terms of
\emph{logic}, see \cite{DBLP:conf/lics/Grohe21,DBLP:conf/kr/CucalaGMK23,benediktandothers,Pfluger_Tena_Cucala_Kostylev_2024} and references therein for related results.
More precisely, Barceló et al.\ show that a basic GNN model with a 
constant number of iterations has exactly the same expressive power
as graded modal logic $\GML$ in restriction to properties definable in first-order logic FO.

In this article, we advance the analysis of the expressive
power of GNNs in two directions. First, we study the relation between GNN models 
based on real numbers, as mostly studied in theory, and GNN models 
based on floating-point numbers, as mostly used in practice. And 
second, we focus on a family of basic recurrent GNNs while
previous research has mainly considered GNNs with a constant number
of iterations, with the notable exception of
 \cite{Pfluger_Tena_Cucala_Kostylev_2024}.
The GNNs studied in the current paper have a simple and natural termination (or acceptance) condition: termination
is signaled via designated feature vectors and thus the GNN ``decides''
by itself about when to terminate. 
We remark that some of our results also apply to constant iteration GNNs and to recurrent GNNs with a 
termination condition based on fixed points, as used in the inaugural work on GNNs~\cite{gori};
see the conclusion section for further details.

We provide three main results. The first one is that recurrent GNNs with floats, or $\GNNF$s, have the same expressive power as the \emph{graded modal substitution calculus} $\GMSC$ \cite{Kuusisto13, dist_circ_mfcs, ahvonen_et_al:LIPIcs.CSL.2024.9}. This is a rule-based modal logic  
that extends the \emph{modal substitution calculus} $\mathrm{MSC}$ \cite{Kuusisto13} with counting modalities. $\MSC$ has been shown to precisely correspond to distributed computation models based on automata \cite{Kuusisto13} and Boolean circuits \cite{dist_circ_mfcs}. 
$\GMSC$ is related to the graded modal $\mu$-calculus, but orthogonal in expressive power. 
The correspondence between $\GNNF$s and $\GMSC$ is as follows.

\textbf{Theorem \ref{theorem: k-GNN[F] = k-FCMPA = GMSC}.}
\emph{The following have the same expressive power: $\GNNF$s, $\GMSC$, and R-simple aggregate-combine $\GNNF$s.}

Here \emph{R-simple aggregate-combine} $\GNNF$s
mean $\GNNF$s that use basic aggregate-combine functions as specified by Barceló et al.\ \cite{DBLP:conf/iclr/BarceloKM0RS20} and the truncated $\mathrm{ReLU}$ as the non-linearity function, see Section \ref{gnndefsjooj} for the formalities.  
The theorem shows that an R-simple model of $\GNNF$s suffices, and in
fact $\GNNF$s with a more complex architecture can be turned into
equivalent R-simple ones. We emphasize that the characterization provided by Theorem~\ref{theorem: k-GNN[F] = k-FCMPA = GMSC} is \emph{absolute}, that is, not relative to a background logic. In contrast, the characterization by Barceló et al.\ \cite{DBLP:conf/iclr/BarceloKM0RS20} is relative to first-order logic. Our characterization does rely, however, on an assumption about floating-point arithmetics. We believe that this assumption is entirely natural as it reflects practical implementations of floats.

Our second result shows that recurrent GNNs with reals, or $\GNN[\R]$s, have the same expressive power as the infinitary modal logic $\VGML$ that consists of infinite disjunctions of $\GML$-formulas.

\textbf{Theorem \ref{omega-GML = GNN = CMPAs}.} $\GNN[\R]$s have the same expressive power as  $\VGML$.

Again, this result is absolute. 
As we assume no restrictions on the arithmetics used in  $\GNN[\R]$s, they are very powerful: 
with the infinitary logic $\VGML$
it is easy to define even undecidable graph properties.
We regard the theorem as an interesting theoretical upper bound on the expressivity of GNNs operating in an unrestricted, recurrent message passing scenario with messages flowing to neighbouring nodes. 
We note that $\GNN[\R]$s can 
easily be shown more expressive
than $\GNNF$s.

Our third result considers $\GNN[\R]$s and $\GNNF$s relative to a very expressive background logic, probably the most natural choice in the recurrent GNN context: \emph{monadic second-order logic} $\MSO$.

\textbf{Theorem \ref{theorem: under MSO: GNN[R] = GNN[F]}.}
\emph{Let $\cP$ be a property expressible in $\MSO$. Then $\cP$ is expressible as a $\GNN[\R]$ if and only if it is expressible as a $\GNNF$.}

This result says that, remarkably, for the very significant and large class of $\MSO$-expressible properties, using actual reals with unrestricted arithmetic gives no more expressive power than using floats, by Theorem~\ref{theorem: k-GNN[F] = k-FCMPA = GMSC} even in the R-simple aggregate-combine setting. Thus, for this class of properties, the theoretical analyses
from the literature do not diverge from the practical implementation!
Taken together, the above results also imply that in restriction to
$\mathrm{MSO}$-expressible properties, $\GNN[\R]$s are equivalent to the
(finitary!) graded modal substitution calculus $\GMSC$.

We also develop characterizations of GNNs in terms of distributed
automata. These are in fact crucial tools in our proofs, but the
characterizations are also interesting in their own right as they
build links between GNNs and distributed computing.  
We study a class of distributed automata called counting message passing 
automata
($\CMPA$s) that may have a countably infinite number of states. Informally, these distributed automata update the state of each node according to the node's own state and the \emph{multiset} of states received from its out-neighbours. We also study their restriction that admits only a finite number of states
($\FCMPA$s) and, furthermore, bounded $\FCMPA$s.  
In the bounded case the multiplicities of states in the received multisets are bounded by some constant $k \in \N$.
A summary of our main results, both absolute and relative to $\MSO$, is given in Table~\ref{introtable}.
\begin{table}[h]
  \centering
  \caption{A summary of our main results. 
  The first row contains the results obtained without relativising to a background logic and the second row contains results relative to $\MSO$. Here $x \equiv y$ means that $x$ and $y$ have the same expressive power while $x < y$  means that $y$ is strictly more expressive than $x$. Further, ``bnd.'' stands for ``bounded'' and ``AC'' for ``aggregate-combine''.}
  \label{introtable}
  {\renewcommand{\arraystretch}{1.25}
  \small
  \resizebox{\textwidth}{!}{\begin{tabular}[t]{rc}
    \hline 
    Absolute:\!\! &    $\GNNF\! \equiv\! \text{R-simple\ AC-} \GNNF \! \equiv\! \GMSC \! \equiv\! \text{bnd. } \FCMPA
    \! < \! \GNN[\R] \! \equiv\! \VGML \equiv\! \CMPA \!$
     \\ 
    \hline 
    $\MSO$:\!\! &  $\GNNF \! \equiv\! \text{R-simple\ AC-} \GNNF \! \equiv\! \GMSC \! \equiv\! \text{bnd. } \FCMPA
    \! \equiv\! \GNN[\R] \! \equiv\! \VGML \equiv\! \CMPA \!$\\ 
    \hline 
  \end{tabular}}}
\end{table}

In Appendix \ref{appendix: On accepting}, we extend Theorem \ref{theorem: k-GNN[F] = k-FCMPA = GMSC} for GNNs with global readout and/or basically any other acceptance condition, which involves augmenting $\GMSC$ with a counting global modality and modifying the acceptance condition of $\GMSC$ analogously. In particular, we show that the result holds for the acceptance conditions studied in~\cite{Pfluger_Tena_Cucala_Kostylev_2024}: one where the number of rounds is determined by graph-size and one where a node is accepted when it reaches an accepting fixed-point. We also prove the result for the condition in the seminal GNN articles \cite{scarcelli, gori}, where the feature vector of each node is required to converge.

\textbf{Related Work.} 
Barceló et al. \cite{DBLP:conf/iclr/BarceloKM0RS20} study aggregate-combine GNNs with a constant number of iterations. They characterize these $\GNN$s---in restriction to properties expressible in first-order logic---in terms of graded modal logic $\GML$ and show that a \emph{global readout mechanism} leads to a model equivalent to the two-variable fragment of first-order logic with counting quantifiers $\mathrm{FOC}^2$. Our work extends the former
result to include 
recurrence in a natural way while we leave studying global readouts as future work; see the conclusion section for further details.
Grohe \cite{DBLP:conf/lics/Grohe23} connects the \emph{guarded fragment of first-order logic with counting} $\mathrm{GFO}+\mathrm{C}$ and polynomial-size bounded-depth circuits, 
linking (non-recurrent) GNNs to the circuit complexity class $\mathrm{TC}^0$. Grohe's characterization utilises dyadic rationals rather than floating-point numbers.
Benedikt et al.\ \cite{benediktandothers} use logics with a generalized form of counting via \emph{Presburger quantifiers} to obtain characterizations for (non-recurrent) GNNs with a bounded activation function. The article also investigates questions of decidability concerning GNNs---a topic we will not study here. 
As a general remark on related work, it is worth mentioning that the expressive power of (basic) recurrent GNN-models is invariant under the Weisfeiler-Leman test. This link has been recently studied in numerous articles \cite{morris19, DBLP:conf/iclr/XuHLJ19, DBLP:conf/lics/Grohe21, DBLP:conf/iclr/BarceloKM0RS20}.

Pfluger et al.\ \cite{Pfluger_Tena_Cucala_Kostylev_2024} investigate recurrent GNNs with two kinds of termination conditions: one based on reaching a fixed point when iteratively generating 
feature vectors, and one where termination occurs after a  number of rounds determined by the size of the input graph. They concentrate on the case of unrestricted aggregation and combination functions, even including all
uncomputable ones. Their main result is relative to a logic LocMMFP introduced specifically
for this purpose, extending first-order logic with a least fixed-point operator over unary monotone formulas. The characterization itself is given in 
terms of the graded two-way $\mu$-calculus. We remark that $\MSO$ significantly
generalizes LocMMFP and that the graded two-way $\mu$-calculus is
incomparable in expressive power to our $\GMSC$.
In contrast to our work and to Barcelo et al.\ \cite{DBLP:conf/iclr/BarceloKM0RS20}, Pfluger et al.\ do not 
discuss the case where the aggregation and combination functions of the $\GNN$s are R-simple or restricted in any other way.
We view our work as complementing yet being in the spirit of both Barceló et al. \cite{DBLP:conf/iclr/BarceloKM0RS20} and Pfluger et al.~\cite{Pfluger_Tena_Cucala_Kostylev_2024}.

GNNs are essentially distributed systems, and logical characterizations for distributed systems have been studied widely. A related research direction begins with Hella et al. \cite{hella2012weak}, Kuusisto \cite{Kuusisto13} and Hella et al. \cite{weak_models} by results linking distributed computation models to modal logics. 
The articles \cite{hella2012weak} and \cite{weak_models} give characterizations of constant-iteration scenarios with modal logics, and \cite{Kuusisto13} lifts the approach to recurrent message-passing algorithms via showing that the modal substitution calculus $\MSC$ captures the expressivity of finite distributed message passing automata. This generalizes the result from \cite{hella2012weak} that characterized the closely related class $\mathsf{SB(1)}$ of local distributed algorithms with modal logic. 
Later Ahvonen et al. \cite{dist_circ_mfcs} showed a match between $\MSC$ and circuit-based distributed systems. Building on the work on $\MSC$, Reiter showed in \cite{reiter17} that the $\mu$-fragment of the modal $\mu$-calculus captures the expressivity of finite message passing automata in the asynchronous scenario.

\section{Preliminaries}

We let $\N$, $\Z_+$ 
and $\R$ denote the sets of non-negative integers, positive
integers, and real numbers respectively.
For all $n \in \Z_+$, we let $[n] \colonequals
\{1, \dots, n\}$ and for all $n \in \N$, we let $[0;n] \colonequals \{0, \dots,
n\}$.
With $\abs{X}$ we denote the cardinality of the set $X$, with $\cP(X)$ the power set of $X$ and with $\cM(X)$ the set of multisets over $X$, i.e., the set of functions $X \to \N$. With $\cM_k(X)$ we denote the set of $k$-multisets over $X$, i.e., the set of functions $X \to [0;k]$.
Given a $k$-multiset $M \in \cM_k(X)$ and $x \in X$, intuitively $M(x) = n < k$ means that there are exactly $n$ copies of $x$ and $M(x) = k$ means that there are $k$ \emph{or more} copies of $x$.

We work with node-labeled \textbf{directed graphs} (possibly with self-loops), and \emph{simply refer to them as \textbf{graphs}}.
Let  $\mathrm{LAB}$ 
denote a countably infinite set of \textbf{node label symbols}, representing features. We denote finite sets of node label symbols by $\Pi \subseteq \mathrm{LAB}$.
Given any $\Pi \subseteq \mathrm{LAB}$, 
a \textbf{$\Pi$-labeled graph} is a triple $G = (V, E, \lambda)$ where $V$ is a set of \textbf{nodes}, 
$E \subseteq V \times V$ is a set of \textbf{edges} and $\lambda \colon V \to \cP(\Pi)$ 
is a \textbf{node labeling function}. Note that a node can carry multiple
label symbols. 
A \textbf{pointed graph} is a pair $(G, v)$ with $v \in V$. 
Given a graph $(V, E, \lambda)$, 
the set of \textbf{out-neighbours} of $v \in V$ is $\{\, w \mid (v, w) \in E
\,\}$. 
Unless stated otherwise, we only consider \emph{finite} graphs, i.e., graphs where the set of nodes is finite.
A \textbf{node property over $\Pi$} is a class of pointed $\Pi$-labeled graphs.
A \textbf{graph property over} $\Pi$ is a class of $\Pi$-labeled graphs. A graph property $\cG$ over $\Pi$ corresponds to a node property $\cN$ over $\Pi$ if the following holds for all $\Pi$-labeled graphs $G$: $G \in \cG$ iff $(G, v) \in \cN$ for every node $v$ of $G$.
Henceforth a property means a node property.
We note that many of our results hold even with infinite graphs and infinite sets of node label symbols. Our results easily extend to graphs that admit labels on both nodes and edges.

\subsection{Graph neural networks}\label{gnndefsjooj}

A graph neural network (GNN) is a neural network architecture for graph-structured data. It may be 
viewed as a distributed system where the nodes of the (directed, node-labeled) input graph calculate with real numbers and communicate synchronously in discrete rounds. More formally,  a \textbf{recurrent graph neural network} $\GNN[\R]$ over $(\Pi, d)$, with $\Pi \subseteq \mathrm{LAB}$ and $d \in \Z_{+}$, is a tuple $\cG = (\R^{d}, \pi, \delta, F)$. 
A recurrent graph neural network computes in a (node-labeled) directed graph as follows.
In any $\Pi$-labeled graph $(V, E, \lambda)$, the \textbf{initialization function} $\pi \colon \cP(\Pi) \to \R^{d}$ assigns to each node $v$ an initial \textbf{feature vector} or \textbf{state} $x_{v}^{0} = \pi(\lambda(v))$.\footnote{In \cite{DBLP:conf/iclr/BarceloKM0RS20, Pfluger_Tena_Cucala_Kostylev_2024} an initialization function is not explicitly included in GNNs; instead 
each node is labeled with an initial feature vector in place of node label symbols.
However, these two approaches are essentially the same.}
In each subsequent round $t = 1, 2, \dots$, every node computes a new feature vector $x_{v}^{t}$ using a \textbf{transition function} $\delta \colon \R^{d} \times \cM(\R^{d}) \to \R^{d}, \delta(x, y) = \COM(x, \AGG(y))$, which is a composition of an \textbf{aggregation function} $\AGG \colon \cM(\R^{d}) \to \R^{d}$ (typically sum, min, max or average) and a \textbf{combination function} $\COM \colon \R^{d} \times \R^{d} \to \R^{d}$ such that 
\[
    x_{v}^{t} = \COM \left( x_{v}^{t-1}, \AGG \left( \{\!\{\, x_{u}^{t-1} \mid (v, u) \in E \,\}\!\} \right)\right)
\]
where double curly brackets $\{\{ ... \}\}$ denote multisets.\footnote{In GNNs here, messages flow in the direction opposite to the edges of graphs, i.e., an edge $(v, u) \in E$ means that $v$ receives messages sent by $u$. This is only a convention that could be reversed via using a modal logic that scans the \emph{inverse relation} of $E$ instead of $E$.} 
The recurrent $\GNN$ $\cG$ \textbf{accepts} a pointed $\Pi$-labeled graph $(G, v)$ if $v$ visits (at least once) a state in the set $F \subseteq \R^{d}$ of \textbf{accepting feature vectors}, i.e., $x_{v}^{t} \in F$ for some $t \in \N$. 
When we do not need to specify~$d$, we may refer to a $\GNN[\R]$ over $(\Pi, d)$ as a $\GNN[\R]$ over $\Pi$.
A \textbf{constant-iteration} $\GNN[\R]$ is a pair $(\cG, N)$ where $\cG$ is a $\GNN[\R]$ and $N \in \N$. It \textbf{accepts} a pointed graph~$(G, v)$ if $x^{N}_{v} \in F$. Informally, we simply run a $\GNN[\R]$ for $N$ iterations and accept (or do not accept) based on the last iteration.
We say that $\cG$ (resp., $(\cG, N)$) \textbf{expresses} a node property $\cP$ over $\Pi$, if for each pointed $\Pi$-labeled graph $(G, w)$: $(G, w) \in \cP$ iff $\cG$ (resp., $(\cG, N)$) accepts~$(G, w)$. A node property $\cP$ over $\Pi$ is \textbf{expressible} as a $\GNN[\R]$ (resp. as a constant-iteration $\GNN[\R]$) if there is a $\GNN[\R]$ (resp. constant-iteration $\GNN[\R]$) expressing $\cP$.

One common, useful and simple possibility for the aggregation and combination functions, which is also used by Barceló et al.\ (see \cite{DBLP:conf/iclr/BarceloKM0RS20}, and also the papers \cite{morris19, DBLP:conf/nips/HamiltonYL17})  is defined by
\[
    \COM \left( x_{v}^{t-1}, \AGG \left( \{\{\, x_{u}^{t-1} \mid (v, u) \in E \,\}\} \right)\right) = f(x_v^{t-1} \cdot C + \sum_{(v, u) \in E} x_u^{t-1} \cdot A + \bb), 
\]
where $f \colon \R^d \to \R^d$ is a non-linearity function (such as the truncated $\mathrm{ReLU}$  also used by Barceló et al. in \cite{DBLP:conf/iclr/BarceloKM0RS20}, defined by $\mathrm{ReLU}^*(x) = \min(\max(0,x),1)$ and applied separately to each vector element), $C, A \in \R^{d \times d}$ are matrices and $\bb \in \R^d$ is a bias vector. 
We refer to $\GNN$s that use aggregation and combination functions of this form and $\mathrm{ReLU}^*$ as the non-linearity function as \textbf{R-simple aggregate-combine} $\GNN$s (here `R' stands 
for $\mathrm{ReLU}^*$).

\begin{example}\label{example: GNN definable properties}
Given $\Pi$ and $p \in \Pi$, \emph{reachability of node label symbol $p$} is the property~$\cP$ over~$\Pi$ that contains those pointed $\Pi$-labeled graphs $(G,v)$ where a path exists from $v$ to some $u$ with $p \in \lambda(u)$. An R-simple aggregate-combine $\GNN[\R]$ over $(\Pi, 1)$ (where $C = A = 1$, $\bb = 0$ and $1$ is the only accepting feature vector) can express $\cP$: In round $0$, a node $w$'s state is $1$ if $p \in \lambda(w)$ and else $0$. In later rounds, $w$'s state is $1$ if $w$'s state was $1$ or it gets $1$ from its out-neighbours; else the state is $0$.
\end{example}

\begin{remark}\label{GNN uncomputable}
Notice that unrestricted $\GNN[\R]$s can express, even in a single iteration, node properties such as that the number of immediate out-neighbours
of a node is a prime number. In fact, for any $U\subseteq \mathbb{N}$, including any undecidable set $U$, a $\GNN[\R]$ can express the property that the number $l$ of immediate out-neighbours is 
in the set $U$. See \cite{benediktandothers} for related undecidability
results.
\end{remark}

Informally, a floating-point system contains a finite set of rational numbers and arithmetic operations $\cdot$ and $+$.
Formally, if $p \in \Z_{+}$, $n \in \N$ and $\beta \in \Z_{+}
\setminus \{1\}$, then a \textbf{floating-point system} is a tuple $S
= ((p, n, \beta), +, \cdot)$ that consists of the set $D_{S}$ of all
rationals accurately representable in the form $0.d_{1} \cdots d_{p}
\times \beta^{e}$ or $-0.d_{1} \cdots d_{p} \times \beta^{e}$ where $0
\leq d_{i} \leq~\beta-1$ and $e \in~\{-n, \dots, n\}$. It also consists of arithmetic operations $+$ and $\cdot$ of type $D_{S} \times~D_{S} \to~D_{S}$. We adopt the common convention where $+$ and $\cdot$ are defined by taking a precise sum/product w.r.t. reals and then rounding to the closest float in $D_S$, with ties rounding to the float with an even least significant digit,
e.g., $0.312 + 0.743$ evaluates to $1.06$ if the real sum $1.055$ is not in the float system. 
Thus, our float systems handle overflow by capping at the maximum value instead of wrapping around.
We typically just write $S$ for $D_{S}$.

Consider GNNs using floats in the place of reals. In GNNs, sum is a common aggregation function 
(also used in R-simple $\GNN$s), and the sum of floats can depend on the ordering of floats, since it is not associative. In real-life implementations, the set $V$ of nodes of the graph studied can typically be associated with an implicit linear order relation $<^{V}$ (which is not part of the actual graph). It is then natural to count features of out-neighbours in the order $<^V$. However, this allows float GNNs to distinguish isomorphic nodes, which violates the desire that GNNs should be invariant under isomorphism. 
For example, summing $1$, $-1$ and $0.01$ in two orders in a system where the numbers must be representable in the form $0.d_{1}d_{2} \times 10^{e}$ or $-0.d_{1}d_{2} \times 10^{e}$: first $(1 + (-1)) + 0.01 = 0 + 0.01 = 0.01$ while $(1 + 0.01) + (-1)  = 1 + (-1) = 0$, since $1.01$ is not representable in the system. A float GNN could distinguish two isomorphic nodes with such ordering of out-neighbours.

To ensure isomorphism invariance for GNNs with floats that use sum,
it is natural to order the floats instead of the nodes. For example, adding floats in increasing order (of the floats) is a natural and simple choice. 
Summing in this increasing order is also used widely in applications, being a reasonable choice w.r.t.\ accuracy (see, e.g., \cite{wilkinson},\cite{floats_robertazzi}, \cite{higham_article}). Generally, floating-point sums in applications have been studied widely, see for example \cite{higham_book}. 
Summing multisets of floats in increasing order leads to a bound in the multiplicities of the elements of the sum; see Proposition \ref{floating-point sum bound} for the formal statement.
Before discussing its proof, we define the \textbf{$k$-projection} of a multiset $M$ as $M_{|k}$ where $M_{|k}(x) = \min\{M(x), k\}$. 
Given a multiset $N$ of floats in float system $S$, we let $\mathrm{SUM_{S}}(N)$ denote the output of the sum $f_{1} + \dots + f_{\ell}$ where \textbf{(1)} $f_{i}$ appears $N(f_{i})$ times (i.e., its multiplicity) in the sum, \textbf{(2)}~the floats appear and are summed in increasing order and \textbf{(3)} $+$ is according to $S$. 
\begin{proposition}\label{floating-point sum bound}
For all floating-point systems $S$, there exists a $k \in \N$ such that for all multisets $M$ over floats in $S$, we have $\mathrm{SUM}_{S}(M) = \mathrm{SUM}_{S}(M_{|k})$.
\end{proposition} 
\begin{proof}
    (Sketch) See also in Appendix \ref{appendix: floating-point sum bound}.     Let $u = 0.0\cdots01 \times \beta^{e}$ and $v = 0.10\cdots0 \times \beta^{e+1}$. Now notice that for a large enough $\ell$, summing $u$ to itself $m > \ell$ times will always give~$v$.
\end{proof}

Due to Proposition \ref{floating-point sum bound}, GNNs with floats using sum in increasing order are bounded in their ability to fully count out-neighbours. 
Thus, it is natural to assume that floating-point GNNs are \textbf{bounded} GNNs, i.e., 
the aggregation function can be written as $\cM_{k}(U^{d}) \to~U^{d}$ for some \textbf{bound} $k \in \N$, 
i.e., for every multiset $M \in \cM(U^d)$, we have $\AGG(M) = \AGG(M_{|k})$, where $M_{|k}$ is the $k$-projection of $M$ (and $U^d$ is the set of states of the $\GNN$).
We finally give a formal definition for floating-point GNNs: a \textbf{floating-point graph neural network} ($\GNNF$) is simply a bounded GNN where the set of states and the domains and co-domains of the functions are restricted to some floating-point system $S$ instead of $\R$ (note that $S$ can be any floating-point system).
In R-simple $\GNNF$s, $\mathrm{SUM}_{S}$ replaces the sum of reals as the aggregation function, and their bound is thus determined by the choice of $S$.
A $\GNNF$ 
obtained by removing the condition on boundedness 
is called an \textbf{unrestricted} $\GNNF$. Note that by default and unless otherwise stated, a $\GNN[\R]$ 
is unbounded, whereas a $\GNNF$ is bounded. 
Now, it is immediately clear that unrestricted $\GNNF$s (with an unrestricted aggregation function) are more expressive than $\GNNF$s: expressing the property that a node has an even number of out-neighbours is easy with unrestricted $\GNNF$s, 
but no bounded $\GNNF$ with bound $k$ can distinguish 
the centers of two star graphs with $k$ and $k+1$ out-neighbours.

\subsection{Logics}\label{sec: logics}

We then define the logics relevant for this paper. 
For $\Pi \subseteq \mathrm{LAB}$, the set of \textbf{$\Pi$-formulae of graded modal logic} ($\GML$) is given by the grammar
\[
    \varphi \coloncolonequals \top \,\mid\, p \,\mid\, \neg \varphi \,\mid\, \varphi \land \varphi \,\mid\, \Diamond_{\geq k} \varphi,
\]
where $p \in \Pi$ and $k \in \N$. The connectives $\lor$, $\rightarrow$, $\leftrightarrow$ are considered abbreviations in the usual way. Note that node label symbols
serve as propositional symbols here. 
A $\Pi$-formula of $\GML$ is interpreted in pointed $\Pi$-labeled graphs. In
the context of modal logic, these are often called (pointed) Kripke models.
Let $G = (V, E, \lambda)$ be a $\Pi$-labeled graph and $w \in V$. The truth of a formula $\varphi$ in a pointed graph $(G, w)$ (denoted $G, w \models \varphi$) is defined  
as usual for the Boolean operators and $\top$, while for $p\in \Pi$ and $\Diamond_{\geq k} \varphi$, we define that $G,w \models p$ iff $p \in \lambda(w)$, and
\[
    G,w \models \Diamond_{\geq k} \varphi
\text{ iff $G, v \models \varphi$ for at least $k$ out-neighbours $v$ of~$w$.}
\]
We use the abbreviations $\Diamond \varphi \colonequals \Diamond_{\geq 1} \varphi$, $\Box \varphi \colonequals \neg \Diamond \neg \varphi$ and $\Diamond_{=n} \varphi \colonequals \Diamond_{\geq n} \varphi \land \neg \Diamond_{\geq n+1} \varphi$. %
The set of \textbf{$\Pi$-formulae of $\VGML$} is given by the grammar
\[
    \varphi \coloncolonequals \psi \,\mid\, \bigvee_{\psi \in \Psi} \psi,
\]
where $\psi$ is a $\Pi$-formula of $\GML$ and $\Psi$ is an at most countable set of $\Pi$-formulae of $\GML$. The truth of infinite disjunctions is defined in the obvious way:
\[
    G, w \models \bigvee_{\psi \in \Psi} \psi \iff G, w \models \psi \text{ for some $\psi \in \Psi$}.
\]

We next introduce the graded modal substitution calculus (or $\GMSC$),
which extends the modal substitution calculus \cite{Kuusisto13,dist_circ_mfcs, ahvonen_et_al:LIPIcs.CSL.2024.9} with counting capabilities.
Define the countable set $\mathrm{VAR} = \{\, V_i \mid i \in \N\,\}$ %
of \textbf{schema variable symbols}.
Let $\cT = \{X_1, \ldots, X_n\} \subseteq \mathrm{VAR}$. The set of \textbf{$(\Pi, \cT)$-schemata} of $\GMSC$ is defined by the grammar
\[
\varphi \coloncolonequals \top \,\mid\, p \,\mid\, X_{i} \,\mid\, \neg \varphi \,\mid\, \varphi \land \varphi \,\mid\, \Diamond_{\geq k} \varphi 
\]
where $p \in \Pi$, $X_{i} \in \cT$ and $k \in \N$. A \textbf{$(\Pi, \cT)$-program} $\Lambda$ of $\GMSC$ consists of two lists of expressions
\[
\begin{aligned}
    &X_1 (0) \colonminus \varphi_1\qquad &&X_1 \colonminus \psi_1 \\
    &\vdots &&\vdots \\
    &X_n (0) \colonminus \varphi_n &&X_n \colonminus \psi_n
\end{aligned}
\]
where $\varphi_1, \ldots, \varphi_n$ are $\Pi$-formulae of $\GML$ and $\psi_1, \ldots, \psi_n$ are $(\Pi, \cT)$-schemata of $\GMSC$. Moreover, each program is associated with a set $\cA \subseteq \cT$ of \textbf{appointed predicates}. A program of \textbf{modal substitution calculus} $\MSC$ is a program of $\GMSC$ that may only use diamonds $\Diamond$ of the standard modal logic. The expressions $X_i (0) \colonminus \varphi_i$ are called \textbf{terminal clauses} and $X_i \colonminus \psi_i$ are called \textbf{iteration clauses}. The schema variable $X_i$ in front of the clause is called the \textbf{head predicate} and the formula $\varphi_i$ (or schema $\psi_i$) is called the \textbf{body} of the clause. The terminal and iteration clauses are the rules of the program. When we do not need to specify $\cT$, we may refer to a $(\Pi, \cT)$-program as a \textbf{$\Pi$-program} of $\GMSC$.
Now, the \textbf{$n$th iteration formula} $X^{n}_{i}$ of a head predicate $X_{i}$ (or the iteration formula of $X_{i}$ in \textbf{round} $n \in \N$) (w.r.t. $\Lambda)$ is defined as follows.
%
%
    The $0$th iteration formula $X^0_i$ is $\varphi_i$ and
    the $(n+1)$st iteration formula $X^{n+1}_i$ is $\psi_{i}$ where each head predicate $Y$ in $\psi_{i}$ is replaced by the formula $Y^n$.
%
%
We write $G, w \models \Lambda$ and say that $\Lambda$ \textbf{accepts} $(G, w)$ iff $G, w \models X^n$ for some appointed predicate $X \in \cA$ and some $n \in \N$. 
Moreover, for all $(\Pi, \cT)$-schemata $\varphi$ that are not head predicates and for $n \in \N$, we let $\varphi^n$ denote the formula (w.r.t. $\Lambda$) where each $Y \in \cT$ in $\varphi$ is replaced by $Y^n$.

Recall that \textbf{monadic second-order logic} $\mathrm{MSO}$ is obtained as an extension of \textbf{first-order logic} $\mathrm{FO}$ by allowing quantification of unary relation variables $X$, i.e., if $\varphi$ is an $\MSO$-formula, then so are $\forall X\varphi$ and $\exists X\varphi$, see e.g. \cite{ebbinghaus1996mathematical} for more details. 
Given a set $\Pi \subseteq \mathrm{LAB}$ of node label symbols, an $\FO$- or $\MSO$-formula $\varphi$ over $\Pi$ is an $\FO$- or $\MSO$-formula over a vocabulary which contains exactly a unary predicate for each $p \in \Pi$ and the edge relation symbol $E$.
Equality is admitted.

Let $\varphi$ be an
$\VGML$-formula, $\GMSC$-schema, $\GMSC$-program, or a rule of a program. The
\textbf{modal depth} (resp. the \textbf{width}) of $\varphi$ is the maximum number of nested diamonds in $\varphi$ (resp. the maximum number $k \in \N$ that appears in a diamond $\Diamond_{\geq k}$ in $\varphi$). If an $\VGML$-formula has no maximum depth (resp., width), its modal depth (resp., width) is~$\infty$. If an $\VGML$-formula has finite modal depth (resp., width), it is \textbf{depth-bounded} (resp., \textbf{width-bounded}). The \textbf{formula depth} of a $\GML$-formula or $\GMSC$-schema is the maximum number of nested operators $\neg$, $\land$ and $\Diamond_{\geq k}$.
Given a $\Pi$-program of $\GMSC$ or a $\Pi$-formula of $\VGML$ $\varphi$ (respectively, a formula $\psi(x)$ of $\MSO$ or $\FO$ over $\Pi$, where the only free variable is the first-order variable $x$), 
we say that $\varphi$ (resp., $\psi(x)$) \textbf{expresses} a node property $\cP$ over $\Pi$, if for each pointed $\Pi$-labeled graph $(G, w)$: $(G, w) \in \cP$ iff $G, w \models \varphi$ (or resp. $G \models \psi(w)$).
A node property $\cP$ over $\Pi$ is \textbf{expressible} in $\GMSC$ (resp., in $\VGML$, $\FO$ or $\MSO$) if there is a $\Pi$-program of $\GMSC$ (resp., a $\Pi$-formula of $\VGML$, $\FO$ or $\MSO$) expressing $\cP$. 

\begin{example}
Recall the property
\emph{reachability of node label symbol $p$} over $\Pi$ defined in Example~\ref{example: GNN definable properties}. It is expressed by the $\Pi$-program of $\GMSC$
$
X(0) \colonminus p$,
$X \colonminus \Diamond X$,
where $X$ is an appointed predicate. The $i$th iteration formula of $X$ is $X^{i} = \Diamond \cdots \Diamond p$ where there are exactly $i$ diamonds.
\end{example}

\begin{example}[\cite{Kuusisto13}]\label{example: centre-point property}
    A pointed $\Pi$-labeled graph $(G,w)$ has the \emph{centre-point property} $\cP$ over $\Pi$ if there exists an $n \in \N$ such that each directed path starting from $w$ leads to a node with no out-neighbours in exactly $n$ steps. It is easy to see that the $\Pi$-program $X (0) \colonminus \Box \bot$, $X \colonminus \Diamond X \land \Box X$, where $X$ is an appointed predicate, expresses $\cP$. In \cite{Kuusisto13}, it was established that the centre-point property is not expressible in $\MSO$ and that there are properties expressible in $\mu$-calculus and $\MSO$ that are not expressible in $\MSC$ (e.g., non-reachability), with the same proofs applying to $\GMSC$. 
\end{example}

\begin{proposition}\label{proposition: GMSC < GML}
    Properties expressible in $\GMSC$ are expressible in $\VGML$, but not vice versa.
\end{proposition}
\begin{proof}
A property over $\Pi$ expressed by a $\Pi$-program $\Gamma$ of $\GMSC$ is expressible in $\VGML$ by the $\Pi$-formula $\bigvee_{X_{i} \in \cA, \, n \in \N} X_{i}^{n}$, where $X_{i}^{n}$ is the $n$th iteration formula of appointed predicate $X_{i}$ of $\Gamma$. Like $\GNN[\R]$s, $\VGML$ can express undecidable properties (cf. Remark~\ref{GNN uncomputable}). Clearly $\GMSC$-programs $\Lambda$ cannot, as configurations defined by $\Lambda$ in a finite graph eventually loop, i.e., the truth values of iteration formulae start repeating cyclically.
\end{proof}

While $\GMSC$ is related to the graded modal $\mu$-calculus ($\mu\mathrm{GML}$), which originates from \cite{DBLP:conf/cade/KupfermanSV02} and is
used in \cite{Pfluger_Tena_Cucala_Kostylev_2024} to characterize a recurrent GNN model, $\mu\mathrm{GML}$ and $\GMSC$ are orthogonal in expressivity. Iteration in $\GMSC$ need \emph{not} be over a monotone function
and does not necessarily yield a fixed point, and there are no syntactic restrictions that would, e.g., force schema variables to be used only positively as in $\mu\mathrm{GML}$.
The centre-point property from Example \ref{example: centre-point property} is a simple property not expressible in $\mu\mathrm{GML}$ (as it is not even expressible in $\MSO$, into
which $\mu\mathrm{GML}$ translates). 
Conversely, $\GMSC$ offers neither greatest fixed points nor fixed point alternation. In particular, natural properties expressible in the $\nu$-fragment of $\mu\mathrm{GML}$ such as non-reachability are
not expressible in $\GMSC$; this is proved in \cite{Kuusisto13} for the non-graded
version, and the same proof applies to $\GMSC$. 
However, the $\mu$-fragment of the graded modal $\mu$-calculus translates into $\GMSC$ (by essentially the same argument as the one justifying Proposition 7 in \cite{Kuusisto13}). 
We also note that $\GMSC$ translates into partial fixed-point logic with choice \cite{DBLP:conf/csl/Richerby04}, but it is not clear whether the same holds without
choice.

\section{Connecting GNNs and logics via automata}\label{section: GNNs = logics}

In this section we establish exact matches between classes of $\GNN$s and our logics.
The first main theorem is Theorem \ref{theorem: k-GNN[F] = k-FCMPA = GMSC}, showing that $\GNNF$s, R-simple aggregate-combine $\GNNF$s
and $\GMSC$ are equally expressive.
Theorem \ref{omega-GML = GNN = CMPAs} is the second main result, showing that $\GNN[\R]$s
and $\VGML$ are equally expressive.
We begin by defining the concept of distributed automata
which we will mainly use as a tool for our arguments but they also lead to nice additional characterizations. 
Informally, we consider a model of distributed automata called counting message passing automata and its variants which operate similarly to $\GNN$s. 
These distributed automata update the state of each node according to the node's own state and the \emph{multiset} of states of its out-neighbours.

Formally, given $\Pi \subseteq \mathrm{LAB}$, a \textbf{counting message passing automaton} $(\CMPA)$
over $\Pi$ is a tuple $(Q, \pi, \delta, F)$ where
$Q$ is an at most countable set of \emph{states} and $\pi$, $\delta$ and
$F$ are defined in a similar way as for $\GNN$s:
$\pi \colon \cP(\Pi) \to Q$ is an \emph{initialization function},
$\delta \colon Q \times \cM(Q) \to Q$ a \emph{transition function} and
$F \subseteq Q$ a set of \emph{accepting states}.  Computation of a
$\CMPA$ over~$\Pi$ is defined in a $\Pi$-labeled graph $G$ analogously to
GNNs: for each node $w$ in $G$, the initialization function gives the
initial state for $w$ based on the node label symbols true in $w$, and the
transition function is applied to the previous state of $w$ and the
multiset of states of out-neighbours of $w$.  Acceptance is
similar to $\GNN$s: the $\CMPA$ \textbf{accepts} $(G, w)$ if the
$\CMPA$ visits (at least once) an accepting state at $w$ in $G$.  A
\textbf{bounded} $\CMPA$ is a $\CMPA$ whose transition function can be
written as $\delta \colon Q \times \cM_{k}(Q) \to Q$ for some $k \in \N$ (i.e.,
$\delta(q, M) = \delta(q, M_{|k})$ for each multiset $M \in \cM(Q)$
and state $q \in Q$).
A \textbf{finite} $\CMPA$ ($\FCMPA$) is a $\CMPA$ with finite $Q$. We define bounded $\FCMPA$s similarly to bounded $\CMPA$s. 
A $\CMPA$ $\cA$ 
over $\Pi$ \textbf{expresses} a node property $\cP$ over $\Pi$ if $\cA$ accepts $(G, w)$ iff $(G, w) \in \cP$. We define whether a node property $\cP$ is expressible by a $\CMPA$ in a way analogous to $\GNN$s.

For any $\Pi$, a $\Pi$-object refers to a GNN over $\Pi$, a $\CMPA$ over $\Pi$, a $\Pi$-formula of $\VGML$ or a $\Pi$-program of $\GMSC$.
Let $\cC$ be the class of all $\Pi$-objects for all $\Pi$.
Two $\Pi$-objects in~$\cC$ are \textbf{equivalent} if they express the same node property over $\Pi$. Subsets $A, B \subseteq \cC$  \textbf{have the same expressive power}, if 
each 
$x \in A$ has an equivalent 
$y \in B$ and vice versa.
It is easy to obtain the following. 
\begin{proposition}\label{theorem: GMSC = k-FCMPA}
    Bounded $\FCMPA$s have the same expressive power as $\GMSC$.
\end{proposition}
\begin{proof}
    (Sketch) Details in Appendix \ref{appendix: GMSC = k-FCMPA}.
    To obtain a bounded $\FCMPA$ equivalent to a $\GMSC$-program $\Lambda$, we first turn $\Lambda$ into an equivalent program $\Gamma$ where the modal depth of terminal clauses (resp., iteration clauses) is $0$ (resp., at most $1$). Then from $\Gamma$ with head predicate set $\cT'$, we construct an equivalent $\FCMPA$ $\cA$ as follows. 
    The set of states of $\cA$ is $\cP(\Pi \cup \cT')$ 
    and $\cA$ enters in round $n \in \N$ in node $w$ into a state that contains precisely the node label symbols true in $w$ and the predicates $X$ whose iteration formula $X^n$ is true at $w$.
    For the converse, 
    we create a head predicate for each state in $\cA$, and let predicates for accepting states be appointed. The terminal clauses simulate $\pi$ using 
    disjunctions of conjunctions of non-negated and negated node label symbols. 
    The iteration clauses simulate $\delta$ using, for each pair $(q,q')$, a subschema specifying the multisets that take $q$ to $q'$.
\end{proof}

We are ready to show equiexpressivity of $\GMSC$
and $\GNNF$s. This applies \emph{without relativising to any background logic}.
The direction from $\GNNF$s to $\GMSC$ is trivial. The other
direction is more challenging, in particular when going all the way to
R-simple $\GNNF$s. 
While size issues were not a concern in this work, we observe that the translation from $\GNNF$s to $\GMSC$ involves only polynomial blow-up in size; the related definitions and proofs are in appendices \ref{appendix: sizes}, \ref{appendix: GMSC = k-FCMPA} and \ref{sec appendix: k-GNNF = k-FCMPA = GMSC}. We also conjecture that a polynomial translation from $\GMSC$ to R-simple $\GNNF$s is possible by 
the results and techniques
in \cite{ahvonen_et_al:LIPIcs.CSL.2024.9}, taking into account differences between $\GNNF$s and R-simple $\GNNF$s w.r.t. the definition of size. A more serious examination of blow-ups would require a case-by-case analysis taking other such details into account.

\begin{theorem}\label{theorem: k-GNN[F] = k-FCMPA = GMSC}
The following have the same expressive power: $\GNNF$s, $\GMSC$, and R-simple aggregate-combine $\GNNF$s.
\end{theorem}
%
\begin{proof}
    (Sketch) Details in Appendix \ref{sec appendix: k-GNNF = k-FCMPA = GMSC}. 
    By definition, a $\GNNF$ is just a bounded $\FCMPA$ and translates to a $\GMSC$-program by Proposition~\ref{theorem: GMSC = k-FCMPA}. 
    To construct an R-simple aggregate-combine $\GNNF$ $\cG$
    for a $\GMSC$-program $\Lambda$ with formula depth $D$, we first turn~$\Lambda$ into an
    equivalent program $\Gamma$, where each terminal clause has the
    body $\bot$, the formula depth of each body of an iteration clause is $D'$ (linear in $D$) and 
    for each subschema of $\Gamma$ that is a conjunction, both conjuncts 
    have the same formula depth if neither conjunct is $\top$. 
    We define binary feature vectors that are split into two halves: 
    the 1st half intuitively calculates the truth values of the head predicates and subschemata of $\Gamma$ one formula depth at a time in the style of Barceló et al. \cite{DBLP:conf/iclr/BarceloKM0RS20}. The 2nd half records the current formula depth under evaluation. $\cG$
    simulates one round of $\Gamma$ in $D'+1$ rounds, using the 2nd half of the features to accept nodes only every $D'+1$ rounds: the truth values of head predicates are correct in those rounds. 
    Note that the choice of floating-point system in $\cG$
    depends on $\Lambda$ and thus no single floating-point system is used by all GNNs resulting from the translation. In fact, fixing a single floating-point system would trivialize the computing model as only
    finitely many functions could be defined.
\end{proof} 

To show that $\GNN[\R]$s and $\VGML$ are equally expressive, 
we first prove a useful theorem.
\begin{theorem}\label{thrm: equi omega-GML CMPA}
    $\CMPA$s have the same expressive power as $\omega$-$\GML$.
\end{theorem} 
%
%
%
\begin{proof}
    (Sketch) Details in Appendix \ref{sec: omega-GML CMPA}.
    To construct a $\CMPA$ for each $\VGML$-formula~$\chi$, we define a $\GML$-formula called the ``\emph{full graded type of modal depth $n$}'' for each pointed graph $(G,w)$, which expresses all the local information of the neighbourhood of $w$ up to depth $n$ 
    (with the maximum out-degree of $G$ plus 1 sufficing for width). 
    We show that each $\VGML$-formula $\chi$ is equivalent to an infinite disjunction $\psi_{\chi}$ of types. We then define a CMPA that computes the type of modal depth $n$ of each node in round $n$. Its accepting states are the types in the type-disjunction  
    $\psi_{\chi}$.
    For the converse, to show that each $\CMPA$ has an equivalent $\VGML$-formula, we first show that two pointed graphs satisfying the same full graded type of modal depth $n$ have identical states in each round $\ell \leq n$ in each $\CMPA$. The $\VGML$-formula is the disjunction containing every type $T$ such that some $(G,w)$ 
    satisfying $T$ is accepted by the automaton in round $n$, where $n$ is the depth of $T$. 
\end{proof}

Next we characterize $\GNN[\R]$s with $\VGML$ \emph{without relativising to a background logic}. As the theorems above and below imply that $\GNN[\R]$s and $\CMPA$s are equally expressive, it follows that $\GNN[\R]$s can be discretized to $\CMPA$s 
having---by definition---an only countable set of states.
\begin{theorem}\label{omega-GML = GNN = CMPAs}
    $\GNN[\R]$s have the same expressive power as  $\VGML$.
\end{theorem}
%
%
\begin{proof}
    (Sketch) Details in Appendix \ref{sec appendix: omega-GML = GNN = GMPAs}.
    For any $\GNN[\R]$, we build an equivalent $\VGML$-formula using the same method as in the proof of Theorem \ref{thrm: equi omega-GML CMPA}, where we show that for each $\CMPA$, we can find an equivalent $\VGML$-formula.
    For the converse, we first translate an $\VGML$-formula to a $\CMPA$ $\cA$ by Theorem \ref{thrm: equi omega-GML CMPA} and then translate $\cA$ to an equivalent $\CMPA$ $\cA'$ with maximal ability to distinguish nodes. Then we build an equivalent $\GNN[\R]$ for $\cA'$ by encoding states into integers. 
\end{proof}

\begin{remark}\label{remark: additional results}
    It is easy to show that unrestricted $\GNNF$s have the same
    expressive power as $\FCMPA$s.
    Moreover, the proof of Theorem \ref{omega-GML = GNN = CMPAs} is easily modified to show that bounded $\GNN[\R]$s have the same expressive power as width-bounded $\VGML$. See Appendix \ref{appendix: remark} for the proofs.
\end{remark}

In Appendix \ref{appendix: constant iteration equivalence} we show that our model of constant-iteration $\GNN[\R]$s is equivalent to the one in Barceló et al. \cite{DBLP:conf/iclr/BarceloKM0RS20}.
Thus Theorem \ref{constant iteration GNN reals} (proven in Appendix \ref{appendix: constant-iteration GNNs to depth-bounded VGML}) generalizes the result in
\cite{DBLP:conf/iclr/BarceloKM0RS20} saying that any property $\cP$
expressible by $\FO$ is expressible as a constant-iteration
$\GNN[\R]$ iff it is expressible in $\GML$. Furthermore, any such $\cP$ is expressible in $\GML$ iff it is expressible in $\VGML$.

\begin{theorem}\label{constant iteration GNN reals}
    Constant-iteration $\GNN[\R]$s have the same expressive power as depth-bounded $\VGML$.
\end{theorem}

\section{Characterizing GNNs over MSO-expressible properties}
\label{sect:MSO}

In this section we consider properties expressible in $\MSO$. 
The first main result is Theorem \ref{MSO GNN reals and GMSC}, where we show that for properties expressible in $\MSO$, the expressive power of $\GNN[\R]$s is captured
by a \emph{finitary} logic. In fact, this logic is $\GMSC$ and by Theorem~\ref{theorem: k-GNN[F] = k-FCMPA = GMSC}, it follows that relative to $\MSO$, $\GNN[\R]$s have the same expressive power as $\GNNF$s (Theorem \ref{theorem: under MSO: GNN[R] = GNN[F]} below). 
Our arguments in this section work uniformly for any finite set $\Sigma_{N}$ of node label symbols.

\begin{theorem}\label{MSO GNN reals and GMSC}
     Let $\cP$ be a property expressible in $\MSO$. Then $\cP$ is expressible
     as a $\GNN[\R]$ if and only if it is expressible in $\GMSC$.
\end{theorem}
Theorem~\ref{MSO GNN reals and GMSC} is proved later in this section.
The recent work \cite{Pfluger_Tena_Cucala_Kostylev_2024} relates to Theorem \ref{MSO GNN reals and GMSC}; see the introduction for a
discussion.   
The proof of Theorem~\ref{MSO GNN reals and GMSC} can easily be adapted to show that a property $\cP$ expressible in $\MSO$ is expressible as a constant-iteration $\GNN[\R]$ iff it is expressible in $\GML$.
This relates to \cite{DBLP:conf/iclr/BarceloKM0RS20} which shows the same for $\FO$ in place of $\MSO$. 
However, in 
this particular case the $\MSO$ version is not an actual generalization as based on \cite{DBLP:journals/tocl/ElberfeldGT16}, we  
show the following in  Appendix~\ref{appendix:MSOplusGMLequalsFO}.  
\begin{restatable}{lemma}{lemMSOplusGMLequalsFO}
\label{lem:MSOplusGMLequalsFO}
  Any property expressible in $\MSO$ and as a constant-iteration $\GNN[\R]$ is also $\mathrm{FO}$-expressible.  
\end{restatable}

Uniting Theorems~\ref{MSO GNN reals and GMSC} and~\ref{theorem: k-GNN[F] = k-FCMPA = GMSC},
we see that (recurrent) $\GNN[\R]$s and $\GNNF$s coincide
relative to $\MSO$. 
It is easy to get a similar result for constant-iteration $\GNN$s (the details are in Appendix \ref{appendix: under MSO: GNN[R] = GNN[F]}).

\begin{theorem}\label{theorem: under MSO: GNN[R] = GNN[F]}
    Let $\cP$ be a property expressible in $\MSO$. Then $\cP$ is expressible
    as a $\GNN[\R]$ if and only if it is expressible as a $\GNNF$.
    The same is true for constant-iteration $\GNN$s.
\end{theorem}
To put Theorem~\ref{theorem: under MSO: GNN[R] = GNN[F]} into perspective, we note that by Example~\ref{example: centre-point property}, the centre-point property is expressible in $\GMSC$, but not in $\MSO$. From Theorem~\ref{theorem: k-GNN[F] = k-FCMPA = GMSC}, we thus obtain the
following.
\begin{proposition}
    There is a property $\cP$ that is expressible as a $\GNNF$ but not in $\MSO$.
\end{proposition}

Theorem \ref{constant iteration GNN reals} shows that
constant-iteration $\GNN[\R]$s and depth-bounded $\VGML$ have the same expressive power. The proof of Proposition \ref{proposition: GMSC <
  GML} shows that already depth-bounded $\VGML$ can express properties
that $\GMSC$ cannot, in particular undecidable ones. Thus, by Theorem~\ref{theorem: k-GNN[F] = k-FCMPA = GMSC} we obtain the following.
\begin{proposition}
    There is a property expressible as a constant-iteration $\GNN[\R]$
    but not as a $\GNNF$.
\end{proposition}
\newcommand{\Amc}{\ensuremath{\mathcal{A}}\xspace}
\newcommand{\Fmc}{\ensuremath{\mathcal{F}}\xspace}
\newcommand{\Qmc}{\ensuremath{\mathcal{Q}}\xspace}
\newcommand{\Smc}{\ensuremath{\mathcal{S}}\xspace}
\newcommand{\mn}[1]{\ensuremath{\mathsf{#1}}}

We next discuss the proof of Theorem~\ref{MSO GNN reals and GMSC}.
We build upon results due to Janin and 
Walukiewicz \cite{DBLP:conf/stacs/Walukiewicz96,DBLP:conf/concur/JaninW96}, reusing an
automaton model from~\cite{DBLP:conf/stacs/Walukiewicz96} that
captures the expressivity of $\MSO$
on tree-shaped graphs. With a tree-shaped graph we mean a \emph{potentially infinite} graph that is a directed tree.
The out-degree of nodes is unrestricted (but nevertheless finite), different nodes may have different degree, and both leaves and infinite paths are
admitted in the same tree. 

We next introduce the mentioned automaton model. Although we are only going to use it on tree-shaped $\Sigma_{N}$-labeled graphs, in its full generality it is actually defined on unrestricted such graphs.
We nevertheless
call them \textbf{tree automata} as they belong to the tradition of more
classical forms of such automata. In particular, a run of an automaton
is tree-shaped, even if the input graph is not.
Formally, a \textbf{parity tree automaton (PTA)} is a tuple $\Amc
= (Q,\Sigma_N,q_0, \Delta, \Omega)$, where $Q$ is a finite
set of \textbf{states}, $\Sigma_N \subseteq \mathrm{LAB}$ is a finite
set of node label symbols, 
$q_0 \in Q$ is an \textbf{initial state},
  $
    \Delta \colon Q \times \cP(\Sigma_N) \to \Fmc
  $
  is a transition relation with \Fmc being the set of all \textbf{transition
  formulas} for \Amc defined below, and $\Omega:Q
  \rightarrow \mathbb{N}$ is a \textbf{priority function}.
  A transition formula for \Amc is a disjunction
  of FO-formulas of the form 
  \[
  \begin{array}{rl}
    \exists x_1 \cdots \exists x_k \, \big ( \mn{diff}(x_1,\dots,x_k)
    \wedge q_1(x_1) \wedge \cdots \wedge q_k(x_k) \wedge 
    \forall z (\, \mn{diff}(z,x_1,\dots,x_k) \rightarrow \psi) \big)
  \end{array} 
  \]
  where $k \geq 0$, $\mn{diff}(y_1,\dots,y_n)$ shortens an $\FO$-formula declaring $y_1,\dots,y_n$ as 
  pairwise distinct, $q_i \in Q$ are states used as unary predicates and 
  $\psi$ is a disjunction of conjunctions of atoms  $q(z)$, with $q \in Q$.
A PTA \Amc accepts a language $L(\Amc)$ consisting of (possibly infinite) graphs $G$. We have $G \in L(\Amc)$ if there is an accepting \textbf{run} of \Amc on $G$, and runs are defined in the spirit of alternating automata. While details are in Appendix \ref{appendix: MSO GNN reals and GMSC},
we mention that transition formulas govern transitions in the run: a
transition of a PTA currently visiting node $v$ in state $q$
consists of sending
copies of itself to out-neighbours of $v$, potentially in states other than $q$. 
It is not required that a copy is sent to every out-neighbour, and multiple copies (in different states) can be sent to the same out-neighbour. 
However, we must find some $(q,\lambda(v),\vartheta) \in \Delta$ such that the transition satisfies $\vartheta$ in the sense that $\vartheta$ is satisfied in the graph with one element for each out-neighbour of $v$ and unary predicates
(states of~\Amc) are interpreted according to the transition. This specific form of PTAs is interesting to us due to the following.

\begin{theorem}
\label{theorem: msotoapta}
Let $\cP$ be a property expressible in $\MSO$ and in $\VGML$. Then there is a PTA $\Amc$ such that for any
 graph $G$:\,  
 $(G, w) \in \cP$ iff
the unraveling of $G$ at $w$ is in $L(\Amc)$.
\end{theorem}
%
%
\begin{proof}
  Let $\varphi(x)$ be the $\MSO$-formula over $\Sigma_{N}$ that expresses $\cP$.
  Theorem~9 in \cite{DBLP:conf/concur/JaninW96} states that for every $\MSO$-sentence $\psi$,
  there is a PTA \Amc such that for every
  tree-shaped $\Sigma_{N}$-labeled graph $G$ we have that $G \in L(\Amc)$ iff $G \models \psi$. 
  To obtain a PTA for $\varphi(x)$, we start from the $\MSO$-sentence $\psi \colonequals \exists x \, (\varphi(x) \wedge \neg \exists y E(y,x))$
  and build the corresponding PTA \Amc. Then \Amc is as desired. In fact, $(G,w) \in \cP$ iff $G \models \varphi(w)$
  iff   $U \models \varphi(w)$ with $U$ the unraveling of $G$ at $w$ since $\cP$ is
  expressible by an $\VGML$-formula which is invariant under unraveling (defined in Appendix \ref{appendix:MSOplusGMLequalsFO}). Moreover,  $U \models \varphi(w)$ iff
   $U \models \psi$ iff $U \in L(\Amc)$. 
\end{proof}
Since $\GMSC$-programs are invariant under unraveling,
we may now prove Theorem~\ref{MSO GNN reals and GMSC} by considering PTAs obtained by Theorem \ref{theorem: msotoapta}
and constructing a 
$\Sigma_{N}$-program $\Lambda$ of $\GMSC$ for each such PTA $\Amc= (Q,\Sigma_N,q_0, \Delta, \Omega)$ so that the following holds: 
For every tree-shaped $\Sigma_N$-labeled graph $G$ with root~$w$, we have $G \in L(\Amc)$ iff $G, w \models \Lambda$. We then say that \Amc and~$\Lambda$ are \textbf{tree equivalent}.

For a state $q \in Q$, let $\Amc_q$ be the PTA defined
like \Amc but with $q$ as its 
initial 
state. For a tree-shaped
graph $T$, set $Q_T= \{ q \in Q \mid T \in L(\Amc_q) \}$.
We say that $\Qmc \subseteq \cP(Q)$ is \textbf{the universal set for} $P\subseteq \Sigma_N$  
 if $\Qmc$ consists precisely of the sets $Q_T$, for all
    tree-shaped graphs~$T$, whose root
        is labeled  exactly with the node label symbols in $P$.
Let $T=(V,E,\lambda)$ be a tree-shaped graph with root $w$ and 
 $k \in \N$, and let $V_k$ be the restriction of $V$ to the nodes on level at most $k$ (the root being on level~0). A \textbf{$k$-prefix decoration} of $T$ 
is a mapping $\mu:V_k \rightarrow \cP(\cP(Q))$ such that the following conditions hold: 
\begin{enumerate}
\item 
for each $S \in \mu(w)$: $q_0 \in S$;
\item 
for all $v \in V$ on the level $k$, $\mu(v)$ is  the universal set for $\lambda(v)$;
\item 
for each $v \in V$ on some level smaller than $k$
that has out-neighbours $u_1,\dots,u_n$ and all
$S_1 \in \mu(u_1),\dots,S_n \in \mu(u_n)$, $\mu(v)$ 
contains the set $S$ that 
contains a state $q \in Q$ iff 
we have $\Delta(q, \lambda(v)) = \vartheta$ such that $\vartheta$ is satisfied
in the following graph:
  \begin{itemize}
  \item 
the universe is $\{u_1,\dots, u_n\}$ and
  \item 
each unary predicate $q' \in Q$ is interpreted as
  $\{ u_i \mid q' \in S_i \}.$
  \end{itemize}
\end{enumerate}
Intuitively, a $k$-prefix decoration of $T$ represents a set of accepting runs of
\Amc on the prefix $T_k$ of $T$. As these runs start in universal sets, for each extension of $T_k$ obtained by attaching trees to nodes on level $k$, we can
find a run among the represented ones that can be extended to an
accepting run of \Amc on that extension. 
In fact $T$ is such an extension. 
The following crucial lemma is 
proved in Appendix \ref{appendix: MSO GNN reals and GMSC}.

\begin{lemma}
\label{lem:deco}
  For every tree-shaped $\Sigma_N$-labeled graph $T$:
  $T \in L(\Amc)$  if and only if there is a $k$-prefix decoration of $T$, for some $k \in \N$.
\end{lemma}

Using the above, we sketch the proof of Theorem \ref{MSO GNN reals and GMSC};
the full proof is in Appendix \ref{appendix: MSO GNN reals and GMSC}. 
\begin{proof}[Proof of Theorem \ref{MSO GNN reals and GMSC}]
By Lemma \ref{lem:deco}, given a PTA \Amc obtained from Lemma \ref{theorem: msotoapta}, we get a tree equivalent $\GMSC$-program $\Lambda$ by building $\Lambda$ to accept the root of a tree-shaped graph $G$ iff $G$ has a $k$-prefix decoration for some $k$. This requires care, but is possible; the details are in Appendix \ref{appendix: MSO GNN reals and GMSC}. 
A crucial part of the proof is transferring the set of transition formulas of \Amc into the rules of $\Lambda$.
\end{proof}

\section{Conclusion}

We have characterized the expressivity of recurrent graph neural networks with floats and reals in terms of modal logics, both in general and relative to $\MSO$. 
In the general scenario, the characterization with floats was done via a finitary logic $\GMSC$ and the one with reals via an infinitary logic consisting of infinite disjunctions of $\GML$-formulae. 
In restriction to $\MSO$, GNNs with floats and reals have the same expressive power and are both characterized via $\GMSC$.
In the above results, we have not yet discussed global readouts and other acceptance conditions; these are discussed in Appendix~\ref{appendix: On accepting}, where we show that the characterization for float GNNs can be generalized for basically any acceptance condition for GNNs, so long as the acceptance condition of the logic is adjusted analogously. In particular, the characterization extends for the acceptance conditions based on node-specific fixed-points and graph size as also explored in \cite{Pfluger_Tena_Cucala_Kostylev_2024}, and another condition---similar to global fixed-points---from the inaugural papers \cite{scarcelli, gori}, where feature vectors must converge in every node. We also show that the characterization extends for global readouts by adding a counting global modality to the logic.

\textbf{Limitations:} 
For our acceptance condition (and for many others such as those based on fixed
points), training and applying a recurrent GNN brings questions of termination.
These are very important from a practical perspective, but not studied in this
paper. A particularly interesting question is whether termination can be 
learned in a natural way during the training phase.

\section*{Acknowledgements}

The list of authors on the first page is given based on the alphabetical order.
Veeti Ahvonen was supported by the Vilho, Yrjö and Kalle Väisälä Foundation of the Finnish Academy of Science and Letters. 
Damian Heiman was supported by the Magnus Ehrnrooth Foundation. 
Antti Kuusisto was supported by the Research Council of Finland consortium project Explaining AI via Logic (XAILOG), grant number 345612, and the Research Council of Finland project Theory of computational logics, grant numbers 352419, 352420, 353027, 324435, 328987; Damian Heiman was also supported by the same project, grant number 353027. 
Carsten Lutz was supported by the
Federal Ministry of Education and Research of Germany (BMBF)
and by S\"achsisches Staatsministerium f\"ur Wissenschaft,
Kultur und Tourismus in the program Center of Excellence
for AI-research ``Center for Scalable Data Analytics and Artificial Intelligence Dresden/Leipzig'', project identification
number: ScaDS.AI. Lutz is also supported by BMBF in DAAD project 57616814
(SECAI, School of Embedded Composite AI) as part of the
program Konrad Zuse Schools of Excellence in Artificial Intelligence.

\bibliography{literature}

\newpage
\appendix

\section{Appendix: Preliminaries}\label{appendix: preliminaries}

\subsection{Notions of size}\label{appendix: sizes}

We start this appendix by presenting definitions for the sizes of $\GMSC$-programs, $\GNNF$s and bounded $\FCMPA$s.

The \textbf{size} of a $\Pi$-program of $\GMSC$ is here defined as the size of $\Pi$ plus the number of occurrences of node label symbols, head predicates and logical operators in the program, with each $\Diamond_{\geq k}$ adding $k$ to the sum instead of just $1$.

The \textbf{size} of a $\GNNF$ or bounded $\FCMPA$ over $\Pi$ is defined here as follows. For $\FCMPA$s define $U \colonequals Q$, and for $\GNNF$s define $U \colonequals S^{d}$ where $S$ is the floating-point system and $d$ the dimension of the $\GNNF$. In both cases, let $n \colonequals \abs{U}$ and let $k$ be the bound. Note that the cardinality of a multiset $\cM_k(U)$ is the number $(k+1)^{n}$ of different functions of type $U \to [0;k]$. Then the size of the $\GNNF$ or bounded $\FCMPA$ is $\abs{\cP(\Pi)}$ plus $n \cdot (k+1)^{n}$, i.e., the sum of the cardinality $\abs{\cP(\Pi)}$ of its initialization function $\cP(\Pi) \to U$ and the cardinality $n \cdot (k+1)^{n}$ of its transition function $U \times \cM_{k}(U) \to U$ as look-up tables.

We note that there is no single way to define the sizes of the objects above. 
For example, counting each node label symbol, head predicate or state into the size just once may be naïve, as there may not be enough distinct symbols in practice. To account for this, we could also define that the size of each node label symbol, head predicate and state is instead $\log(k)$, where $k$ is the number of node label symbols, head predicates or states respectively. Likewise, the size of a diamond $\Diamond_{\geq \ell}$ could be considered $\log(k)$, where $k$ is the maximum $\ell$ appearing in a diamond in the program. However, this happens to not affect the size of our translation from $\GNNF$s to $\GMSC$.
As another example, defining the size of a single floating-point number as one may be naïve; one could define that the size of a single floating-point number is the length of the string that represents it in the floating-point system instead of one.

\subsection{Proof of Proposition \ref{floating-point sum bound}}\label{appendix: floating-point sum bound}

We recall Proposition \ref{floating-point sum bound} and give a more detailed proof.

\textbf{Proposition \ref{floating-point sum bound}.}
\emph{For all floating-point systems $S$, there exists a $k \in \N$ such that for all multisets $M$ over floats in $S$, we have $\mathrm{SUM}_{S}(M) = \mathrm{SUM}_{S}(M_{|k})$.}
\begin{proof}
    Consider a floating-point system $S = ((p, n, \beta), +, \cdot)$ and a multiset $M \in \cM(S)$.
    We assume $f = 0.d_{1} \cdots d_{p} \times \beta^{e} \in S$ (the case where $f = -0.d_{1} \cdots d_{p} \times \beta^{e}$ is symmetric). Assume that the sum of all numbers smaller than $f$ in $M$ in increasing order is $s \in S$. Let $s + f^{0}$ denote $s$ and let $s+f^{\ell+1}$ denote $(s + f^{\ell}) + f$ for all $\ell \in \N$. It is clear that $s+f^{\ell+1} \geq s+f^{\ell}$, and furthermore, if $s+f^{\ell+1} = s+f^{\ell}$, then $s+f^{\ell'} = s+f^{\ell}$ for all $\ell' \geq \ell$. If $s+f^{\ell+1} = s+f^{\ell}$, then we say that the sum $s + f^{\ell}$ has \emph{stabilized}.

    Let $\beta' \colonequals \beta - 1$. We may assume that $s \geq -0.\beta' \cdots \beta' \times \beta^{n}$. Since $s+f^{\ell+1} \geq s+f^{\ell}$ for each $\ell \in \N$ and $S$ is finite, there must exist some $k \in \Z_{+}$ such that either $s+f^{k} = s+f^{k-1}$ or $s+f^{k} = 0.\beta' \cdots \beta' \times \beta^{n}$ but then $s + f^{k+1} = s + f^{k}$. Clearly, $k = \abs{S}$ is sufficient for each $s$ and $f$, and thus satisfies the condition of the proposition.

    A smaller bound can be found as follows.
    Assume without loss of generality that $d_{1} \neq 0$ (almost all floating-point numbers in $S$ can be represented in this form and the representation does not affect the outcome of a sum; we consider the case for floats not representable in this way separately). 
    Now, consider the case where $s = -0.\beta' \cdots \beta' \times \beta^{e'}$ for some $e' \in [-n, n]$, as some choice of $e'$ clearly maximizes the number of times $f$ can be added. 
    Now it is quite easy to see that that $s + f^{\beta^{p} - \beta^{p-1}} \geq -0.10 \cdots 0 \times \beta^{e'}$, unless the sum stabilizes sooner. This is because there are exactly $\beta^{p} - \beta^{p-1}$ numbers in $S$ from $-0.\beta' \cdots \beta' \times \beta^{e'}$ to $-0.10 \cdots 0 \times \beta^{e'}$. After this another $2 \beta^{p-1}$ additions ensures that $s + f^{\beta^{p} + \beta^{p-1}} \geq 0$. (Consider the case where $p = 3$ and $\beta = 10$ and we add the number $f = 0.501 \times 10^{0}$ to $-0.100 \times 10^{3}$ repeatedly. We first get $-0.995 \times 10^{2}$, then $-0.990 \times 10^{2}$, etc. Each two additions of $f$ is more than adding $f' = 0.001 \times 10^{3}$ once, and $f'$ can be added exactly $\beta^{p-1}$ times to $-0.100 \times 10^{3}$ until the sum surpasses $0$, meaning that $f$ can be added at most $2 \beta^{p-1}$ times. On the other hand, if $f$ were $0.500 \times 10^{0}$ or smaller, then the sum would have never reached $-0.100 \times 10^{3}$ in the first place.)

    Let $\ell \colonequals \beta^{p} + \beta^{p-1}$.
    Let $s' \colonequals s+f^{\ell}$ and assume the sum $s'$ has not yet stabilized.
    Next, it is clear that for any $i \in [p]$ we have $s' + f^{\beta^{i}} \geq 0.10 \cdots 0 \times \beta^{e+i}$ because $d_{1} \geq 1$. 
    For $d_{1} \leq \frac{\beta}{2}$, we have that $s' + f^{\beta^{p}} = 0.10 \cdots 0 \times \beta^{e+p}$ which is where the sum stabilizes. 
    If on the other hand $d_{1} > \frac{\beta}{2}$, we have that $s' + f^{\beta^{p+1}} = 0.10 \cdots 0 \times \beta^{e+p+1}$, which is where the sum stabilizes (this is because $d_{1} > \frac{\beta}{2}$ causes the sum to round upwards until the exponent reaches the value $e+p+1$).
    Thus, we get a threshold of $\beta^{p+1} + \beta^{p} + \beta^{p-1}$.
    
    The case where $d_{1} = 0$ for all representations of $f$ is analogous; we simply choose one representation and shift our focus to the first $i$ such that $d_{i} \neq 0$ and perform the same examination. The position of $i$ does not affect the overall analysis.

    We conjecture that a smaller bound than $\beta^{p+1} + \beta^{p} + \beta^{p-1}$ is possible, but involves a closer analysis.
\end{proof}

\section{Appendix: Connecting GNNs and logics via automata}

\subsection{An extended definition for distributed automata}\label{appendix: An extended definition for distributed automata}

We give a more detailed definition for distributed automata.
Given $\Pi \subseteq \mathrm{LAB}$, a \textbf{counting message passing automaton} (or $\mathrm{CMPA}$) over $\Pi$ is a tuple $(Q, \pi, \delta, F)$, where
\begin{itemize}
    \item 
    $Q$ is an at most countable non-empty set of \textbf{states},
    \item 
    $\pi \colon \cP(\Pi) \to Q$ is an \textbf{initialization function},
    \item 
    $\delta \colon Q \times \cM(Q) \to Q$ is a \textbf{transition function}, and
    \item 
    $F \subseteq Q$ is a set of \textbf{accepting states}.
\end{itemize}
If the set of states is finite, then we say that the counting message passing automaton is finite and call it an $\mathrm{FCMPA}$. Note that in the context of $\FCMPA$s, finiteness refers specifically to the number of states. We also define a subclass of counting message passing automata: a \textbf{$k$-counting message passing automaton} ($k$-$\CMPA$) is a tuple $(Q, \pi, \delta, F)$, where $Q$, $\pi$ and $F$ are defined analogously to $\CMPA$s and the transition function $\delta$ can be written in the form $Q \times \cM_k(Q) \to Q$, i.e., for all multisets $M \in \cM(Q)$ we have $\delta(q,M) = \delta(q,M_{|k})$ for all $q \in Q$. A $k$-$\FCMPA$ naturally refers to a $k$-$\CMPA$ whose set of states is finite; this is truly a finite automaton. For any $k \in \N$, each $k$-$\CMPA$ (respectively, each $k$-$\FCMPA$) is called a \textbf{bounded} $\CMPA$ (resp., a \textbf{bounded} $\FCMPA$).

Let $G = (V, E, \lambda)$ be a $\Pi$-labeled graph. 
We define the computation of a $\CMPA$ formally. A $\mathrm{CMPA}$ (or resp. $k$-$\CMPA$) $(Q, \pi, \delta, F)$ over $\Pi$ and a $\Pi$-labeled graph $G = (V, E, \lambda)$ define a distributed system, which executes in $\omega$-rounds as follows. Each \textbf{round} $n \in \N$ defines a \textbf{global configuration} $g_n \colon V \to Q$ which essentially tells in which state each node is in round $n$. If $n = 0$ and $w \in V$, then $g_0(w) = \pi(\lambda(w))$. Now assume that we have defined $g_n$. Informally, a new state for $w$ in round $n+1$ is computed by $\delta$ based on the previous state of $w$ and the multiset of previous states of its immediate out-neighbours. More formally, let $w$ be a node and $v_1, \ldots, v_{m}$ its immediate out-neighbours (i.e., the nodes $v \in V$ s.t. $(w, v) \in E$). Let $S$ denote the corresponding multiset 
of the states $g_n(v_1), \ldots, g_n(v_{m})$. 
We define $g_{n+1}(w) = \delta(g_n(w), S)$. We say that $g_n(w)$ is the \textbf{state of $w$ at round $n$}.

A $\mathrm{CMPA}$ (or resp. $k$-$\CMPA$) \textbf{accepts} a pointed graph $(G, w)$ if and only if $g_n(w) \in F$ for some $n \in \N$. If $g_{n}(w) \in F$, 
then we say that the $\CMPA$ \textbf{accepts $(G, w)$ in round $n$}. The computation for $\FCMPA$s is defined analogously (since each $\FCMPA$ is a $\CMPA$).

\subsection{Proof of Proposition \ref{theorem: GMSC = k-FCMPA}}\label{appendix: GMSC = k-FCMPA}

Note that extended definitions for automata are in Appendix~\ref{appendix: An extended definition for distributed automata}.

First we recall Proposition \ref{theorem: GMSC = k-FCMPA}.

\textbf{Proposition \ref{theorem: GMSC = k-FCMPA}.}
\emph{Bounded $\FCMPA$s have the same expressive power as $\GMSC$.}

To prove Proposition \ref{theorem: GMSC = k-FCMPA} (in the end of this subsection), we first prove an auxiliary Lemma \ref{lemma: GMSC to GMSC[1]} that informally shows that we can translate any $\GMSC$-program into an equivalent program, where the modal depth of the terminal clauses is zero and the modal depth of the iteration clauses is at most one. Then in Lemma \ref{lemma: GMSC to k-FCMPA} we show with the help of Lemma \ref{lemma: GMSC to GMSC[1]} that for each $\GMSC$-program we can construct an equivalent bounded $\FCMPA$. Finally, we show in Lemma \ref{lemma: k-FCMPA to GMSC} that for each bounded $\FCMPA$ we can construct an equivalent $\GMSC$-program.

We start with an auxiliary result.
\begin{lemma}\label{lemma: GMSC to GMSC[1]}
    For every $\Pi$-program of $\GMSC$, we can construct an equivalent $\Pi$-program $\Gamma$ of $\GMSC$ such that the modal depth of each terminal clause of $\Gamma$ is $0$, and the modal depth of each iteration clause of $\Gamma$ is at most $1$.
\end{lemma}
\begin{proof}
    The proof is similar to the proof of Theorem $11$ in \cite{dist_circ_mfcs} (or the proof of Theorem 5.4 in \cite{ahvonen2023descriptive}). We consider an example and conclude that the strategy mentioned in the cited papers can be generalized to our framework.

    For simplicity, in the example below we consider a program where the modal depth of each terminal clause is zero, since it is easy to translate any program of $\GMSC$ into an equivalent $\GMSC$-program, where the modal depth of the terminal clauses is zero. Then from the program of $\GMSC$, where the modal depth of each terminal clause is zero, it is easy to obtain an equivalent $\GMSC$-program where the modal depth of each terminal clause is zero and the modal depth of each iteration clause is at most one.
    
    Consider a $\{q\}$-program $\Lambda$, with the rules $X(0) \colonminus \bot$, $X \colonminus \Diamond_{\geq 3}( \neg \Diamond_{\geq 2} \Diamond_{\geq 1} X \land \Diamond_{\geq 3} q)$. The modal depth of the iteration clause of $\Lambda$ is $3$. The program $\Lambda$ can be translated into an equivalent program of $\GMSC$ where the modal depth of the iteration clauses is at most one as follows. 
    First, we define the following subprogram called a ``clock'':
    \[
    \begin{aligned}
        &T_1 (0) \colonminus \top \quad&&T_1 \colonminus T_3 \\
        &T_2 (0) \colonminus \bot \quad&&T_2 \colonminus T_1 \\
        &T_3 (0) \colonminus \bot \quad&&T_3 \colonminus T_2. \\
    \end{aligned}
    \]
    We split the evaluation of each subschema between corresponding head predicates $X_{1,1}$, $X_{1,2}$, $X_2$ and $X_3$ and define their terminal clauses and iteration clauses as follows. The body for each terminal clause is $\bot$ and the iteration clauses are defined by
    \begin{itemize}
        \item $X_{1,1} \colonminus (T_1 \land \Diamond_{\geq 1} X_3) \lor (\neg T_1 \land X_{1,1})$,
        \item $X_{1,2} \colonminus (T_1 \land \Diamond_{\geq 3} q) \lor (\neg T_1 \land X_{1,2})$,
        \item $X_{2} \colonminus (T_2 \land \Diamond_{\geq 2} X_{1,1}) \lor (\neg T_2 \land X_{2})$,
        \item $X_3 \colonminus (T_3 \land \Diamond_{\geq 3}(\neg X_{2} \land X_{1,1})) \lor (\neg T_3 \land X_3)$.
    \end{itemize}
    The appointed predicate of the program is $X_3$.

    Let us analyze how the obtained program works. The program works in a periodic fashion: a single iteration round of $\Lambda$ is simulated in a $3$ step period by the obtained program. The clock (i.e. the head predicates $T_1$, $T_2$ and $T_3$) makes sure that each level of the modal depth is evaluated once during each period in the correct order.
    For example, if $T_2$ is true, then the truth of $X_2$ depends on the truth of $\Diamond_{\geq 2} X_{1,1}$ and when $T_2$ is false, then the truth of $X_2$ stays the same. The head predicate $X_3$ essentially simulates the truth of $X$ in $\Lambda$ in the last round of each period. It is easy to show by induction on $n \in \N$ that for all pointed $\{q\}$-labeled graphs $(G, w)$ and for all $k \in [0;2]$: 
    $G, w \models X^n$ iff $G, w \models X_3^{n3 + k}$.
\end{proof}

With Lemma \ref{lemma: GMSC to GMSC[1]}, it is easy to construct an equivalent bounded $\FCMPA$ for any $\GMSC$-program.
\begin{lemma}\label{lemma: GMSC to k-FCMPA}
    For each $\Pi$-program of $\GMSC$, we can construct an equivalent bounded $\FCMPA$ over $\Pi$. 
\end{lemma}

\begin{proof}
    Let $\Lambda$ be a $(\Pi, \cT)$-program of $\GMSC$. By Lemma \ref{lemma: GMSC to GMSC[1]} we obtain an equivalent $(\Pi, \cT')$-program $\Gamma$ where the modal depth is $0$ for each terminal clause and at most $1$ for each iteration clause.  Let $k$ be the width of $\Gamma$. We create an equivalent $k$-$\FCMPA$ $\cA_{\Gamma}$ as follows.
    
    We create a state for each subset $q \subseteq \Pi \cup \cT'$; these form the state set $Q = \cP(\Pi \cup \cT')$.
    
    We formulate the initialization function $\pi$ as follows: $\pi(P) = q$ if $\Pi \cap q = P$ and for each head predicate $X \in \cT'$ (with the terminal clause $X(0) \colonminus \varphi$) we have that $X \in q$ if and only if $\varphi$ is true when exactly the node label symbols in $P$ are true. (Here it is vital that the terminal clauses have modal depth $0$.)

    Before we define the transition function, we define an auxiliary relation $\Vdash$ for schemata with modal depth at most $1$. 
    Let $N \in \cM(Q)$ be a multiset of states and let $q$ be a state. 
    Given a $(\Pi, \cT)$-schema $\psi$ with modal depth at most $1$, we define the relation $(q, N) \Vdash \psi$ recursively as follows: $(q, N) \Vdash \top$ always, 
    $(q, N) \Vdash p$ iff $p \in q$, $(q, N) \Vdash X$ iff $X \in q$,
    $(q, N) \Vdash \neg \varphi$ iff $(q, N) \not\Vdash \varphi$, $(q, N) \Vdash \varphi \land \theta$ iff $(q, N) \Vdash \varphi$ and $(q, N) \Vdash \theta$, and $(q, N) \Vdash \Diamond_{\geq \ell} \varphi$ iff there is a set $Q' \subseteq Q$ of states such that $\sum_{q' \in Q'} N(q') \geq \ell$ and $(q', \emptyset) \Vdash \varphi$ for all $q' \in Q'$ (note that $\varphi$ has modal depth $0$).

    Now we shall formulate the transition function $\delta$ as follows. Let $N \in \cM(Q)$ be a multiset of states and let $q$ be a state. 
    For each node label symbol $p \in \Pi$, $p \in \delta(q, N)$ iff $(q, N) \Vdash p$ (i.e. $p \in q$).
    For each head predicate $X \in \cT'$ with the iteration clause $X \colonminus \psi$ we have that $X \in \delta(q, N)$ if and only if $(q, N) \Vdash \psi$. (Here it is crucial that the iteration clauses have modal depth at most $1$.)

    The set of accepting states is defined as follows. If $X \in q$ for some appointed predicate $X$ of $\Gamma$ then $q \in F$; otherwise $q \notin F$.

    If $G$ is $\Pi$-labeled graph and $v$ is a node in $G$, then we let $v^n$ denote the state of $\cA_{\Gamma}$ at $v$ in round $n$. Moreover, if $N$ is the set of out-neighbours of $v$, we let $N^n_v$ denote the multiset $\{\{ u^n \mid u \in N \}\}$.
    Then we prove by induction on $n \in \N$ that for any $(\Pi, \cT')$-schema $\varphi$ of modal depth $1$ and for each pointed $\Pi$-labeled graph $(G, w)$, we have $(w^{n}, N^{n}_w) \Vdash \varphi$ iff $G, w \models \varphi^{n}$.

    If $n = 0$ we prove by induction on the structure of $\varphi$ that $(w^{0}, N^{0}_w) \Vdash \varphi$ iff $G, w \models \varphi^{0}$.
    \begin{itemize}
        \item Case $\varphi = p \in \Pi$: Trivial, since $(w^{0}, N^{0}_w) \Vdash p$ iff $p \in w^0$ iff $G, w \models p$ iff $G, w \models p^0$.
        \item Case $\varphi = X \in \cT'$: Let $\psi_X$ be the body of the terminal clause of $X$. Now $(w^{0}, N^{0}_w) \Vdash X$ iff $X \in w^0$ iff $G, w \models \psi_X$ iff $G, w \models X^0$.
    \end{itemize}
    Now, assume that the induction hypothesis holds for $(\Pi, \cT')$-schemata $\psi$ and $\theta$ with modal depth at most $1$.
    \begin{itemize}
        \item Case $\varphi \colonequals \neg \psi$: Now, $(w^{0}, N^{0}_w) \Vdash \neg \psi$ iff $(w^{0}, N^{0}_w) \not\Vdash \psi$. By the induction hypothesis, this is equivalent to $G, w \not\models \psi^0$ which is equivalent to $G, w \models \neg \psi^0$. 
        \item Case $\varphi \colonequals \psi \land \theta$: Now, $(w^{0}, N^{0}_w) \Vdash \psi \land \theta$ iff $(w^{0}, N^{0}_w) \Vdash \psi$ and $(w^{0}, N^{0}_w) \Vdash \theta$. By the induction hypothesis, this is equivalent to $G, w \models \psi^0$ and $G, w \models \theta^0$, which is equivalent to $G, w \models \psi^0 \land \theta^0$. 
        \item Case $\varphi \colonequals \Diamond_{\geq \ell} \psi$:
        First it is easy to show that for every $(\Pi, \cT')$-schema $\psi$ of modal depth $0$ and for every state $q$ of $\cA_{\Gamma}$ that $(q, \emptyset) \Vdash \psi$ iff $(q, N) \Vdash \psi$ for every multiset $N$ of states of $\cA_{\Gamma}$.
        Now, $(w^{0}, N^{0}_w) \Vdash \Diamond_{\geq \ell} \psi$ iff there is a set $Q' \subseteq Q$ of states such that $\sum_{q' \in Q'} N^0_w(q') \geq \ell$ and $(q', \emptyset) \Vdash \psi$ for every $q' \in Q'$ iff there is a set $N$ of out-neighbours of $w$ such that $\abs{N} \geq \ell$ and $(u^0, N^0_u) \Vdash \psi$ for every $u \in N$. By the induction hypothesis this is true iff there is 
        a set $N$ of out-neighbours of $w$ such that $\abs{N} \geq \ell$ and $G, u \models \psi^0$ for every $u \in N$, which is equivalent to $G, w \models \Diamond_{\geq \ell} \psi^0$.    
    \end{itemize}

    Now, assume that the induction hypothesis holds for $n$ and let us prove the case for $n +1$. Similarly to the case $n = 0$, we prove by induction on structure of $\varphi$ that $(w^{n+1}, N^{n+1}_w) \Vdash \varphi$ iff $G, w \models \varphi^{n+1}$. The proof by induction is almost identical, except in the case where $\varphi = X \in \cT'$.
    Let $\chi_X$ be the body of the iteration clause of $X$. Then $(w^{n+1}, N^{n+1}_w) \Vdash X$ iff $X \in w^{n+1}$ iff $(w^n, N^n_w) \Vdash \chi_X$. By the induction hypothesis, this is equivalent to $G, w \models \chi_X^{n}$, which is equivalent to $G, w \models X^{n+1}$.

    Thus, the result above implies that for all pointed $\Pi$-labeled graphs $(G, w)$, the automaton $\cA_{\Gamma}$ is in the state $q$ at $w$ in round $n$ iff $G, w \models X^n$ and $G, w \models p$ for every $X, p \in q$, and $G, w \not\models X^n$ and $G, w \not\models p$ for every $p, X \notin q$. Therefore, by the definition of the set of accepting states we see that $\cA_{\Gamma}$ is equivalent to $\Gamma$.
\end{proof}

The converse direction is easier. Note that the definitions for sizes of $\GMSC$-programs and bounded $\FCMPA$s can be found in Appendix \ref{appendix: sizes}.
\begin{lemma}\label{lemma: k-FCMPA to GMSC}
    For each bounded $\FCMPA$ over $\Pi$, we can construct an equivalent $\Pi$-program of $\GMSC$. The size of the constructed $\GMSC$-program is polynomial in the size of the $\FCMPA$.
\end{lemma}
\begin{proof}
    The proof below is analogous to the proof of Theorem $1$ in \cite{Kuusisto13}.
    Let $\Pi$ be a 
    set of node label symbols, and let $\cA = (Q, \pi, \delta, F)$ be a $k$-$\FCMPA$ over $\Pi$. 
    For each state $q \in Q$, we define a head predicate $X_q$ and the corresponding rules as follows. The terminal clause for $X_q$ is defined by
    \[
    X_q (0) \colonminus \bigvee_{P \subseteq \Pi,\; \pi(P) = q} \Big( \bigwedge_{p \in P} p \land \bigwedge_{p \in \Pi\setminus P} \neg p \Big).
    \]
    To define the iteration clause for each $X_q$ we first define some auxiliary formulae. Given $q, q' \in Q$, we let $M_k(q, q') = \{\, S \in \cM_k(Q) \mid \delta(q, S) = q' \,\}$. Notice that the number of multisets $S$ specified by the set $M_k(q, q')$ is finite. Now, for each $S \in M_k(q, q')$ we define
    \[
    \varphi_S \colonequals \bigwedge_{q \in Q,\, S(q) = n,\, n < k} \Diamond_{=n} X_{q} \land \bigwedge_{q \in Q,\, S(q) = k} \Diamond_{\geq k} X_{q}.
    \]
    Then the iteration clause for $X_q$ is defined by
    $
    X_q \colonminus \bigwedge_{q' \in Q} \Big( X_{q'} \rightarrow \bigvee_{S \in M_k(q', q)} \varphi_S \Big). 
    $
    It is easy to show by induction on $n \in \N$ that for every pointed $\Pi$-labeled graph $(G, w)$ it holds that $G, w \models X_q^n$ if and only if
    $w$ is in the state $q$ in round $n$ (see \cite{Kuusisto13} for the details; the only difference is swapping sets for multisets). 
    The size of $\cA$ is $\abs{\cP(\Pi)} + \abs{Q} (k+1)^{\abs{Q}}$ by definition. The size of all terminal clauses of the constructed program is altogether $\abs{Q} + \ordo(\abs{\Pi} \cdot \abs{\cP(\Pi)})$, as there are $\abs{Q}$ terminal clauses that altogether encode each element of $\cP(\Pi)$ and each encoding is of size $\ordo(\abs{\Pi})$. The size of all iteration clauses of the constructed program is altogether $\ordo(\abs{Q}^{2} + k \abs{Q}^{2} (k+1)^{\abs{Q}}) = \ordo(k \abs{Q}^{2} (k+1)^{\abs{Q}})$, as there are $\abs{Q}$ iteration clauses with $\abs{Q}$ conjuncts each, and they altogether encode every element of $\cM_{k}(Q)$ exactly $\abs{Q}$ times and each encoding is of size $\ordo(k \abs{Q})$. Note that the cardinality of $\cM_k(Q)$ is $(k+1)^{\abs{Q}}$, i.e., the number of functions of type $Q \to [0;k]$. Thus, the size of the program is 
    \[
        \abs{\Pi} + \abs{Q} + \ordo(\abs{\Pi} \cdot \abs{\cP(\Pi)}) + \ordo(k \abs{Q}^{2} (k+1)^{\abs{Q}}) = \ordo(\abs{\Pi} \cdot \abs{\cP(\Pi)} + k \abs{Q}^{2} (k+1)^{\abs{Q}})
    \]
    which is clearly less than $\ordo((\abs{\cP(\Pi)} + \abs{Q} (k+1)^{\abs{Q}})^2)$ and therefore polynomial in the size of $\cA$. If head predicates and states were encoded in binary, the result would still be polynomial, as it would simply add a factor of $\log(\abs{\Pi})$ and $\log(\abs{Q})$ respectively to the sizes of terminal and iteration clauses.
\end{proof}

\begin{proof}[Proof of Proposition \ref{theorem: GMSC = k-FCMPA}]
    Note that the proof uses auxiliary results that are introduced in this subsection.
    Lemma \ref{lemma: GMSC to k-FCMPA} shows that for every $\Pi$-program of $\GMSC$ we can construct an equivalent bounded $\FCMPA$ over $\Pi$. Lemma \ref{lemma: k-FCMPA to GMSC} shows that for every bounded $\FCMPA$ over $\Pi$ we can construct an equivalent $\Pi$-program of $\GMSC$. Thus, we conclude that $\GMSC$ has the same expressive power as bounded $\FCMPA$s.
\end{proof}

\subsection{Proof of Theorem \ref{theorem: k-GNN[F] = k-FCMPA = GMSC}}\label{sec appendix: k-GNNF = k-FCMPA = GMSC}

First we recall Theorem \ref{theorem: k-GNN[F] = k-FCMPA = GMSC}.

\textbf{Theorem \ref{theorem: k-GNN[F] = k-FCMPA = GMSC}.} 
\emph{The following have the same expressive power: $\GNNF$s, $\GMSC$, and R-simple aggregate-combine $\GNNF$s.}

In the end of this appendix section, we formally prove Theorem \ref{theorem: k-GNN[F] = k-FCMPA = GMSC}. Informally, this is done in the following steps. First, we note that it is easy to obtain an equivalent bounded $\FCMPA$ for a given $\GNNF$ and by Proposition \ref{theorem: GMSC = k-FCMPA} we can translate the bounded $\FCMPA$ into an equivalent $\GMSC$-program. The converse direction is much more interesting. In Lemma \ref{lemma: GMSC to simple GNN} we show that for each $\GMSC$-program we can construct an equivalent R-simple aggregate-combine $\GNNF$.

Before we show how to obtain an equivalent R-simple $\GNNF$ for each $\GMSC$-program, we first want to modify the programs so that they are easier to handle via R-simple $\GNNF$s. In particular, the terminal clauses of a program may involve diamonds, the rules of a program may have differing amounts of nested logical operators, and so may the two conjuncts of each conjunction. Thus we prove a lemma which intuitively shows how to ``balance'' $\GMSC$-programs such that this is not the case. We let $\mathrm{fdepth}(\varphi)$ denote the formula depth of a schema $\varphi$ (see the definition of formula depth in Section \ref{sec: logics}).
The formula depth of a terminal clause (resp., of an iteration clause) refers to the formula depth of the body of the clause. The formula depth of a $\GMSC$-program is the maximum formula depth of its clauses.

\begin{lemma}\label{lem: simplify GMSC-program}
    For each $\Pi$-program $\Lambda$ of $\GMSC$ with formula depth $D$, we can construct an equivalent $\Pi$-program $\Gamma$ of $\GMSC$ which has the following properties.
    \begin{enumerate}
    \item Each terminal clause is of the form $X (0) \colonminus \bot$.
    \item Each iteration clause has the same formula depth $\max(3, D + 2)$.
    \item If $\varphi \land \theta$ is a subschema of $\Lambda$ such that neither $\varphi$ nor $\theta$ is $\top$, then $\varphi$ and $\theta$ have the same formula depth.
\end{enumerate}
\end{lemma}
\begin{proof}
    First we define a fresh auxiliary predicate $I$ with the rules $I (0) \colonminus \bot $ and $I \colonminus \top$.
    Informally, we will modify the rules of $\Lambda$ such that the terminal clauses are simulated by the iteration clauses with the help of $I$.
    Then for each head predicate of $\Lambda$ with terminal clause $X (0) \colonminus \varphi$ and iteration clause $X \colonminus \psi$ we define new rules in $\Gamma$ as follows: $X (0) \colonminus \bot$ and $X \colonminus ( \neg I \land \varphi) \lor (I \land \psi)$. Now, the formula depth of $\Gamma$ is $\max(3, D + 2)$.

    First, we ``balance'' the formula depth of each iteration clause in $\Gamma$ according to the second point in the statement of the lemma. Let $d$ be the formula depth of an iteration clause $X \colonminus \psi$, let $D' = \max(3, D + 2)$ and let $n = D'-d$. If $n$ is even, then the new iteration clause for $X$ is  $X \colonminus (\neg)^n \psi$, where $(\neg)^n = \neg \cdots \neg$ denotes $n$ nested negations. If $n$ is odd, then the new iteration clause for $X$ is $X \colonminus (\neg)^{n-1} (\psi \land \top)$. In either case, the new iteration clause for $X$ has formula depth $D'$.
    
    Then we show how to ``balance'' each subschema of $\Gamma$ according to the third point in the statement of the lemma as follows. Let $\varphi \land \psi$ be a subschema of $\Gamma$ such that neither $\varphi$ nor $\psi$ is $\top$, and w.l.o.g. assume that $\mathrm{fdepth}(\varphi) < \mathrm{fdepth}(\psi)$ and let $n = \mathrm{fdepth}(\psi) - \mathrm{fdepth}(\varphi)$. If $n$ is even, then we replace $\varphi \land \psi$ in $\Gamma$ with $(\neg)^n \varphi \land \psi$. If $n$ is odd, then we replace $\varphi \land \psi$ in $\Gamma$ with $(\neg)^{n-1} (\varphi \land \top) \land \psi$.

    We have now obtained the desired $\Gamma$. 
    It is easy to see that $\Gamma$ is equivalent to $\Lambda$ as follows. If $X$ is a head predicate that appears in both $\Lambda$ and $\Gamma$, then for each $n \in \N$ and for each pointed $\Pi$-labeled graph, we have $G, w \models X^n$ w.r.t. $\Lambda$ iff $G, w \models X^{n+1}$ w.r.t. $\Gamma$.
    In~$\Gamma$, the auxiliary predicate $I$ makes sure that each head predicate $X$ in $\Gamma$ that appears in $\Lambda$ simulates in round $1$ the corresponding terminal clause of $X$ in $\Lambda$ and in subsequent rounds simulates the corresponding iteration clause of $X$ in $\Lambda$.
    Furthermore, for each head predicate $X$ in $\Gamma$ and for each pointed $\Pi$-labeled graph $(G,w)$, we have $G, w \not\models X^0$ w.r.t.~$\Gamma$. Thus, $\Gamma$ and $\Lambda$ are equivalent.
\end{proof}

Now we are ready to show the translation from $\GMSC$ to R-simple aggregate-combine $\GNNF$s.

\begin{lemma}\label{lemma: GMSC to simple GNN}
    For each $\Pi$-program of $\GMSC$ we can construct an equivalent R-simple aggregate-combine $\GNNF$ over $\Pi$. 
\end{lemma}
\begin{proof}
Let $\Lambda$ be a $\Pi$-program of $\GMSC$.
Informally, an equivalent R-simple aggregate-combine $\GNNF$ for $\Lambda$ is constructed as follows.
First from $\Lambda$ we construct an equivalent $\Pi$-program $\Gamma$ of $\GMSC$ with Lemma \ref{lem: simplify GMSC-program}. Let $D$ be the formula depth of $\Gamma$.

Intuitively, we build an R-simple aggregate-combine $\GNNF$ $\cG_{\Gamma}$ for $\Gamma$ which periodically computes a single iteration round of $\Gamma$ in $D+1$ rounds. The feature vectors used by $\cG_{\Gamma}$ are \emph{binary} vectors $\bv = \bu\bw$ where: 
\begin{itemize}
    \item $\bu$ has one bit per each (distinct) subschema of a body of an iteration clause in $\Gamma$ as well as for each head predicate, and $\bv_{1}$ keeps track of their truth values, and
    \item $\bw$ has $D+1$ bits and it keeps track of the formula depth that is currently being evaluated.
\end{itemize}
Therefore, $\cG_{\Gamma}$ is a $\GNNF$ over $(\Pi, N + D+1)$, where $N$ is the number of (distinct) subschemata of the bodies of the iteration clauses of $\Gamma$ as well as head predicates, and $D$ is the maximum formula depth of the bodies of the iteration clauses.
In order to be able to compute values, the floating-point system for $\cG_{\Gamma}$ is chosen to be high enough. More precisely, we choose a floating-point system $S$ which can express all integers from $0$ at least up to $K_{\max}$, where $K_{\max}$ is the width of $\Gamma$. 
Note that $0$s and $1$s are also represented in the floating-point system $S$.
Note that although feature vectors exist that have elements other than $1$s and $0$s, they are not used.

Next we define the functions $\pi$ and $\COM$, and the set $F$ of accepting states for $\cG_{\Gamma}$. (Note that $\AGG$ is just the sum in increasing order.)
We assume an enumeration $\mathrm{SUB(\Gamma)} \colonequals (\varphi_{1}, \dots, \varphi_{N})$ of subschemata and head predicates in $\Gamma$ such that if $\varphi_k$ is a subschema of~$\varphi_\ell$, then $k \leq \ell$. 
The initialization function $\pi$ of $\cG_{\Gamma}$ with input $P \subseteq \Pi$ outputs a feature vector $\bv \in \{0,1\}^{N+D+1}$, where the value of each component corresponding to a node label symbol is defined as follows: the component for $p$ is $1$ iff $p \in P$. The other components are $0$s (excluding the possible subschema $\top$ which is assigned $1$, as well as the very last $(N + D + 1)$th bit which is also assigned $1$).

Recall that $S$ is the floating-point system for $\cG_\Gamma$.
The combination function (as per the definition of R-simple $\GNN$s) is $\COM(\bx, \by) = \sigma(\bx \cdot C + \by \cdot A + \bb)$ where $\sigma$ is the truncated $\mathrm{ReLU}$ ($\mathrm{ReLU}^*$) defined by $\mathrm{ReLU}^*(x) = \min(\max(0,x), 1)$, 
$\bb \in S^{N+D+1}$ and $C, A \in S^{(N+D+1) \times (N+D+1)}$, where $\bb, C$ and $A$ are defined as follows. For $k, \ell \leq N+D+1$, we let $C_{k, \ell}$ (resp. $A_{k, \ell}$) denote the element of $C$ (resp. $A$) at the $k$th row and $\ell$th column. Similarly for $\ell \leq N+D+1$ and for any vector $\bv \in S^{N+D+1}$ including $\bb$, 
we let $\bv_\ell$ denote the $\ell$th value of $\bv$. Now, we define the top-left $N \times N$ submatrices of $C$ and $A$ and the first $N$ elements of $\bb$ in the same way as Barceló et al. \cite{DBLP:conf/iclr/BarceloKM0RS20}. For all $\ell \leq N$ we define as follows. 
\begin{itemize}
    \item If $\varphi_{\ell} \in \Pi \cup \{\top\}$, then $C_{\ell, \ell} = 1$.
    \item If $\varphi_\ell$ is a head predicate $X$ with the iteration clause $X \colonminus \varphi_{k}$, then $C_{k, \ell} = 1$.
    \item If $\varphi_{\ell} = \varphi_{j} \land \varphi_{k}$, then $C_{j, \ell} = C_{k, \ell} = 1$ and $\bb_{\ell} = -1$.
    \item If $\varphi_{\ell} = \neg \varphi_{k}$, then $C_{k, \ell} = -1$ and $\bb_{\ell} = 1$.
    \item If $\varphi_{\ell} = \Diamond_{\geq K} \varphi_{k}$, then $A_{k, \ell} = 1$ and $\bb_{\ell} = -K + 1$.
\end{itemize}
Next, we define that the bottom-right $(D+1) \times (D+1)$ submatrix of $C$ is the $(D+1) \times (D+1)$ identity matrix, except that the rightmost column is moved to be the leftmost column. More formally, for all $N + 1 \leq \ell < N+D+1$ we have $C_{\ell, \ell+1} = 1$ and also $C_{N+D+1, N+1} = 1$.
Lastly, we define that all other elements in $C$, $A$ and $\bb$ are $0$s.
Finally, we define that $\bv \in F$ if and only if $\bv_\ell = 1$ (i.e., the $\ell$th value of $\bv$ is $1$) for some appointed predicate $\varphi_{\ell}$ and also $\bv_{N+D+1} = 1$. 

Recall that the formula depth of $\Gamma$ is $D$. Let $\bv(w)_i^t$ denote the value of the $i$th component of the feature vector in round $t$ at node $w$. 
It is easy to show by induction that for all $n \in \N$, for all pointed $\Pi$-labeled graphs $(G, w)$, and for every schema $\varphi_\ell$ in $\mathrm{SUB}(\Gamma)$ of formula depth $d$, we have
\[
\bv(w)_\ell^{n (D+1) + d} = 1 \text{ if } G, w \models \varphi_\ell^n \text{ and $\bv(w)_\ell^{n (D+1) + d} = 0$ if $G, w \not\models \varphi_\ell^n$}.
\]
Most of the work is already done by Barceló et al. (see \cite{DBLP:conf/iclr/BarceloKM0RS20} for the details), but we go over the proof as there are some additional considerations related to recurrence.

For the base case, let $n = 0$. We prove the case by induction over the formula depth $d$ of~$\varphi_{\ell}$. First, let $d = 0$.
\begin{itemize}
    \item Case 1: $\varphi_{\ell} \in \Pi \cup \{\top\}$. Now 
    $\bv_{\ell}^{0} = 1$ and $G,w \models \top^{0}$ if $\varphi_{\ell} = \top$. If $\varphi_{\ell} = p \in \Pi$, then by the definition of the initialization function $\pi$ we have $\bv(w)_{\ell}^{0} = 1$ iff $p \in \lambda(w)$ iff $G,w \models p^{0}$, and $\bv(w)_{\ell}^{0} = 0$ iff $p \notin \lambda(w)$ iff $G,w \not\models p^{0}$.
    \item Case 2: $\varphi_{\ell} = X$, where $X$ is a head predicate of $\Gamma$. Now $\bv(w)_{\ell}^{0} = 0$ due to the definition of the initialization function $\pi$, and $G,w \not\models X^{0}$ because each head predicate of $\Gamma$ has the terminal clause $\bot$.
\end{itemize}
Now, assume the claim holds for $n = 0$ for any formulae $\varphi_{j}, \varphi_{k}$ with formula depth $d-1$. We show that the claim holds for $n = 0$ for formulae $\varphi_{\ell}$ of formula depth $d$.
\begin{itemize}
    \item Case 3: $\varphi_{\ell} = \varphi_{j} \land \varphi_{k}$ for some $\varphi_{j}, \varphi_{k}$ (recall that $\Gamma$ only contains conjunctions where both conjuncts have the same formula depth if neither of them are $\top$). This means we have $C_{j, \ell} = C_{k, \ell} = 1$ and $\bb_{\ell} = -1$. Moreover, $C_{m, \ell} = 0$ for all $m \neq j, k$ and $A_{m, \ell} = 0$ for all $m$. Now 
    \[
        \bv(w)_{\ell}^{d} = \mathrm{ReLU^{*}}\left(\bv(w)_{j}^{d-1} + \bv(w)_{k}^{d-1} -1\right).
    \]
    Thus $\bv(w)_{\ell}^{d} = 1$ iff $\bv(w)_{j}^{d-1} = \bv(w)_{k}^{d-1} = 1$. By the induction hypothesis this is equivalent to $G,w \models \varphi_{j}^{0}$ and $G,w \models \varphi_{k}^{0}$ which is equivalent to $G,w \models (\varphi_{j} \land \varphi_{k})^{0}$. Likewise, $\bv(w)_{\ell}^{d} = 0$ iff $\bv(w)_{j}^{d-1} = 0$ or $\bv(w)_{k}^{d-1} = 0$. By the induction hypothesis this is equivalent to $G,w \not\models \varphi_{j}^{0}$ or $G,w \not\models \varphi_{k}^{0}$ which is equivalent to $G,w \not\models (\varphi_{j} \land \varphi_{k})^{0}$.
    \item Case 4: $\varphi_{\ell} = \neg \varphi_{k}$ for some $\varphi_{k}$. This means we have $C_{k, \ell} = -1$ and $\bb_{\ell} = 1$. Moreover, $C_{m, \ell} = 0$ for all $m \neq k$ and $A_{m, \ell} = 0$ for all $m$. Now 
    \[
        \bv(w)_{\ell}^{d} = \mathrm{ReLU^{*}}\left(-\bv(w)_{k}^{d-1} + 1\right).
    \]
    Thus $\bv(w)_{\ell}^{d} = 1$ iff $\bv(w)_{k}^{d-1} = 0$. By the induction hypothesis this is equivalent to $G,w \not\models \varphi_{k}^{0}$ which is further equivalent to $G,w \models (\neg\varphi_{k})^{0}$. Likewise, $\bv(w)_{\ell}^{d} = 0$ iff $\bv(w)_{k}^{d-1} = 1$. By the induction hypothesis this is equivalent to $G,w \models \varphi_{k}^{0}$ which is equivalent to $G,w \not\models (\neg\varphi_{k})^{0}$.
    \item Case 5: $\varphi_{\ell} = \Diamond_{\geq K} \varphi_{k}$ for some $\varphi_{k}$. This means we have $A_{k, \ell} = 1$ and $\bb_{\ell} = -K + 1$. Moreover, $C_{m, \ell} = 0$ for all $m$ and $A_{m, \ell} = 0$ for all $m \neq k$. Now 
    \[
        \bv(w)_{\ell}^{d} = \mathrm{ReLU^{*}}\left(\mathrm{SUM}_{S}\left(\{\{ \, \bv(v)_{k}^{d-1} \mid (w,v) \in E \, \}\}\right) -K + 1\right),
    \]
    where $\mathrm{SUM}_{S} \colon \cM(S) \to S$ is the sum of floating-point numbers in $S$ in increasing order (see Section \ref{gnndefsjooj} for more details) and $E$ is the set of edges of $G$. Thus $\bv(w)_{\ell}^{d} = 1$ iff there are at least $K$ out-neighbours $v$ of $w$ such that $\bv(v)_{k}^{d-1} = 1$. By the induction hypothesis, this is equivalent to there being at least $K$ out-neighbours $v$ of $w$ such that $G,v \models \varphi_{k}^{0}$ which is further equivalent to $G,w \models (\Diamond_{\geq K} \varphi_{k})^{0}$.
    Likewise, $\bv(w)_{\ell}^{d} = 0$ iff there are fewer than $K$ out-neighbours $v$ of $w$ such that $\bv(v)_{k}^{d-1} = 1$. By the induction hypothesis, this is equivalent to there being fewer than $K$ out-neighbours $v$ of $w$ such that $G,v \models \varphi_{k}^{0}$ which is equivalent to $G,w \not\models (\Diamond_{\geq K} \varphi_{k})^{0}$.
\end{itemize}
Next, assume the claim holds for $n$ for all formulae of any formula depth. We show that it also holds for $n+1$. We once again prove the claim by structure of $\varphi_{\ell}$. Cases 3, 4 and 5 are handled analogously to how they were handled in the case $n = 0$, so we only consider cases 1 and 2.
\begin{itemize}
    \item Case 1: $\varphi_{\ell} \in \Pi \cup \{\top\}$. This means we have $C_{\ell, \ell} = 1$ and $\bb_{\ell} = 0$. Moreover, $C_{m, \ell} = 0$ for all $m \neq \ell$ and $A_{m, \ell} = 0$ for all $m$. Now 
    \[
        \bv(w)_{\ell}^{(n+1)(D+1)} = \mathrm{ReLU^{*}}\left(\bv(w)_{\ell}^{n(D+1) + D}\right).
    \]
    Thus we see that $\bv(w)_{\ell}^{(n+1)(D+1)} = \bv(w)_{\ell}^{n(D+1) + D}$ and a trivial induction shows that also $\bv(w)_{\ell}^{(n+1)(D+1)} = \bv(w)_{\ell}^{n(D+1)}$. Now $\bv(w)_{\ell}^{(n+1)(D+1)} = 1$ iff $\bv(w)_{\ell}^{n(D+1)} = 1$. By the induction hypothesis this is equivalent to $G,w \models \varphi_{\ell}^{n}$ which is equivalent to $G,w \models \varphi_{\ell}^{n+1}$ because $\varphi_{\ell}^{n+1} = \varphi_{\ell}^{n}$. Likewise, $\bv(w)_{\ell}^{(n+1)(D+1)} = 0$ iff $\bv(w)_{\ell}^{n(D+1)} = 0$. By the induction hypothesis this is equivalent to $G,w \not\models \varphi_{\ell}^{n}$ which is equivalent to $G,w \not\models \varphi_{\ell}^{n+1}$.
    \item Case 2: $\varphi_{\ell} = X$, where $X$ is a head predicate of $\Gamma$ with the iteration clause $\varphi_{k}$ of formula depth $D$ (recall that in $\Gamma$, each iteration clause has formula depth $D$). This means we have $C_{k, \ell} = 1$ and $\bb_{\ell} = 0$. Moreover, $C_{m, \ell} = 0$ for all $m \neq k$ and $A_{m, \ell} = 0$ for all $m$. Now
    \[
        \bv(w)_{\ell}^{(n+1)(D+1)} = \mathrm{ReLU^{*}}\left(\bv(w)_{k}^{n(D+1) + D}\right).
    \]
    Thus $\bv(w)_{\ell}^{(n+1)(D+1)} = 1$ iff $\bv(w)_{k}^{n(D+1) + D} = 1$. By the induction hypothesis this is equivalent to $G,w \models \varphi_{k}^{n}$ which is equivalent to $G,w \models X^{n+1}$. Likewise, we see that $\bv(w)_{\ell}^{(n+1)(D+1)} = 0$ iff $\bv(w)_{k}^{n(D+1) + D} = 0$. By the induction hypothesis this is equivalent to $G,w \not\models \varphi_{k}^{n}$ which is equivalent to $G,w \not\models X^{n+1}$.
\end{itemize}
This concludes the induction.

We also know for all $n \geq 1$ and all $N+1 \leq \ell < N+D+1$ that $\bv(w)_{\ell+1}^{n} = \bv(w)_{\ell}^{n-1}$ and $\bv(w)_{N+1}^{n} = \bv(w)_{N + D+1}^{n-1}$. This is because $C_{\ell, \ell+1} = 1$, $C_{\ell', \ell+1} = 0$ for all $\ell' \neq \ell$ and $\bb_{\ell+1} = 0$ and thus
\[
    \bv(w)_{\ell+1}^{n} = \mathrm{ReLU^{*}}\left(\bv(w)_{\ell}^{n-1}\right),
\]
and also $C_{N+D+1, N+1} = 1$, $C_{\ell', N+1} = 0$ for all $\ell' \neq N+D+1$ and $\bb_{N+1} = 0$ and thus
\[
    \bv(w)_{N+1}^{n} = \mathrm{ReLU^{*}}\left(\bv(w)_{N + D+1}^{n-1}\right).
\]
By the initialization this means for all $1 \leq \ell, \ell' \leq D+1$ and all $n \in \N$ that $\bv(w)_{N+\ell}^{n(D+1) + \ell'} = 1$ iff $\ell' = \ell$ and $0$ otherwise. Specifically, this means that $\bv(w)_{N + D + 1}^{n(D+1) + \ell'} = 1$ iff $\ell' = D+1$ and $\bv(w)_{N + D + 1}^{n(D+1) + \ell'} = 0$ otherwise.

Thus if $\varphi_{\ell}$ is an appointed predicate $X$ of $\Gamma$, then we know for all $n \in \N$ that $G,w \models X^{n}$ iff $\bv(w)_{\ell}^{n(D+1)} = 1$ and we also know that $\bv(w)_{N + D + 1}^{n(D+1)} = 1$ and thereby $\bv(w)^{n(D+1)} \in F$. Thus $\cG_{\Gamma}$ is equivalent to $\Gamma$.

Since $\cG_\Gamma$ is equivalent to $\Gamma$ which is equivalent to $\Lambda$, we are done.
\end{proof}

We note that the proof of Lemma \ref{lemma: GMSC to simple GNN} generalizes for other types of $\GNNF$s, such as the type described below. Let $S = ((p, n, \beta), +, \cdot)$ be the floating-point system, where $p=1$, $n = 1$ and $\beta = K_{\max}+1$, where $K_{\max}$ is again the maximum width of the rules of $\Gamma$ in the proof above. We believe that the proof generalizes for other floating-point systems which can represent all non-negative integers up to $K_{\max}$. Recall that $K_{\max}$ is the width of $\Gamma$ in the proof of Lemma \ref{lemma: GMSC to simple GNN}. 
Consider a class of $\GNNF$s over $S$ with dimension $2(N+D+1)$ (twice that in the proof above) where the aggregation function is $\mathrm{SUM}_{S} \colon \cM(S) \to S$ (the sum of floating-point numbers in $S$ in increasing order applied separately to each element of feature vectors) and whose combination function $\mathrm{COM} \colon S^{2(N+D+1)} \times S^{2(N+D+1)} \to S^{2(N+D+1)}$ is defined by 
\[
    \COM \left( x, y\right) = \mathrm{ReLU}(x \cdot C + y \cdot A + \bb),
\]
where $C, A \in S^{2(N+D+1) \times 2(N+D+1)}$ are matrices, $\bb \in S^{2(N+D+1)}$ is a bias vector  
and $\mathrm{ReLU}$ is the rectified linear unit defined by $\mathrm{ReLU}(x) \colonequals \max(0, x)$ as opposed to $\mathrm{ReLU}^*$.
The last $2(D+1)$ elements of feature vectors function as they do above, counting the $2(D+1)$ steps required each time an iteration of $\Gamma$ is simulated (here the computation takes twice as long compared to above). The first $2N$ elements calculate the truth values of subschemata as before, but now each calculation takes two steps instead of one as each subschema is assigned two elements instead of one. The first of these elements calculates the truth value of the subschema as before. However, due to the use of $\mathrm{ReLU}$ instead of $\mathrm{ReLU^*}$ the value of the element might be more than $1$ if the subschema is of the form $\Diamond_{\geq K} \varphi_{k}$. Thus, the second element normalizes the values by assigning a weight of $K_{\max}$ ensuring that each positive value becomes $K_{\max}$ and then assigning the bias $-K_{\max}+1$ to bring them all down to $1$.
More formally, for each subschema $\varphi_{\ell}$, let $\ell$ be the first and $N + \ell$ the second element associated with $\varphi_{\ell}$. In the initialization step we define that if $\varphi_{\ell} \in P \subseteq \Pi$ 
or $\varphi_{\ell} = \top$, then $\pi(P)_{N + \ell} = 1$ and other elements of $\pi(P)$ are $0$s. Now we define as follows for all $1 \leq \ell \leq N$.
\begin{itemize}
    \item If $\varphi_{\ell} \in \Pi \cup \{\top\}$, then $C_{N + \ell, \ell} = 1$.
    \item If $\varphi_\ell$ is a head predicate $X$ with the iteration clause $X \colonminus \varphi_{k}$, then $C_{N + k, \ell} = 1$.
    \item If $\varphi_{\ell} = \varphi_{j} \land \varphi_{k}$, then $C_{N + j, \ell} = C_{N + k, \ell} = 1$ and $\bb_{\ell} = -1$.
    \item If $\varphi_{\ell} = \neg \varphi_{k}$, then $C_{N + k, \ell} = -1$ and $\bb_{\ell} = 1$.
    \item If $\varphi_{\ell} = \Diamond_{\geq K} \varphi_{k}$, then $A_{N + k, \ell} = 1$ and $\bb_{\ell} = -K + 1$.
\end{itemize}
In each of the above cases we also define that $C_{\ell, N + \ell} = K_{\max}$ and $\bb_{N + \ell} = -K_{\max}+1$. Lastly for all $2N+1 \leq \ell < 2(N+D+1)$ we define (as before) that $C_{\ell, \ell+1} = 1$ and also $C_{2(N+D+1), 2D+1} = 1$. All other elements of $C$, $A$ and $\bb$ are $0$s. The set $F$ of accepting feature vectors is defined such that $\bv \in F$ iff $\bv_{N + \ell} = 1$ for some appointed predicate $\varphi_{\ell}$ and also $\bv_{2(N+D+1)} = 1$.
It is then straightforward to prove for all $n \in \N$, for all pointed $\Pi$-labeled graphs $(G,w)$ and for every schema $\varphi_{\ell} \in \mathrm{SUB}(\Gamma)$ of formula depth $d$ that
\[
    \bv(w)_{\ell}^{2n(D+1) + 2d-1} \geq 1 \text{ and } \bv(w)_{N + \ell}^{2n(D+1) + 2d} = 1 \text{ if } G,w \models \varphi_{\ell}^{n}
\]
and
\[
    \bv(w)_{\ell}^{2n(D+1) + 2d-1} = \bv(w)_{N + \ell}^{2n(D+1) + 2d} = 0 \text{ if } G,w \not\models \varphi_{\ell}^{n}.
\] 
The proof is similar to the above, but we also need to show that normalization works as intended, i.e., $\bv(w)_{\ell}^{2n(D+1) + 2d-1} \geq 1$ implies $\bv(w)_{N + \ell}^{2n(D+1) + 2d} = 1$ and $\bv(w)_{\ell}^{2n(D+1) + 2d-1} = 0$ implies $\bv(w)_{N + \ell}^{2n(D+1) + 2d} = 0$. To see this, note that $C_{\ell, N + \ell} = K_{\max}$ and $\bb_{N + \ell} = -K_{\max}+1$ and also $C_{\ell', N + \ell} = 0$ for all $\ell' \neq \ell$ and $A_{\ell', N + \ell} = 0$ for all $\ell'$. Thus
\[
    \bv(w)_{N + \ell}^{2n(D+1) + 2d} = \mathrm{ReLU}\left(K_{\max} \cdot \bv(w)_{\ell}^{2n(D+1) + 2d-1} -K_{\max}+1\right),
\]
which is $0$ if $\bv(w)_{\ell}^{2n(D+1) + 2d-1} = 0$ and $1$ if $\bv(w)_{\ell}^{2n(D+1) + 2d-1} \geq 1$ (because by the choice of $S$ we have $K_{\max} \cdot x = K_{\max}$ for all $x \in S$, $x \geq 1$).

Having proved Lemma \ref{lemma: GMSC to simple GNN}, we are now ready to conclude the main theorem of this appendix section expressed below.
\begin{proof}[Proof of Theorem \ref{theorem: k-GNN[F] = k-FCMPA = GMSC}]
    Note that the proof uses auxiliary results introduced in this subsection.
    Each $\GNNF$ is trivial to translate into an equivalent bounded $\FCMPA$ with linear blow-up in size by the fact that the definitions of $\GNNF$s and bounded $\FCMPA$s are almost identical. Lemma \ref{lemma: GMSC to simple GNN} shows that each $\GMSC$-program can be translated into an equivalent R-simple aggregate-combine $\GNNF$.
    The translation from $\GNNF$s to $\GMSC$-programs causes only polynomial blow-up in size by Lemma \ref{lemma: k-FCMPA to GMSC}.
\end{proof}

\subsection{Proof of Theorem \ref{thrm: equi omega-GML CMPA}}\label{sec: omega-GML CMPA}

Note that $k$-$\CMPA$s are defined in Appendix \ref{appendix: An extended definition for distributed automata}.
First we recall Theorem \ref{thrm: equi omega-GML CMPA}.

\textbf{Theorem \ref{thrm: equi omega-GML CMPA}.}
\emph{$\CMPA$s have the same expressive power as $\omega$-$\GML$.}

The proof of Theorem \ref{thrm: equi omega-GML CMPA} is in the end of this subsection but we need some auxiliary results first.
We show in Lemma \ref{lemma: omega-GML to CMPA} that we can construct an equivalent 
counting type automaton over $\Pi$
for each $\Pi$-formula of $\VGML$.
Informally, to do this, we first define a $\Pi$-formula of $\GML$ called the ``full graded type of modal depth $n$'' for each pointed graph, which expresses all the local information of its neighbourhood up to depth $n$. We show in Proposition \ref{Formula_to_types_GML} that for each $\VGML$-formula there is a logically equivalent disjunction of types. We also define counting type automata that compute the type of modal depth $n$ of each node in every round $n$. The accepting states of the resulting automaton are exactly those types that appear in the disjunction of types.

Then we show in Lemma \ref{lemma: cmpa to omega-gml} that for each $\CMPA$ over $\Pi$ we can construct an equivalent $\Pi$-formula of $\VGML$. Informally, to do this we first prove in Lemma \ref{lem: Types_to_states_GML} that two pointed graphs that satisfy the same full graded type of modal depth $n$ also have identical states in each round $\ell \leq n$ in each $\CMPA$. For each $n \in \N$, we consider exactly those full graded types of modal depth $n$ which are satisfied by some pointed graph that is accepted by the automaton in round $n$. We obtain the desired $\VGML$-formula by taking the disjunction of all these types across all $n \in \N$.

By similar arguments, we also obtain Theorem \ref{theorem: equiv bounded CMPAs and bounded omega-GML} which is analogous to Theorem \ref{thrm: equi omega-GML CMPA} but restricted to bounded $\CMPA$s and width-bounded $\VGML$, which involves defining analogous width-bounded concepts.

Now, we start formalizing the proof of Theorem \ref{thrm: equi omega-GML CMPA}.
To show that for each $\VGML$-formula we can construct an equivalent $\CMPA$, we need to define the concepts of graded $\Pi$-types, full graded $\Pi$-types and counting type automata. The graded $\Pi$-type of modal depth $n$ and width $k$ of a pointed $\Pi$-labeled graph $(G, w)$ is a $\Pi$-formula of $\GML$ that contains all the information from the $n$-neighbourhood of $w$ (according to outgoing edges), with the exception that at each distance from $w$ we can only distinguish between at most $k$ identical branches. The full graded $\Pi$-type of modal depth $n$ lifts this limitation and contains all the information from the neighbourhood of depth $n$ of $w$. A counting type automaton of width $k$ is a $k$-$\CMPA$ that calculates the graded $\Pi$-type of modal depth $n$ and width $k$ of a node in each round $n$. Likewise, we define counting type automata which calculate the full graded $\Pi$-type of modal depth $n$ of a node in each round $n$.

Let $\Pi$ be a set of node label symbols, let $(G, w)$ be a pointed $\Pi$-labeled graph and let $k, n \in \N$. The \textbf{graded $\Pi$-type of width $k$ and modal depth $n$} of $(G, w)$ (denoted $\tau^{(G,w)}_{k, n}$) is defined recursively as follows. For $n = 0$ we define that
\[
    \tau^{(G, w)}_{k, 0} \colonequals \bigwedge_{p_{i} \in \lambda(w)} p_{i} \land \bigwedge_{p_{i} \notin \lambda(w)} \neg p_{i}.
\]
Assume we have defined the graded $\Pi$-type of width $k$ and modal depth $n$ of all pointed $\Pi$-labeled graphs, and let $T_{k, n}$ denote the set of such types. The graded $\Pi$-type of $(G, w)$ of width $k$ and modal depth $n + 1$ is defined as follows:
\[
\begin{aligned}
    \tau^{(G, w)}_{k, n+1} \colonequals \tau^{(G, w)}_{k, 0} &\land \bigwedge_{\ell = 0}^{k-1} \{\, \Diamond_{= \ell} \tau \mid \tau \in T_{k, n} \text{ and } G, w \models \Diamond_{= \ell} \tau \,\} \\
    &\land \{\, \Diamond_{\geq k} \tau \mid \tau \in T_{k, n} \text{ and } G, w \models \Diamond_{\geq k} \tau \,\}.
\end{aligned}
\]
Canonical bracketing and ordering is used to ensure that no two types are logically equivalent, and thus
each pointed $\Pi$-labeled graph has exactly one graded $\Pi$-type of each modal depth and width.

A \textbf{counting type automaton of width $k$} over $\Pi$ is a $k$-$\CMPA$ 
defined as follows. The set $Q$ of states is the set $\bigcup_{n \in \N} T_{k, n}$ of all graded $\Pi$-types of width $k$ (of any modal depth). The initialization function $\pi \colon \cP(\Pi) \to Q$ is defined such that $\pi(P) = \tau^{(G, w)}_{k, 0}$ where $(G, w)$ is any pointed $\Pi$-labeled graph satisfying exactly the node label symbols in $P \subseteq \Pi$. Let $N \colon T_{k, n} \to \N$ be a multiset of graded $\Pi$-types of width $k$ and modal depth $n$, and let $\tau$ be one such type (note that $T_{k, n}$ is finite for any $k, n \in \N$). 
Let $\tau_{0}$ be the unique type of modal depth $0$ in $T_{k, n}$ that does not contradict $\tau$.
We define the transition function $\delta \colon Q \times \cM(Q) \to Q$ such that
\[
        \delta(\tau, N) = \tau_{0} \land \bigwedge_{\ell = 0}^{k-1} \{\, \Diamond_{= \ell} \sigma \mid N(\sigma) = \ell \,\}         \land \{\, \Diamond_{\geq k} \sigma \mid N(\sigma) \geq k \,\},
\]
For other $N$ and $\tau$ that do not all share the same modal depth, we may define the transition as we please such that $\delta(q, N) = \delta(q, N_{|k})$ for all $q$ and $N$.

We similarly define the \textbf{full graded $\Pi$-type of modal depth $n$} of a pointed $\Pi$-labeled graph $(G, w)$ (denoted by $\tau^{(G, w)}_{n}$) which contains all the local information of the $n$-depth neighbourhood of $(G, w)$ with no bound on width. For $n = 0$, we define that $\tau^{(G, w)}_{0} = \tau^{(G, w)}_{k, 0}$ for any $k \in \N$. Assume that we have defined the full graded $\Pi$-type of modal depth $n$ of all pointed $\Pi$-labeled graphs, and let $T_{n}$ be the set of such full types. The full graded $\Pi$-type of modal depth $n+1$ of $(G, w)$ is defined as follows:
\[
    \tau^{(G, w)}_{n+1} \colonequals \tau^{(G, w)}_{0} \land \bigwedge_{\ell \geq 1} \{\, \Diamond_{= \ell} \tau \mid \tau \in T_{n}, (G, w) \models \Diamond_{= \ell} \tau \,\} \land \Diamond_{= \abs{\cN(w)}} \top,
\]
where $\cN(w)$ is the set of out-neighbours of $w$.
The formula tells exactly how many out-neighbours of a node satisfy each full graded type of the previous modal depth; to keep the formulae finite (over finite graphs), instead of containing conjuncts $\Diamond_{= 0} \tau$ it tells exactly how many out-neighbours a node has.

A \textbf{counting type automaton} over $\Pi$ is a $\CMPA$ defined as follows. The set of states is the set of all full graded $\Pi$-types. The initialization function $\pi \colon \cP(\Pi) \to Q$ is defined such that $\pi(P) = \tau^{(G, w)}_{0}$ where $(G, w)$ is any pointed $\Pi$-labeled graph satisfying exactly the node label symbols in $P$. Let $N \colon T_{n} \to \N$ be a multiset of full graded $\Pi$-types of some modal depth $n$, and let $\tau$ be one such type. Let $\tau_{0}$ be the unique full type of modal depth $0$ in $T_{n}$ that does not contradict $\tau$. We define the transition function $\delta \colon Q \times \cM(Q) \to Q$ such that
\[
    \delta(\tau, N) = \tau_{0} \land \bigwedge_{\ell \geq 1} \{\, \Diamond_{= \ell} \sigma \mid N(\sigma) = \ell \,\} \land \Diamond_{= \abs{N}} \top,
\]
For other $N$ and $\tau$ that do not all share the same modal depth, we may define the transition as we please.

We prove the following useful property. 
\begin{proposition}\label{Formula_to_types_GML}
    Each $\Pi$-formula $\varphi$ of modal depth $n$ and width $k$ of $\GML$ has 
    \begin{enumerate}
        \item a logically equivalent countably infinite disjunction of full graded $\Pi$-types of modal depth $n$ and
        \item a logically equivalent finite disjunction of graded $\Pi$-types of width $k$ and modal depth~$n$.
    \end{enumerate}
\end{proposition}
\begin{proof}
    First we prove the case for graded types of width $k$, then we prove the case for full graded types. Let $T_{k, n}$, as above, denote the set of all graded $\Pi$-types of width $k$ and modal depth $n$. Let $\Phi = \{\, \tau \in T_{k,n} \mid \tau \models \varphi \,\}$ and $\neg \Phi= \{\, \tau \in T_{k,n} \mid \tau \models \neg\varphi \,\}$, and let $\bigvee \Phi$ denote the disjunction of the types in $\Phi$. Note that this disjunction is finite since the set $T_{k,n}$ is finite.
    Obviously we have that $\Phi \cap \neg \Phi = \emptyset$ and $\bigvee \Phi \models \varphi$. To show that $\varphi \models \bigvee \Phi$, it suffices to show that $\Phi \cup \neg \Phi = T_{k,n}$.
    Assume instead that $\tau \in T_{k,n} \setminus (\Phi \cup \neg \Phi)$. Then there exist pointed $\Pi$-labeled graphs $(G, w)$ and $(H, v)$ that satisfy $\tau$ such that $G, w \models \varphi$ and $H, v \models \neg \varphi$. Since $(G, w)$ and $(H, v)$ satisfy the same graded $\Pi$-type of modal depth $n$ and width $k$, there can be no $\Pi$-formula of $\GML$ of modal depth at most $n$ and width at most $k$ that distinguishes $(G, w)$ and $(H, v)$, but $\varphi$ is such a formula, which is a contradiction. Ergo, $\bigvee \Phi$ and $\varphi$ are logically equivalent.

    In the case of full graded types we first observe that the set $T_n$ of full graded $\Pi$-types of modal depth $n$ is countable. Thus the set $\Phi = \{\, \tau \in T_{n} \mid \tau \models \varphi \,\}$ is countable, as is the set $\neg \Phi= \{\, \tau \in T_{n} \mid \tau \models \neg\varphi \,\}$. Now, with the same proof as for graded types of width $k$, we can show that $\bigvee \Phi$ and $\varphi$ are logically equivalent. Therefore, $\varphi$ is also logically equivalent to a countably infinite disjunction of full graded $\Pi$-types of modal depth $n$.
\end{proof}

Now, we are ready to show the translation from $\VGML$ to $\CMPA$s.
\begin{lemma}\label{lemma: omega-GML to CMPA}
    For each $\Pi$-formula of $\VGML$, we can construct an equivalent counting type automaton over $\Pi$. If the formula has finite width $k$, we can also construct an equivalent counting type automaton of width $k$.
\end{lemma}
\begin{proof}
    Assume the class $\cK$ of pointed $\Pi$-labeled graphs is expressed by the countable disjunction $\psi \colonequals \bigvee_{\varphi \in S} \varphi$ of $\Pi$-formulae $\varphi$ of $\GML$.
    
    First, we prove the case without the width bound. 
    By Proposition \ref{Formula_to_types_GML}, each $\varphi \in S$ is logically equivalent with a countably infinite disjunction $\varphi^*$ of full graded $\Pi$-types of $\GML$ such that the modal depth of $\varphi^*$ is the same as the modal depth of $\varphi$.
    We define a counting type automaton $\cA$ whose set $F$ of accepting states is the set of types that appear as disjuncts of $\varphi^*$ for any $\varphi \in S$. Now $(G, w) \in \cK$ if and only if $G, w \models \psi$ if and only if $G, w \models \tau$ for some $\tau \in F$ if and only if the state of $(G,w)$ is $\tau$ in round $n$ in $\cA$, where $n$ is the modal depth of $\tau$. Ergo, $\cA$ accepts exactly the pointed $\Pi$-labeled graphs in $\cK$.

    The case for the width bound is analogous. First with Proposition \ref{Formula_to_types_GML} we transform each disjunct $\varphi \in S$ 
    to a finite disjunction $\varphi^+$ of graded $\Pi$-types with the same width and modal depth as $\varphi$. Then, for $\psi$ we construct an equivalent type automaton of width $k$, where the set of accepting states is the set of graded types of width $k$ that appear as disjuncts of $\varphi^+$ for any $\varphi \in S$.
\end{proof}

Before we prove Lemma \ref{lemma: cmpa to omega-gml} we prove another helpful (and quite obvious) lemma.
\begin{lemma}\label{lem: Types_to_states_GML}
    Two pointed $\Pi$-labeled graphs $(G, w)$ and $(H, v)$ satisfy exactly the same full graded $\Pi$-type of modal depth $n$ (respectively, graded $\Pi$-type of width $k$ and modal depth $n$) if and only if they share the same state in each round up to $n$ for each $\CMPA$ (resp., each $k$-$\CMPA$) over $\Pi$.
\end{lemma}
\begin{proof}
    We prove the claim by induction over $n$ without the width bound, since again the case for the width bound is analogous. Let $n = 0$. Two pointed $\Pi$-labeled graphs $(G, w)$ and $(H, v)$ share the same full graded $\Pi$-type of modal depth $0$ if and only if they satisfy the exact same node label symbols if and only if each initialization function $\pi$ assigns them the same initial state.
    
    Now assume the claim holds for $n$. Two pointed $\Pi$-labeled graphs $(G, w)$ and $(H, v)$ satisfy the same full graded $\Pi$-type of modal depth $n + 1$ if and only if \textbf{1)} they satisfy the same full graded $\Pi$-type of modal depth $0$ and \textbf{2)} for each full graded $\Pi$-type $\tau$ of modal depth $n$, they have the same number of neighbors that satisfy $\tau$. Now \textbf{1)} holds if and only if $(G, w)$ and $(H, v)$ satisfy the same node label symbols, and by the induction hypothesis \textbf{2)} is equivalent to the neighbors of $(G, w)$ and $(H, v)$ sharing (pair-wise) the same state in each round (up to $n$) for each $\CMPA$. This is equivalent to $(G, w)$ and $(H, v)$ satisfying the same full graded $\Pi$-type of modal depth $0$,
    and receiving the same multiset of states as messages in round $n$ in each $\CMPA$. By the definition of the transition function, this is equivalent to $(G, w)$ and $(H, v)$ sharing the same state in round $n + 1$ for each $\CMPA$.
\end{proof}

Now, we are able to show the translation from $\CMPA$s to $\VGML$.
\begin{lemma}\label{lemma: cmpa to omega-gml}
    For each $\CMPA$ (respectively, each $k$-$\CMPA$) over $\Pi$, we can construct an equivalent $\Pi$-formula of $\VGML$ (resp., of width $k$).
\end{lemma}
\begin{proof}
    We prove the claim without the width bound, since the case with the width bound is analogous.
    Assume that the class $\cK$ of pointed $\Pi$-labeled graphs is expressed by the counting message passing automaton $\cA$ over $\Pi$. Let $\cT$ be the set of all full graded $\Pi$-types and let 
    \[
    \Phi = \{\, \tau^{(G,w)}_{n} \in \cT \mid \text{$\cA$ accepts the pointed $\Pi$-labeled graph $(G, w) \in \cK$ in round $n$} \,\}.
    \]
    We define the countable disjunction $\bigvee_{\tau \in \Phi} \tau$ and show that $G, w \models \bigvee_{\tau \in \Phi} \tau$ if and only if $\cA$ accepts $(G, w)$.
    Note that $\Phi$ is countable since $\cT$ is countable.
    
    If $G, w \models \bigvee_{\tau \in \Phi} \tau$, then $G, w \models \tau^{(H,v)}_{n}$ for some pointed $\Pi$-labeled graph $(H, v)$ accepted by~$\cA$ in round $n$. This means that $(G,w)$ and $(H,v)$ satisfy the same full graded $\Pi$-type of modal depth $n$. By Lemma \ref{lem: Types_to_states_GML}, this means that $(G, w)$ and $(H, v)$ share the same state in~$\cA$ in each round $\ell \leq n$. Since $\cA$ accepts $(H, v)$ in round $n$, $\cA$ also accepts $(G, w)$ in round $n$. Conversely, if $\cA$ accepts $(G, w)$ in round $n$, then $\tau^{(G,w)}_{n} \in \Phi$ and thus $G, w \models \bigvee_{\tau \in \Phi} \tau$.
\end{proof}

Finally, we prove Theorem \ref{thrm: equi omega-GML CMPA}.
\begin{proof}[Proof of Theorem \ref{thrm: equi omega-GML CMPA}]
    Note that the proof uses auxiliary results that are introduced in this subsection.
    Lemma \ref{lemma: omega-GML to CMPA} shows that for each $\Pi$-formula of $\VGML$ 
    we can construct an equivalent $\CMPA$ 
    over $\Pi$ and Lemma \ref{lemma: cmpa to omega-gml} shows that for each $\CMPA$ 
    over $\Pi$ we can construct an equivalent $\Pi$-formula of $\VGML$. 
\end{proof}

By Lemma \ref{lemma: omega-GML to CMPA} and Lemma \ref{lemma: cmpa to omega-gml}, it is also straightforward to conclude a similar result for bounded $\CMPA$s and width-bounded $\VGML$-formulae.
\begin{theorem}\label{theorem: equiv bounded CMPAs and bounded omega-GML}
Bounded $\CMPA$s have the same expressive power as width-bounded formulae of $\VGML$.
\end{theorem}

\subsection{Proof of Theorem \ref{omega-GML = GNN = CMPAs}}\label{sec appendix: omega-GML = GNN = GMPAs}

First we recall Theorem \ref{omega-GML = GNN = CMPAs}.

\textbf{Theorem \ref{omega-GML = GNN = CMPAs}.}
\emph{$\GNN[\R]$s have the same expressive power as  $\VGML$.}

Informally, Theorem \ref{omega-GML = GNN = CMPAs} is proved as follows (in the end of this section we give a formal proof). We first prove Lemma \ref{lemma: GNN to omega-GML} which shows that we can translate any $\GNN[\R]$ into an equivalent $\VGML$-formula. To prove Lemma \ref{lemma: GNN to omega-GML}, we need to first prove an auxiliary result Lemma \ref{lem: full types and GNNs}, which shows that if two pointed graphs satisfy the same full graded type of modal depth $n$, then those two pointed graphs share the same feature vector in each round (up to $n$) with any $\GNN[\R]$. For these results, we need graded types which were introduced in Section \ref{sec: omega-GML CMPA}.  

Informally, the converse direction (Lemma \ref{lemma: omega-GML to GNN}) is proved as follows. We first translate the $\Pi$-formula of $\VGML$ into an equivalent $\CMPA$ over $\Pi$ by Theorem \ref{thrm: equi omega-GML CMPA}. Then we translate the equivalent $\CMPA$ into an equivalent counting type automaton (again, see Section \ref{sec: omega-GML CMPA} for the definition of counting type automata). Then we prove that we can construct an equivalent $\GNN[\R]$ for each counting type automaton. Informally, we encode each state of the counting type automaton into an integer and the $\GNN[\R]$ can use them to mimic the type automaton in every step. 

In both directions, we also consider the case where $\GNN[\R]$s are bounded and the $\VGML$-formulae are width-bounded.

Before proving Lemma \ref{lemma: GNN to omega-GML}, we first establish the following useful lemma. 
\begin{lemma}\label{lem: full types and GNNs}
    Two pointed $\Pi$-labeled graphs $(G, w)$ and $(H, v)$ satisfy exactly the same full graded $\Pi$-type of modal depth $n$ (respectively the same graded $\Pi$-type of width $k$ and modal depth $n$) if and only if they share the same state in each round (up to $n$) for each unbounded $\GNN[\R]$ (resp., $\GNN[\R]$ with bound $k$) over $\Pi$.
\end{lemma}
\begin{proof}
    The proof is analogous to that of Lemma \ref{lem: Types_to_states_GML}.
\end{proof}

With the above lemma, we are ready to prove Lemma \ref{lemma: GNN to omega-GML}.
\begin{lemma}\label{lemma: GNN to omega-GML}
    For each $\GNN[\R]$ $\cG$ over $\Pi$, we can construct an equivalent $\Pi$-formula of $\VGML$. 
    Moreover, if $\cG$ is bounded with the bound $k$, we can construct an equivalent $\Pi$-formula of $\VGML$ of width $k$.
\end{lemma}
\begin{proof}
    Assume that $\cG$ is a $\GNN[\R]$ over $(\Pi,d)$. Let $\cK$ be the class of pointed $\Pi$-labeled graphs expressed by $\cG$.
    Now let
    \[
    \Phi = \{\, \tau^{(G,w)}_{n} \mid \text{$\cG$ accepts the pointed $\Pi$-labeled graph $(G, w) \in \cK$ in round $n$} \,\},
    \]
    where $\tau^{(G,w)}_{n}$ is the full graded $\Pi$-type of modal depth $n$ of $(G,w)$ (see Section \ref{sec: omega-GML CMPA}). 
    Note that there are only countably many formulae in $\Phi$ since there are only countably many full graded $\Pi$-types.
    Consider the counting type automaton $\cA$ over $\Pi$, where $\Phi$ is the set of accepting states. We will show that $\cA$ accepts $(G, w)$ if and only if $\cG$ accepts $(G, w)$.
    
    If $(G, w)$ is accepted by $\cA$, then $G, w \models \tau^{(H,v)}_{n}$ for some pointed $\Pi$-labeled graph $(H, v)$ accepted by $\cG$ in round $n$. This means that $(G,w)$ and $(H,v)$ satisfy the same full graded $\Pi$-type of modal depth $n$. By Lemma \ref{lem: full types and GNNs}, this means that $(G, w)$ and $(H, v)$ share the same state in $\cG$ in each round $\ell \leq n$. Since $\cG$ accepts $(H, v)$ in round $n$, $\cG$ also accepts $(G, w)$ in round $n$. Conversely, if $\cG$ accepts $(G, w)$ in round $n$, then $\tau^{(G,w)}_{n} \in \Phi$ and thus $(G, w)$ is accepted by $\cA$ by the definition of counting type automata. Thus $\cA$ and $\cG$ are equivalent.
    By Theorem \ref{thrm: equi omega-GML CMPA} we can translate the type automaton $\cA$ into an equivalent $\Pi$-formula of $\VGML$. Therefore, for $\cG$ we can obtain an equivalent $\VGML$-formula.

    Next consider the case where $\cG$ is bounded. We follow the same steps as above with the following modification: instead of constructing a set of full graded types, we construct a set of graded types of width $k$: 
    \[
    \Phi = \{\, \tau^{(G, w)}_{k, n} \mid \text{$\cG$ accepts the pointed $\Pi$-labeled graph $(G, w) \in \cK$ in round $n$} \,\}. 
    \]
    The reasoning of the second paragraph is then modified to refer to graded types of width~$k$ and GNNs with the bound $k$ respectively using Lemma \ref{lem: full types and GNNs}. By Theorem \ref{theorem: equiv bounded CMPAs and bounded omega-GML} we obtain an equivalent width-bounded $\VGML$-formula for the counting type automaton of width~$k$.
\end{proof}

We then show the other direction of Theorem \ref{omega-GML = GNN = CMPAs} in the next lemma.
We first define a graph neural network model that is used as a tool in the proof that follows. A \textbf{recurrent graph neural network over natural numbers} $\GNN[\N]$ over $(\Pi, d)$, is a $\GNN[\R]$ over $(\Pi, d)$ where the feature vectors and the domains and co-domains of the functions are restricted to $\N^d$ instead of $\R^d$.
\begin{lemma}\label{lemma: omega-GML to GNN}
    For each $\Pi$-formula of $\VGML$, we can construct an equivalent $\GNN[\R]$ over $(\Pi, 1)$. 
    Moreover, for each $\Pi$-formula of $\VGML$ of width $k$, we can construct an equivalent bounded $\GNN[\R]$ over $(\Pi, 1)$ with the bound $k$.
\end{lemma}
\begin{proof} 
    We first give an informal description, and then we give a formal proof. We show that each counting type automaton over $\Pi$ can be translated into an equivalent $\GNN[\N]$ over $(\Pi, 1)$, which suffices since counting type automata have the same expressive power as $\VGML$ and $\CMPA$s by Lemma \ref{lemma: cmpa to omega-gml} and by Lemma \ref{lemma: omega-GML to CMPA}.
    Informally, each counting type automaton can be simulated by a graph neural network as follows. Each full type has a minimal tree graph that satisfies the full type. Each such tree can be encoded into a binary string (or more precisely into an integer) in a standard way. These binary strings are essentially used to simulate the computation of the counting type automaton. At each node and in each round, the $\GNN[\N]$ simulates the counting type automaton by combining the multiset of binary strings obtained from its out-neighbours and the local binary string into a new binary string that corresponds to the minimal tree that satisfies the type that would be obtained by the counting type automaton in the same round at the same node.

    Now, we formally prove the statement.
    Let $\cA$ be a counting type automaton. Now we create an equivalent $\GNN[\N]$ $\cG$ over $(\Pi, 1)$ as follows. Each full graded type is converted into the unique finite rooted tree graph $(\cT, r)$ of the same depth as the modal depth of the full type. This graph is encoded into a binary string (in a standard way) as follows.
    \begin{itemize}
        \item The first $n$ bits of the string are $1$s, telling the number of nodes in the tree; we choose some ordering $v_{1}, \dots, v_{n}$ for the nodes. (A natural ordering would be to start with the root, then list its children, then its grandchildren, etc.. The children of each node can be ordered in increasing order of magnitude according to the encodings of their generated subtrees. Another option is to assume that the domain of the tree is always~$[n]$ for some $n \in \N$ and use the standard ordering of integers.) This is followed by a~$0$.
        \item The next $nk$ bits are a bit string $\bb$ that tells which node label symbol is true in which node; for each $\ell \in [n]$ and $i \in [k]$, the $((\ell-1)k + i)$:th bit of $\bb$ is $1$ if and only if $\cT, v_{\ell} \models p_{i}$.
        \item The last $n^{2}$ bits form a bit string $\bb'$, where for each $i, j \in [n]$, the $((i-1)n + j)$:th bit of $\bb'$ is $1$ if and only if $(v_{i}, v_{j}) \in E$, i.e., there is an edge from $v_{i}$ to $v_{j}$.
    \end{itemize}
    
    The GNN $\cG$ operates on these binary strings (we can either use the integers encoded in binary by the binary strings, or interpret the binary strings as decimal strings). In round $0$, the initialization function maps each set of node label symbols to the encoding of the unique one-node tree graph where the node satisfies exactly the node label symbols in question. In each subsequent round $n$, each node receives a multiset $N$ of binary encodings of full graded types of modal depth $n-1$. The aggregation function calculates a binary string that encodes the following graph:
    \begin{itemize}
        \item The graph contains a node $r$ where every node label symbol is false.
        \item For each copy of each element in the multiset $N$, the graph contains a unique copy of the rooted tree encoded by that element.
        \item There is an edge from $r$ to the root of each of these rooted trees; the result is itself a rooted tree where $r$ is the root.
    \end{itemize}
    The combination function receives a binary string corresponding to the node's full graded type $\tau$ of modal depth $n-1$, as well as the string constructed above. It takes the above string and modifies the bits corresponding to the node label symbols at the root; it changes them to be identical to how they were in the root in $\tau$. The obtained binary encoding is the feature of the node in round $n$.

    It is clear that the constructed $\GNN[\N]$ $\cG$ calculates a node's full graded type of modal depth $n$ in round $n$. 
    This is because our construction is an embedding from full graded types to natural numbers. Any other such embedding would also suffice, as the images of the aggregation and combination function are always full graded types and $\GNN[\R]$s do not limit the aggregation and combination functions.
    We can choose the accepting states to be the numbers that encode the accepting states in the corresponding counting type automaton $\cA$, in which case the GNN accepts exactly the same pointed graphs as $\cA$.

    Now consider the case where $\cA$ is a counting type automaton of width $k$. We follow similar steps as above with the following modifications.
    \begin{itemize}
        \item We convert each graded type of width $k$ and modal depth $n$ to a corresponding rooted tree. We choose the tree with depth $n$ and the least amount of branches; in other words, each node has at most $k$ identical out-neighbours.
        \item The aggregation function of $\cG$ constructs a graph in the same way as above, except that it creates at most $k$ copies of a graph encoded by an element in the multiset.
    \end{itemize}
    It is clear that the resulting GNN is bounded, as the aggregation function is bounded by $k$.

    We have now shown that for each counting type automaton we can construct an equivalent $\GNN[\R]$ (or more precisely an equivalent $\GNN[\N]$) and respectively for each counting type automaton of width $k$ we can construct an equivalent bounded $\GNN[\R]$ (or more precisely an equivalent bounded $\GNN[\N]$). Therefore, by Lemma \ref{lemma: cmpa to omega-gml} and by Lemma \ref{lemma: omega-GML to CMPA} for each formula of $\omega$-$\GML$ (and respectively for each $\CMPA$) we can construct an equivalent $\GNN[\R]$. Analogously, for each width-bounded formula of $\omega$-$\GML$ (and resp., for each bounded $\CMPA$) we can construct an equivalent bounded $\GNN[\R]$.
\end{proof}

We are now ready to formally prove Theorem \ref{omega-GML = GNN = CMPAs}.

\begin{proof}[Proof of Theorem \ref{omega-GML = GNN = CMPAs}]
    Note that the proof uses auxiliary results that are introduced in this subsection.
    By Lemma \ref{lemma: GNN to omega-GML} we can construct an equivalent $\Pi$-formula of $\VGML$ for each $\GNN[\R]$ over $\Pi$. By Lemma \ref{lemma: omega-GML to GNN} we can construct an equivalent $\GNN[\R]$ over $\Pi$ for each $\Pi$-formula of $\VGML$.
\end{proof}

\subsection{Proof of Remark \ref{remark: additional results}}\label{appendix: remark}

In this section we formally prove the claims in Remark \ref{remark: additional results}, expressed in Theorems \ref{theorem: FCMPA = unrestricted GNNF} and \ref{theorem: bounded GNNR = width-bounded VGML}.
Note that we do not require boundedness for $\FCMPA$s in the statement of the lemma below.

First we prove an auxiliary result.
\begin{lemma}\label{lemma: FCMPA to GNN[F]}
    For each $\FCMPA$ over $\Pi$ with state set $Q$, we can construct an equivalent unrestricted $\GNNF$ over $(\Pi, \abs{Q}^{2})$.
\end{lemma}
\begin{proof}
    We encode the $\FCMPA$ fully into the $\GNNF$. First we give an informal description. The feature vectors encode which state a node is in using one-hot encoding, i.e., only the bit corresponding to the occupied state is $1$ and others are $0$. 
    For a received multiset $M$ of feature vectors, the aggregation function encodes the function $\delta_{M} \colon Q \to Q, \delta_{M}(q) = \delta(q, M)$ into a vector, i.e., it tells for each state $q_{i}$ which state $q_{j}$ satisfies $\delta(q_{i}, M) = q_{j}$. To encode such a function we need $\abs{Q}^2$ components in the feature vector.
    Finally, the combination function receives (the encodings of) $q$ and $\delta_{M}$ and computes (the encoding of) $\delta_{M}(q)$.
    
    Now we define the construction formally. Let $\Pi$ be a 
    set of node label symbols. Let $\cA = (Q, \pi, \delta, F)$ be an $\FCMPA$ over $\Pi$. Assume some arbitrary ordering $<_Q$ between the states $Q$. Let $q_1, \ldots, q_{\abs{Q}}$ enumerate the states of $Q$ w.r.t. $<_Q$. 
    Given a multiset $M \in \cM(Q)$ the function $\delta_{M} \colon Q \to Q$, $\delta_{M}(q) = \delta(q, M)$ specified by $M$ is possible to encode to the binary string $\bd_M \in \{0,1\}^{\abs{Q}^2}$ as follows.
    If $\delta_M (q_i) = q_j$, then the $((i-1)\abs{Q} + j)$:th bit of~$\bd_M$ is $1$; the other bits are $0$. That is, $\bd_M$ encodes the function $\delta_M$ in binary.

    We construct a GNN over $(\Pi, \abs{Q}^2)$ over a floating-point system $S$ that includes at least $0$ and $1$ (that is to say, small floating-point systems suffice). 
    For all $i \leq \abs{Q}$, we let $\bd_i \in S^{\abs{Q}^2}$ denote the one-hot string where exactly the $i$th bit is $1$ and the other bits are $0$s. 
    \begin{itemize}
        \item Let $P \subseteq \Pi$. The initialization function $\pi'$ is defined as $\pi'(P) = \bd_i$ where $\pi(P) = q_i$. 
        \item The aggregation function $\AGG$ is defined as follows. Assume that the multiset $M$ contains only one-hot strings $\bd_i$, where $i \leq \abs{Q}$. Now $M$ corresponds to a multiset $M^*$, where each $\bd_i$ is replaced by $q_i$. Then $\AGG(M)$ is $\bd_{M^*}$.
        Otherwise, the aggregation function is defined in an arbitrary way. 
        \item The combination function $\COM$ is defined as follows. Let $\bd, \bd' \in S^{\abs{Q}^2}$ such that $\bd = \bd_i$ is a one-hot string for some $i \leq \abs{Q}$ and $\bd' = \bd_M$ for some multiset $M \in \cM(Q)$. Then we define that $\COM(\bd, \bd') = \COM(\bd_i, \bd_M) = \bd_j$ if and only if $\delta_M(q_i) = q_j$. Otherwise $\COM$ is defined in an arbitrary way.
        \item The set $F'$ of accepting states is defined as follows. If $\bd = \bd_j$ is a one-hot string, then $\bd_j \in F' $ if and only if $q_j \in F$. Otherwise $F'$ is defined in an arbitrary way.
    \end{itemize}
    It is easy to show by induction over $n \in \N$ that for any pointed $\Pi$-labeled graph $(G, w)$, the state of $\cA$ at $w$ in round $n$ is $q_i \in Q$ if and only if the state of $\cG$ at $w$ in round $n$ is $\bd_i$.
    
    We note that this proof is based only on providing two embeddings to $S^{d}$ where $S$ is the floating-point system and $d$ is the dimension of the $\GNNF$. The first is an embedding from the set $Q$ and the second is an embedding from the set of functions $Q \to Q$, and the two embeddings need not be related in any way. Thus, any choice of floating-point system $S$ and dimension $d$, such that there are at least $\abs{Q}^{2}$ expressible feature vectors, would suffice and we would be able to define such embeddings and the subsequent aggregation and combination functions. In particular, dimension $\abs{Q}^{2}$ suffices for any floating-point system $S$.
\end{proof}

The first claim in Remark \ref{remark: additional results} is stated as follows.
\begin{theorem}\label{theorem: FCMPA = unrestricted GNNF}
    $\FCMPA$s have the same expressive power as unrestricted $\GNNF$s.
\end{theorem}
\begin{proof}
    It is straightforward to translate an unrestricted $\GNNF$ into an equivalent $\FCMPA$. By Lemma \ref{lemma: FCMPA to GNN[F]}, we can construct an equivalent unrestricted $\GNNF$ for each $\FCMPA$. 
\end{proof}

The second claim in Remark \ref{remark: additional results} is stated as follows.
\begin{theorem}\label{theorem: bounded GNNR = width-bounded VGML}
    Bounded $\GNN[\R]$s have the same expressive power as width-bounded $\VGML$.
\end{theorem}
\begin{proof}
    Note that this proof uses auxiliary results that are in Appendix \ref{sec appendix: omega-GML = GNN = GMPAs}.  
    By Lemma~\ref{lemma: GNN to omega-GML} we can construct an equivalent width-bounded $\Pi$-formula of $\VGML$ for each bounded $\GNN[\R]$ over $\Pi$. By Lemma \ref{lemma: omega-GML to GNN} we can construct an equivalent bounded $\GNN[\R]$ over $\Pi$ for each width-bounded $\Pi$-formula of $\VGML$.
\end{proof}

\subsection{Proof that multiple GNN layers can be simulated using a single layer}\label{appendix: constant iteration equivalence}

In the literature, GNNs are often defined as running for a constant number of iterations unlike our recurrent GNN model, see for example \cite{DBLP:conf/iclr/BarceloKM0RS20, DBLP:conf/lics/Grohe23, DBLP:conf/lics/Grohe21}. Each iteration of the GNN is considered its own layer, and each layer has its own aggregation and combination function. More formally for any $N \in \N$, an \textbf{$N$-layer $\GNN[\R]$} $\cG_{N}$ over $(\Pi, d)$ is a tuple $(\R^{d}, \pi, (\delta^{i})_{i \in [N]}, F)$, 
where $\pi \colon \cP(\Pi) \to \R^{d}$ is the \emph{initialization function}, the function $\delta^{i} \colon \R^{d} \times \cM(\R^{d}) \to \R^{d}$ is the \textbf{transition function of layer $i$} of the form $\delta^{(i)}(q, M) = \COM^{i}(q,\AGG^{i}(M))$ (where $\AGG^{i} \colon \cM(\R^{d}) \to \R^{d}$ is the \textbf{aggregation function of layer $i$} and $\COM^{i} \colon \R^{d} \times \R^{d} \to \R^{d}$ is the \textbf{combination function of layer $i$}) and $F \subseteq \R^{d}$ is the set of \emph{accepting feature vectors}. We define the computation of $\cG_{N}$ in a $\Pi$-labeled graph $G = (V, E, \lambda)$ as follows. In round $0$, the feature vector of a node $v$ is $x^{0}_{v} = \pi(\lambda(v))$. In round $i \in [N]$,
the feature vector of a node is
\[
    x^{i}_{v} = \delta^{i}(x^{i-1}_{v}, \{\{\, x^{i-1}_{u} \mid (v,u) \in E \,\}\}) = \COM^{i}(x^{i-1}_{v}, \AGG^{i}(\{\{\, x^{i-1}_{u} \mid (v,u) \in E \,\}\})).
\]
We say that $\cG_{N}$ \textbf{accepts} a pointed graph $(G, w)$ if and only if $x^{N}_{w} \in F$. Concepts concerning node properties, equivalence and same expressive power are defined as for other models of $\GNN$s.

\begin{proposition}\label{proposition: barcelo model and ours}
    For each $N$-layer $\GNN[\R]$, we can construct an equivalent constant-iteration $\GNN[\R]$.
\end{proposition}
\begin{proof}
    Intuitively, we add a clock to the feature vectors of the $N$-layer $\GNN[\R]$ that tells the transition function of the constant-iteration $\GNN[\R]$ which layer to simulate.

    Let $\cG_{N} = (\R^{d}, \pi, (\delta^{i})_{i \in [N]}, F)$ be an $N$-layer $\GNN[\R]$ 
    over $(\Pi, d)$. We construct a constant-iteration $\GNN[\R]$ $(\cG, N) = ((\R^{d+N}, \pi', \delta', F'), N)$ over $(\Pi, d+N)$ as follows. In all feature vectors used, exactly one of the last $N$ elements is $1$ and others are $0$s. For the initialization function $\pi'$ we define for all $P \subseteq \Pi$ that $\pi'(P) = (\pi(P)_{1}, \dots, \pi(P)_{d},1,0,\dots,0) \in~\R^{d+N}$, where $\pi(P)_{i}$ denotes the $i$th element of $\pi(P) \in \R^{d}$. Before specifying the transition function, we define the following:
    \begin{itemize}
        \item For each feature vector $x = (x_{1}, \dots, x_{d+N}) \in \R^{d+N}$, we let $x' = (x_{1}, \dots, x_{d})$ and $x'' = (x_{d+1}, \dots, x_{d+N})$.
        \item Likewise, for each $M \in \cM(\R^{d+N})$, let $M'$ be $M$ where each $x \in M$ is replaced with $x'$.
        \item Let $f \colon \R^{N} \to \R^{N}$ be a function such that we have $f(1,0,\dots,0) = (0,1,0,\dots,0)$, $f(0,1,0,\dots,0) = (0,0,1,0,\dots,0)$, and so forth until $f(0,\dots,0,1) = (0,\dots,0,1)$.
    \end{itemize} 
    Now, assuming that $x \in \R^{d+N}$ and $M \in \cM(\R^{d+N})$ such that exactly the $i$th element of $x''$ and each $y''$ in $M$ is $1$ and others are $0$s, we define that $\delta'(x, M) = (\delta^{i}(x', M'), f(x''))$ (i.e., we concatenate $\delta^{i}(x', M')$ and $f(x'')$). For other inputs, we define $\delta'$ arbitrarily.
    The set $F'$ of accepting feature vectors is the set of feature vectors $(x_{1}, \dots, x_{d},0,\dots,0,1)$, where $(x_{1}, \dots, x_{d})$ is an accepting feature vector of $\cG_{N}$. It is easy to show that $(\cG, N)$ accepts a pointed $\Pi$-labeled graph $(G, w)$ if and only if $\cG_{N}$ accepts $(G, w)$ (note that for both $N$-layer $\GNN[\R]$s and constant-iteration $\GNN[\R]$s, only the feature vector of a node in round $N$ counts for acceptance).
\end{proof}

\subsection{Proof of Theorem \ref{constant iteration GNN reals}}\label{appendix: constant-iteration GNNs to depth-bounded VGML}

First we recall Theorem \ref{constant iteration GNN reals}.

\textbf{Theorem \ref{constant iteration GNN reals}.} \emph{Constant-iteration $\GNN[\R]$s have the same expressive power as depth-bounded $\VGML$.}

\begin{proof}
    Note that we \emph{heavily} use the proofs of Lemma \ref{lemma: GNN to omega-GML} and Lemma \ref{lemma: omega-GML to GNN}.

    Assume that $(\cG, N)$ is a constant-iteration $\GNN[\R]$ over $(\Pi, d)$. Let $\cK$ be the class of pointed $\Pi$-labeled graphs expressed by $(\cG, N)$. 
    Now let 
    \[
    \Phi = \{\, \tau^{(G,w)}_N \mid \text{$\cG$ accepts $(G, w) \in \cK$ in round $N$} \,\}
    \]
    where $\tau^{(G,w)}_{N}$ is the full graded $\Pi$-type of modal depth $N$ of $(G,w)$ (see Section \ref{sec: omega-GML CMPA}).
    Consider the counting type automaton $\cA$ over $\Pi$, where $\Phi$ is the set of accepting states. It is easy to show with an analogous argument as in the proof of Lemma \ref{lemma: GNN to omega-GML} that $\cA$ accepts $(G, w)$ if and only if $(\cG, N)$ accepts $(G, w)$.

    For the converse, assume that $\psi$ is a depth-bounded $\VGML$-formula over $\Pi$ of modal depth~$D$. First by Lemma \ref{Formula_to_types_GML} we translate $\psi$ into an equivalent formula $\psi^*$ which is a disjunction of full graded $\Pi$-types such that the modal depth of $\psi^*$ is the same as $\psi$. 
    By Lemma \ref{lemma: omega-GML to CMPA}, $\psi^*$ is equivalent to a counting type automaton $\cA$ over $\Pi$ whose accepting states are the types that appear as disjuncts of $\psi^*$.
    Since the depth of each type is bounded by~$D$, each pointed graph accepted by $\cA$ is accepted in some round $r \leq D$. By the proof of Lemma \ref{lemma: omega-GML to GNN} we can construct an equivalent $\GNN[\R]$ $\cG$ over $\Pi$ such that for all pointed $\Pi$-labeled graphs $(G, w)$ and for all $n \in \N$: $\cA$ accepts $(G, w)$ in round $n$ iff $\cG$ accepts $(G, w)$ in round $n$. 
    Now, it is easy to modify $\cG$ such that if $\cG$ accepts a pointed graph in round $m$ then it also accepts that pointed graph in every round $m' > m$. 
    Therefore, for all pointed graphs $(G,w)$ we have that $(\cG, D)$ accepts $(G, w)$ iff $G, w \models \psi^*$ iff $G, w \models \psi$.
\end{proof}

\section{Appendix: Characterizing GNNs over MSO-expressible properties}\label{Appendix: Characterizing GNNs over MSO-expressible properties}

\subsection{Proof of Lemma \ref{lem:MSOplusGMLequalsFO}}\label{appendix:MSOplusGMLequalsFO}

In the end of this subsection we give the formal proof of Lemma \ref{lem:MSOplusGMLequalsFO}. First we give some preliminary definitions.

Let $G=(V,E,\lambda)$ be a graph and let $w_0 \in V$ be a node in $G$. A
\textbf{walk in $G$ starting at $w_0$} is a sequence $p=w_0,\dots,w_n$
of elements of $V$ such that $(w_i,w_{i+1}) \in E$ for all $i \leq n-1$. We use $\mn{tail}(p)$ to denote $w_n$. Now, the \textbf{unraveling} of $G$ at $w_0$ is the graph $U=(V',E',\lambda')$ defined as follows:
\[
\begin{aligned}
  V' &=&& \text{the set of all walks in $G$ starting at $w_0$} \\
  E' &=&& \{\, (p,p') \in V' \times V' \mid p' = pw \text{ for some } w \in V \,\} \\
  \lambda'(p) &=&& \lambda(\mn{tail}(p)) \text{ for all } p \in V'.
\end{aligned}
\]
We say that a formula $\varphi(x)$ is \textbf{invariant under unraveling} if for every graph $G=(V,E,\lambda)$ and 
every $w \in V$, we have $G \models \varphi(w)$ iff $U \models \varphi(w)$, with $U$ the unraveling of $G$ at $w$. Invariance under unraveling is defined in the same way also for GNNs. The following
is easy to prove, see for example \cite{DBLP:books/el/07/BBW2007}.
\begin{lemma}\label{lem: unraveling objects}
   The following are invariant under unraveling: $\VGML$, $\GMSC$, $\GNNF$s, $\GNN[\R]$s, and their constant iteration depth versions. 
\end{lemma}

We recall Lemma \ref{lem:MSOplusGMLequalsFO} and prove it.

\noindent
{\bf Lemma~\ref{lem:MSOplusGMLequalsFO}.} 
\emph{Any property expressible in $\MSO$ and as a constant-iteration $\GNN[\R]$ is also $\mathrm{FO}$-expressible.} 
\begin{proof}
    We present two different proofs. The first one is independent from Theorem \ref{MSO GNN reals and GMSC} and the second one is not.

    Assume that the $\MSO$-formula $\varphi(x)$ over $\Sigma_N$ expresses the same node property as
    the constant iteration depth $\GNN[\R]$~$(\cG, k)$ over $\Sigma_N$, where $k \in \N$ is the iteration depth of $\cG$. It is shown in \cite{DBLP:journals/tocl/ElberfeldGT16} that on every class of graphs of bounded treedepth, $\MSO$ and $\FO$ have the same expressive power. 
    The class $\mathcal{C}$ of all tree-shaped $\Sigma_N$-labeled graphs of 
    depth at most $k$ has bounded treedepth. We thus find an $\FO$-formula $\vartheta(x)$ over $\Sigma_N$ that is logically equivalent to $\varphi(x)$ on $\mathcal{C}$, i.e., for all $T \in \mathcal{C}$ with root $w$ we have $T \models \vartheta(w)$ iff $T \models \varphi(w)$.

    We may manipulate $\vartheta(x)$ into an FO-formula
    $\widehat \vartheta(x)$ such that for any pointed
    graph $(G,w)$, we have $G \models \widehat \vartheta(w)$ if and only if $U_k \models \vartheta(w)$ 
    with $U_k$ the restriction of the unraveling of $G$ at $w$ to elements on level at most $k$. More precisely,
    to construct $\widehat \vartheta(x)$ we do the following:
    \begin{itemize}
        \item First we define an auxiliary formula 
        \[
        \psi_{\leq k}(x,y) \colonequals \bigvee_{0 \leq \ell \leq k} \exists y_0 \cdots \exists y_\ell \Big( y_0 = x \land y_\ell = y \land \bigwedge_{0 \leq m < \ell} E(y_m, y_{m+1}) \Big) 
        \]
        which intuitively states that $y$ lies at distance at most $k$ from $x$. 
        \item Then $\widehat \vartheta(x)$ is obtained from $\vartheta(x)$ by recursively replacing subformulae of type $\exists y \psi$ with $\exists y ( \psi_{\leq k}(x,y) \land \psi )$ as follows. 
        First, we simultaneously replace each subformula of quantifier depth $1$. Having replaced subformulae of quantifier depth $\ell$, we then simultaneously replace subformulae of quantifier depth $\ell + 1$.
    \end{itemize}
    
   We then have, for every graph $G$, the following where $U$ is the unraveling of $G$ at $w$ and $U_k$ denotes the restriction of $U$ to elements on level at most $k$ (the root being on level~0):
    \[
    G \models \varphi(w) \text{ iff }
    U_k \models \varphi(w) \text{ iff }
    U_k \models \vartheta(w) \text{ iff }
    G \models \widehat \vartheta(w) 
    \]
    The first equivalence holds because $\varphi$ is expresses the same property as $\cG$, the second one by choice of $\vartheta$, and the third one by construction of $\widehat \vartheta$.

    We also present an alternative proof which takes advantage of Theorem \ref{MSO GNN reals and GMSC}. 
    By Theorem \ref{constant iteration GNN reals} $(\cG, k)$ is equivalent to some depth-bounded $\Sigma_N$-formula $\psi$ of $\VGML$. 
    On the other hand, by Theorem \ref{MSO GNN reals and GMSC} $\cG$ is equivalent to some $\GMSC$-program $\Lambda$, since the property expressed by $\cG$ is expressible in $\MSO$. By the proof of Proposition \ref{proposition: GMSC < GML} $\Lambda$ is equivalent to some width-bounded $\Sigma_N$-formula $\psi'$ of $\VGML$. Now, it is easy to show that $\psi \land \psi'$ is equivalent to some $\Sigma_N$-formula $\psi^*$ of $\GML$, since $\psi$ is depth-bounded and $\psi'$ is width-bounded. 
    This can be seen by transforming $\psi$ and $\psi'$ into disjunctions of (non-full) graded types by applying Proposition \ref{Formula_to_types_GML} (see also Appendix \ref{sec: omega-GML CMPA} for the definition of graded types). 
    Since $(\cG, k)$ is equivalent to the $\Sigma_N$-formula $\psi^*$ of $\GML$, it expresses a node property also expressible in $\FO$, since $\GML$ is a fragment of $\FO$.
\end{proof}

\subsection{Proof of Theorem \ref{MSO GNN reals and GMSC}}\label{appendix: MSO GNN reals and GMSC}

First we recall Theorem \ref{MSO GNN reals and GMSC}.

\medskip
\noindent
\textbf{Theorem \ref{MSO GNN reals and GMSC}.}
\emph{Let $\cP$ be a property expressible in $\MSO$. Then $\cP$ is expressible as a $\GNN[\R]$ if and only if it is expressible in $\GMSC$.}
\medskip

We recall some details of the proof sketch of Theorem~\ref{MSO GNN reals and GMSC}.
We use an automaton model proposed in~\cite{DBLP:conf/stacs/Walukiewicz96} that captures the expressive power of $\MSO$ on tree-shaped graphs. (Note that the automaton model is defined in Section \ref{sect:MSO}.) 
We then show that the automaton for an $\MSO$-formula $\varphi$ that expresses the same property as a $\GNN[\R]$ (and thus as a formula of $\VGML$) can be translated into a $\GMSC$-program expressing the same property.
To do this, we prove the important Lemma \ref{lem:deco} which shows that for all tree-shaped graphs $T$: the automaton for $\varphi$ accepts $T$ iff there is a $k$-prefix decoration of $T$ for some $k \in \N$. Intuitively, a $k$-prefix decoration of $T$ represents a set of accepting runs of the automaton for $\varphi$ on the prefix $T_k$ of $T$ (the formal definitions are in Section \ref{sect:MSO}). Then we build a $\GMSC$-program that accepts a tree-shaped graph $T$ with root $w$ iff there is a $k$-prefix decoration of $T$ for some $k \in \N$.

We next define in a formal way the semantics of parity
tree automata (the definition of a parity tree automaton (PTA) is in Section \ref{sect:MSO}). For what follows, a \textbf{tree} $T$ is a 
subset of~$\N^*$, the set of all finite words over $\N$, that is closed under prefixes. We say that
$y \in T$ is a \textbf{successor} of $x$ in $T$ if 
$y=xn$ for some $n \in \N$. Henceforth we will call successors \textbf{out-neighbours}. Note that the empty word $\varepsilon$
is then the root of any tree $T$.
A \textbf{$\Sigma$-labeled tree} is a pair $(T,\ell)$
with $T$ a tree and $\ell:T \rightarrow \Sigma$ a node
labeling function. A \textbf{maximal path} $\pi$ in a tree $T$ is a subset of
$T$ such that $\varepsilon \in \pi$ and for each $x \in \pi$ that is
not a leaf in~$T$, $\pi$ contains one 
out-neighbour 
of $x$.
\begin{definition}[Run]
\label{def:altrun}
Let $G$ be a $\Sigma_N$-labeled graph with $G=(V,E,\lambda)$
and $\Amc$ a PTA with $\Amc=~(Q,\Sigma_N,q_0, \Delta, \Omega)$. A
\textbf{run}\footnote{Often semantics for parity tree automata are given with parity games, for the details see for example \cite{GradelErich2003ALaI}. Informally, parity games are played by two players called Eloise and Abelard, where Eloise tries to show that the PTA accepts a given graph. Informally, the semantics introduced here represents a winning strategy of Eloise in parity games and similar semantics are used for example in \cite{automata-theoretic-approarc-vardi}.} of \Amc on $G$ is a $Q \times V$-labeled tree
$(T,\ell)$ such that the following conditions are satisfied:
\begin{enumerate}
    \item $\ell(\varepsilon)=(q_0,v)$ for some $v \in V$;
    
    \item for each
$x \in T$ with $\ell(x)=(q,v)$, 
the following graph satisfies the formula 
$\Delta(q, \lambda(v))$:\footnote{Note that the graph is empty if and only if $v$ is a dead end in $G$.}
  \begin{itemize}

  \item the universe consists of all $u$ with $(v,u) \in E$;

  \item each unary predicate $q' \in Q$ is interpreted as the set
  \[\{\, u \mid \text{ there is 
  an out-neighbour
  $y$ of $x$ in $T$ such that }
  \ell(y)=(q',u) \,\}.\]

  \end{itemize}
\end{enumerate}
  A run $(T,\ell)$ is \textbf{accepting} if for every infinite maximal path $\pi$ of $T$,
  the maximal $i \in \mathbb{N}$, for which the set $\{ x \in \pi \mid \ell(x)=(q,d)
  \text{ with } \Omega(q)=i \}$ is infinite, is even.  We use $L(\Amc)$ to
  denote the language accepted by \Amc, i.e., the set of  $\Sigma_N$-labeled graphs $G$ such that there is an
  accepting run of \Amc on $G$.
\end{definition}

We remark that, in contrast to the standard semantics of $\FO$, the
graph defined in Point~2 of the above definition may be empty
and thus 
transition formulas may also be interpreted in the 
\emph{empty graph}. A transition formula is true in this graph if and only
if it does not contain any  existential
quantifiers, that is, $k=0$. 
Note that a transition formula $\vartheta$ without existential quantifiers is a formula of the form $\forall z (\, \mn{diff}(z) \rightarrow \psi )$, where $\psi$ is a disjunction of conjunctions of atoms $q(z)$ which are unary predicates for the states of the automaton. Such a formula may or may not be true in a non-empty graph. For example, if $\psi$ in  $\vartheta$ is a logical falsity (the empty disjunction), then $\vartheta$ is satisfied only in the empty graph.

Let $\cP$ be a node property over $\Sigma_N$ which is expressible in $\MSO$ and also in $\VGML$.
Let  $\Amc= (Q,\Sigma_N,q_0, \Delta, \Omega)$ be a PTA
that is obtained by Theorem \ref{theorem: msotoapta} from $\cP$. If $\psi$ is the $\VGML$-formula expressing $\cP$, then we may simply say that $\Amc$ and $\psi$ are equivalent.
 We identify a sequence $S_1,\dots,S_n$ of subsets of $Q$
 with a graph $\mn{struct}(S_1,\dots,S_n)$ defined as
 follows:
   \begin{itemize}

   \item the universe is $\{1,\dots, n\}$;

   \item each unary predicate $q' \in Q$ is interpreted as the set
   $\{ i \mid q' \in S_i \}.$

   \end{itemize}
   For a tree-shaped graph $T$, we let $T_k$
   denote the restriction of $T$ to the nodes 
   whose distance from the root is at most $k$.
   An \textbf{extension} of $T_k$ is then any tree-shaped
   graph $T'$ such that $T'_k =T_k$, that is, 
   $T'$ is obtained from $T_k$ by extending the tree from the nodes at distance~$k$ from the root by attaching subtrees, but not from any node at distance $\ell < k$ from the root. Note that \emph{$k$-prefix decorations of $T$} and \emph{universal sets of states} are defined in Section~\ref{sect:MSO}.

\noindent
{\bf Lemma~\ref{lem:deco}.}
  For every tree-shaped $\Sigma_N$-labeled graph $T$:
  $T \in L(\Amc)$  if and only if there is a $k$-prefix decoration of $T$, for some $k \in \N$.
\begin{proof}
    ``$\Rightarrow$''. Assume that $T \in L(\Amc)$. 
    Let $\bigvee_i \psi_i$
    be the $\Sigma_N$-formula of $\VGML$ that \Amc is equivalent to.
    Then  $T \models \psi_i$ for some
    $i$. Let $k$  be the modal depth of $\psi_i$. Then $T' \models \psi_i$ for every extension  $T'$ of $T_k$.

    Next we construct a mapping $\mu \colon V_k \to \cP(\cP(Q))$ and show that $\mu$ is a $k$-prefix decoration of $(T, w)$. We construct $\mu$ as follows. For all $v \in V$ on the level $k$, we define $\mu(v)$ as the universal set for $\lambda(v)$. Then analogously to the third condition in the definition of $k$-prefix decorations, we define $\mu(v)$ for each node $v$ on a level smaller than $k$. Now, all we have to do is show that the first condition in the definition of $k$-prefix decorations is satisfied, i.e., for each $S \in \mu(w)$, $q_0 \in S$. Then we may conclude that $\mu$ is a $k$-prefix decoration. 

    Assume by contradiction that there exists a set $S \in \mu(w)$ such that $q_{0} \notin S$. (Note that $\mu(v) \neq \emptyset$ for all $v \in V_{k}$ by the definition of $k$-prefix decorations.) By Condition~$3$ of $k$-prefix decorations, there must exist sets $S^{1}_{1} \in \mu(u^{1}_{1}), \dots, S^{1}_{n_{1}} \in \mu(u^{1}_{n_{1}})$ (where $u^{1}_{1}, \dots, u^{1}_{n_{1}}$ are the nodes on level $1$) such that $q \in S$ iff the transition formula $\Delta(q, (\lambda(w)))$ is satisfied in the graph where each $u^{1}_{i}$ is labeled $S^{1}_{i}$ for each $1 \leq i \leq n_{1}$. By the same logic, we find such sets $S^{m}_{1} \in \mu(u^{m}_{1}), \dots, S^{m}_{n_{m}} \in \mu(u^{m}_{n_{m}})$ (where $u^{m}_{1}, \dots, u^{m}_{n_{m}}$ are the nodes on level $m$) for each $m \leq k$ (note that these nodes do not necessarily share the same predecessor). Let $T'$ be an extension of $T_{k}$ that is obtained by attaching to each node $u^{k}_{i}$ of $T_{k}$ on level $k$ some rooted tree $T''$ that is accepted by $\Amc$ precisely when starting from one of the states in the set we chose for that node, i.e., 
    $Q_{T''} = S^{k}_{i}$.

    Now, we can demonstrate that there is no accepting run $r_{T'} = (T_{T'}, \ell)$ of $\cA$ on $T'$. Any such run has to begin with $\ell(\varepsilon) = (q_{0}, w)$. Note that the out-neighbours of $\varepsilon$ cannot be labeled with exactly the labels $(q, u^{1}_{i})$ such that $q \in S^{1}_{i}$ because $\Delta(q_{0}, \lambda(w))$ is not satisfied in the graph consisting of the out-neighbours $u^{1}_{i}$ of $w$ labeled with $S^{1}_{i}$ for each $1 \leq i \leq n_{1}$. In fact, the out-neighbours of $\varepsilon$ cannot be labeled with any subset of such labels either, because by definition all transition formulae $\vartheta$ are monotonic in the sense that 
    \[
    (V, E, \lambda) \not\models \vartheta \implies (V, E, \lambda') \not\models \vartheta \text{ for all } \lambda' \subseteq \lambda.
    \]
        Thus, there is 
        an out-neighbour
        $x_{1}$ of $\varepsilon$ in $T_{T'}$ such that $\ell(x_{1}) = (q, u^{1}_{i})$ where $q \notin S^{1}_{i}$, and we may continue this examination starting from $x_{1}$ in the same way. Inductively, we see that for any level $m$ 
    we find a son $x_{m}$ of $x_{m-1}$ such that $\ell(x_{m}) = (q, u^{m}_{i})$ where $q \notin S^{m}_{i}$, including the level $m = k$ where we let $\ell(x_{k}) = (q, u^k_i)$. 
    Thus $q \notin S^k_i$. However,
    we have
     $S^{k}_{i}=Q_{T''}$ 
     and it is witnessed
     by the run $r_{T'}$ that
     \Amc accepts the tree $T''$
     rooted at $u^k_i$ when started in state $q$. Thus, $q \in S^k_i$, a contradiction.
    
    ``$\Leftarrow$''. Let $\mu \colon V_k \to \cP(\cP(Q))$ be a $k$-prefix decoration of $T$, for some $k$. We may
    construct from $\mu$ an accepting run $r=(T',\ell)$ of $\Amc$ on $T$. For every node $v$ in $T_k$ on level $k$, let $T_v$ denote
    the subtree of $T$ rooted at~$v$. Let a {\bf semi-run} be defined like a run except that it needs not satisfy the first condition from the definition of runs (that is, it need not start in the initial state
    of the PTA). A semi-run being {\bf accepting} is defined exactly as for runs.

    Take any node $v$ in $T_k$ on level $k$. Since $\mu(v)$ is
    universal for $\lambda(v)$, we find an $S_v \in \mu(v)$ such that
    $q \in S_v$ if and only if $T_v \in L(\Amc_q)$. (Note that $\mu(v) \neq \emptyset$ for all $v \in V_{k}$ by the definition of $k$-prefix decorations.)
    Consequently, for
    each $q \in S_v$ we find an accepting semi-run $r_{v,q}=(T'_{v,q},\ell_{v,q})$ of \Amc on $T_v$ with $\ell_{v,q}(\varepsilon)=(q,v)$. We now proceed upwards across~$T_k$, assembling
    all these semi-runs into a run. In particular, we choose a set $S_u \in \mu(u)$ for every node $u$ in $T_k$
    and, for all $q \in S_u$,  a semi-run $r_{u,q}=(T'_{u,q},\ell_{u,q})$ with $\ell_{u,q}(\varepsilon)=(q,u)$.
    
    Let $v$ be a node in $T_k$ that has not yet been treated and
    such that its out-neighbours  $u_1,\dots,u_n$ have already been
    treated, that is, we have already selected a set $S_{u_i} \in \mu(u_i)$ for $1 \leq i \leq n$ along with the associated semi-runs. Due to Condition~3 of $k$-prefix decorations, we find a set $S_v \in \mu(v)$ such that
    for each $q \in S_v$: 
    $\Delta(q,\lambda(v))$ is satisfied by the graph
    $\mn{struct}(S_{u_1},\dots,S_{u_n})$.
    Choose this set $S_v$.\footnote{Note that if $v$ has no out-neighbours, then $S_v$ is the set of states $q$ such that $\Delta(q, \lambda(v))$ is satisfied in the empty graph.}
    As for the semi-runs, let $q \in S_v$.
    Then $\Delta(q,\lambda(v))$ is satisfied by the above graph. 
    We may thus choose as $r_{v,q}$ the semi-run that is obtained as follows:
    \begin{enumerate}
        \item start with a fresh root $\varepsilon$ and set
        $\ell_{v,q}(\varepsilon) = (q,v)$;

        \item for each $i \in \{1,\dots,n\}$ and each $q' \in S_{u_i}$, add the semi-run $r_{u_i,q'}$ as a subtree,
        making the root of $r_{u_i,q'}$ an out-neighbour of the 
        fresh root that we had chosen.
    \end{enumerate}
    Note in step $2$ that for each $1 \leq i \leq n$, if $S_{u_{i}} = \emptyset$, then no semi-run $r_{u_{i}, q'}$ is added as a subtree.
    Let $w$ be the root of $T$.
    In view of Condition~1 of $k$-prefix decorations, it is easy
    to verify that the semi-run $r_{w,q_0}$ is in fact a run of \Amc on $T$. Moreover, this run is accepting since we had started with accepting semi-runs at the nodes $v$ in $T_k$ on level $k$ and 
    the finite initial piece that $r_{w,q_0}$ adds on top of those 
    semi-runs has no impact on which states occur infinitely often in infinite paths. 
\end{proof}
By Lemma~\ref{lem:deco}, we may finish the proof of Theorem~\ref{MSO GNN reals and GMSC} by constructing a $\GMSC$-program $\Lambda$
such that for every tree-shaped $\Sigma_N$-labeled graph $T$ with root~$w$, we have  $T,w \models \Lambda$ iff there is a $k$-prefix decoration of $T$, for some $k$. The definition of $\Lambda$ follows the
definition of $k$-prefix decorations. This construction serves as the proof of Lemma \ref{lem: deco GMSC} below it.

\medskip

We define a fresh schema variable $X_S$ for all $S \subseteq Q$. 
First, a set $P \subseteq \Sigma_N$ of node label symbols can be specified with the formula
\[
\varphi_{P} \colonequals \bigwedge_{p \in P} p \wedge \bigwedge_{p \in \Sigma_N \setminus P} \neg p,
\]
which states that the node label symbols in $P$ are true and all others are false.
Let $\cQ_{P}$ denote the universal set for $P$.
For every $X_{S}$,
the program $\Lambda$ contains the following terminal clause, reflecting Condition~2 of $k$-prefix decorations:
\[
X_S(0) \colondash \bigvee_{S \in \cQ_{P}} \varphi_{P}.
\]
Note that if the disjunction is empty, we have $X_S(0) \colonminus \bot$.

We also define a special appointed predicate $A$ that is true when all the head predicates~$X_S$ that do not contain $q_0$ are false, reflecting Condition~1 of $k$-prefix decorations. More formally $A$ is the only appointed predicate of the program, and it is defined as follows: $A(0) \colonminus \bot$ and 
\[
A \colonminus \bigwedge_{q_{0} \notin S} \neg X_{S}.
\]
For the iteration clauses of head predicates $X_{S}$, we need some preliminaries. Let 
$K$ be the maximum over all $k$ such that $\Delta$ mentions
a transition formula 
  \[
    \vartheta \colonequals \exists x_1 \cdots \exists x_k \, \big ( \mn{diff}(x_1,\dots,x_k)
    \wedge q_1(x_1) \wedge \cdots \wedge q_k(x_k) \wedge 
    \forall z (\, \mn{diff}(z,x_1,\dots,x_k) \rightarrow \psi) \big).
  \]

A \emph{counting configuration} $c \in \cM_{K+1}(\cP(Q))$ is a multiset of sets of states that contains each set of states at most $K+1$ times. 
A sequence $S_1,\dots,S_n$ of subsets of $Q$ \textbf{realizes} 
the counting configuration
$c$ if for each $S \subseteq Q$, one of the following holds:
\begin{itemize}
    \item $c(S) \leq K$ and the number of sets $S_i$ among $S_1,\dots,S_n$ with $S_i=S$ is $c(S)$;

 \item $c(S) = K+1$ and the number of sets $S_i$ among $S_1,\dots,S_n$ with $S_i=S$ exceeds $K$.
\end{itemize}
It is easy to prove the following lemma. 
\begin{lemma}
    Let $\vartheta$ be a transition formula mentioned in $\Delta$
    and assume that $S_1,\dots,S_n$ and $S'_1,\dots,S'_{n'}$ realize
    the same counting configuration. Then
    $\mn{struct}(S_{1},\dots,S_{n})$ satisfies $\vartheta$ iff
    $\mn{struct}(S'_{1},\dots,S'_{n'})$ satisfies $\vartheta$.
\end{lemma}
By the above lemma, we may write $c \models \vartheta$, meaning
that $\mn{struct}(S_{1},\dots,S_{n})$ satisfies $\vartheta$ for
any (equivalently: all) $S_1,\dots,S_n$ that realize $c$.

It is easy to see that we can specify a counting configuration $c \in \cM_{K+1}(\cP(Q))$ in $\GMSC$ because we can count out-neighbours; for example, we can write the following formula, which states that for each set $S$ of states there are exactly $c(S)$ out-neighbours where the label includes $S$, unless $c(S) = K+1$, in which case it permits more such out-neighbours:
\[
    \psi_{c} \colonequals \bigwedge_{c(S) = \ell \leq K} \Diamond_{= k} X_{S} \land \bigwedge_{c(S) = K+1} \Diamond_{\geq K+1} X_{S}.
\]
For each set $P$ of node label symbols and set $S$ of states, we can specify the set of counting configurations in $\cM_{K+1}(\cP(Q))$ that satisfy a formula $\Delta(q, P)$ if and only if $q \in S$ with the following disjunction:
\[
\Psi_{S, P} \colonequals \bigvee_{\substack{c \in \cM_{K+1}(\cP(Q)) \\ c \models \Delta(q, P) \text{ iff } q \in S}} \psi_c.
\]
Recall also that the formula $\varphi_{P}$ specifies the set of node label symbols in $P$.
The iteration clause for $X_{S}$ must state that there is a set $P$ of node label symbols that are true and for which $\Psi_{S, P}$ is true. This can be expressed as below:
\[
  X_S \colondash \bigvee_{P \subseteq \Sigma_N} \Big( \varphi_{P} \wedge \Psi_{S, P} \Big).
\]

We prove that the $\GMSC$-program $\Lambda$ characterizes $k$-prefix decorations associated with $\cA$. 
A \textbf{pseudo $k$-prefix decoration} $\mu$ is defined like a $k$-prefix decoration except that it needs not satisfy the first condition of the definition of $k$-prefix decorations (that is, every set in the root given by $\mu$ does not need to contain the initial state). Now, given $\ell \in \N$ and a tree-shaped $\Sigma_N$-labeled graph $T$, let $\mu^\ell_{T}$ denote the pseudo $\ell$-prefix decoration of $T$. For each $\ell \in \N$ and each set $S$ of states of $\Lambda$, we show by induction on $n \in [0; \ell]$ that for every tree-shaped $\Sigma_N$-labeled graph $T$ and for each node $v$ in $T$ on level $\ell - n$, we have
\[
T, v \models X^n_S \iff S \in \mu^\ell_{T}(v).
\]

In the case $n = 0$, the claim holds trivially for the nodes $v$ on level $\ell$ by definition of $\Lambda$, since then $T, v \models X^0_S$ if and only if $S$ is in the universal set for $P$, where $P$ is the set of node label symbols that appear in $v$. 

Assume that the claim holds for $n < \ell$; we will show that it also holds for $n + 1$. 
Let $v$ be a node of $T$ on the level $\ell - (n+1)$.

First assume that $T, v \models X^{n+1}_S$, that is, for some $P \subseteq \Sigma_{N}$, we have $T, v \models \varphi^{n}_P$ and $T, v \models \Psi^{n}_{S, P}$. Therefore, there is a multiset $c \in \cM_{K+1}(\cP(Q))$ such that $c \models \Delta(q, P)$ iff $q \in S$ and $T, v \models \psi^{n}_c$. 
Thus, for every set $S'$ of states of $\cA$ we have that if $c(S') = m \leq K$, there are exactly $m$ out-neighbours $v_1, \ldots, v_m$ of $v$ such that $T, v_i \models X_{S'}^{n}$ for every $i \in [m]$, and if $c(S') = K+1$, there are at least $K+1$ out-neighbours $v_1, \ldots, v_{K+1}$ of $v$ such that $T, v_i \models X_{S'}^{n}$ for every $i \in [K+1]$.
By the induction hypothesis, for every set $S'$ of states of $\cA$ we have that if $c(S') = m \leq K$, there are exactly $m$ out-neighbours $v_1, \ldots, v_m$ of $v$ such that $S' \in \mu^{\ell}_T(v_i)$ for every $i \in [m]$, and if $c(S') = K+1$, there are at least $K+1$ out-neighbours $v_1, \ldots, v_{K+1}$ of $v$ such that $S' \in \mu^{\ell}_T(v_i)$ for every $i \in [K+1]$.
Since $c \models \Delta(q, P)$ iff $q \in S$ and by the definition of pseudo $k$-prefix decorations, we have $S \in \mu^{\ell}_T(v)$.

Then assume that $S \in \mu^\ell_T(v)$.  Let $\{ v_1, \ldots, v_m \}$ be the set of out-neighbours of $v$. By the induction hypothesis for each $i \in [m]$ and each set $S'$ of states of $\cA$, it holds that $S' \in \mu^\ell_T(v_i)$ iff $T, v_i \models X_{S'}^n$. By the definition of pseudo $k$-prefix decorations there are sets $S_1 \in \mu^\ell_T(v_1), \ldots, S_m \in \mu^\ell_T(v_m)$ such that $q \in S$ iff $\mathsf{struct}(S_1, \ldots, S_m) \models \Delta(q, P)$, where $P$ is the set of node label symbols that are true in $v$. Now, let $c \in \cM_{K+1}(\cP(Q))$ be a counting configuration realized by $S_1, \ldots, S_m$. Now since $K$ is the maximum over all $k$ such that $\Delta$ mentions a transition formula
\[
\vartheta \colonequals \exists x_1 \cdots \exists x_k \, \big ( \mn{diff}(x_1,\dots,x_k)
    \wedge q_1(x_1) \wedge \cdots \wedge q_k(x_k) \wedge 
    \forall z (\, \mn{diff}(z,x_1,\dots,x_k) \rightarrow \psi) \big),
\]
we have $c \models \Delta(q, P)$. Therefore, $T, v \models \psi^n_c$ and trivially $T, v \models \varphi^n_P$, and thus $T, v \models X^{n+1}_S$.

Now, if $\Lambda$ accepts a $\Sigma_N$-labeled tree-shaped graph $(T, w)$, then by the result above it means that there is a pseudo $k$-prefix decoration $\mu$ such that $T, w \models A^{k+1}$. Thus in round $k$ we have $T, w \not\models X_S^k$ for all $S$ where $q_0 \notin S$. Therefore $S \notin \mu(w)$ for all $S$ where $q_0 \notin S$, i.e., $\mu$ is actually a $k$-prefix decoration. Similarly if a tree-shaped graph $(T, w)$ has a $k$-prefix decoration, then $T, w \models A^{k+1}$ by the result above. Thus we have proved the lemma below.

\begin{lemma}\label{lem: deco GMSC}
    For every tree-shaped $\Sigma_N$-labeled graph $T$ with root~$w$: $T,w \models \Lambda$ iff there is a $k$-prefix decoration of $T$, for some $k$. 
\end{lemma}

Now we are ready to prove Theorem \ref{MSO GNN reals and GMSC}.
\begin{proof}[Proof of Theorem \ref{MSO GNN reals and GMSC}]
    Note that some of the results needed for this proof are given in this appendix section and some are given outside the appendix. Let $\cP$ be a node property expressible in $\MSO$ by an $\MSO$-formula $\varphi(x)$ over $\Sigma_N$. 
    
    Assume that $\cP$ is expressible by a $\Sigma_N$-program of $\GMSC$. By Proposition \ref{proposition: GMSC < GML} $\cP$ is also expressible $\Sigma_N$-formula of $\VGML$ and thus by Theorem \ref{omega-GML = GNN = CMPAs} as a $\GNN[\R]$ over $\Sigma_N$.
    For the converse, assume that $\cP$ is expressible as a $\GNN[\R]$ over $\Sigma_N$. Thus $\cP$ is expressible by $\Sigma_N$-formula of $\VGML$ by Theorem \ref{omega-GML = GNN = CMPAs}. Therefore, there is a PTA $\cA$ and a $\Sigma_N$-program $\Lambda$ of $\GMSC$ such that for any $\Sigma_N$-labeled graph $G$ with root $w$: 
    \[
    \begin{aligned}
        &G \models \varphi(w) &&\overset{\text{Theorem  \ref{theorem: msotoapta}}}{\iff} \text{the unraveling $U$ of $G$ at $w$ is in $L(\cA)$} \\
        &{} &&\overset{\text{Lemma \ref{lem:deco}}}{\iff} \text{there is a $k$-prefix decoration of $U$, for some $k \in \N$}\\
        &{} &&\overset{\text{Lemma \ref{lem: deco GMSC}}}{\iff} U, w \models \Lambda\\
        &{} &&\overset{\text{Lemma \ref{lem: unraveling objects}}}{\iff} G, w \models \Lambda.
    \end{aligned}
    \]
    Thus, we have proven Theorem \ref{MSO GNN reals and GMSC}.
\end{proof}

\subsection{Proof of Theorem \ref{theorem: under MSO: GNN[R] = GNN[F]}}\label{appendix: under MSO: GNN[R] = GNN[F]}

We recall and prove Theorem \ref{theorem: under MSO: GNN[R] = GNN[F]}.

\textbf{Theorem \ref{theorem: under MSO: GNN[R] = GNN[F]}.}
\emph{Let $\cP$ be a node property expressible in $\MSO$. Then $\cP$ is expressible as a $\GNN[\R]$ if and only if it is expressible as a $\GNNF$. The same is true for constant-iteration $\GNN$s.}

\begin{proof}
    Let $\cP$ be an $\MSO$-expressible property over $\Sigma_N$.
    By Theorem \ref{MSO GNN reals and GMSC}, $\cP$ is expressible as a $\GNN[\R]$ iff it is expressible in $\GMSC$. Thus by Theorem \ref{theorem: k-GNN[F] = k-FCMPA = GMSC} $\cP$ is expressible as a $\GNN[\R]$ iff it is expressible as a $\GNNF$.

    For constant-iteration $\GNN$s we work as follows.
    First assume that $\cP$ is expressible as a constant-iteration $\GNN[\R]$ over $\Sigma_N$. Then by Lemma \ref{lem:MSOplusGMLequalsFO}, $\cP$ is also expressible in $\FO$. Thus, by Theorem 4.2 in \cite{DBLP:conf/iclr/BarceloKM0RS20} 
    $\cP$ is expressible by a $\Sigma_{N}$-formula of $\GML$ (trivially, each constant-iteration $\GNN$ $(\cG, L)$ is trivial to translate to an $L$-layer $\GNN$; see Appendix~\ref{appendix: constant iteration equivalence} for the definition of $L$-layer $\GNN$s). For a $\GML$-formula it is easy to construct an equivalent constant-iteration $\GNNF$; the technique is essentially the same as in the proof of Lemma~\ref{lemma: GMSC to simple GNN} (omit the clock in the construction) and in the proof of Theorem 4.2 in \cite{DBLP:conf/iclr/BarceloKM0RS20}. The converse direction is trivial.
\end{proof}

\section{On accepting}\label{appendix: On accepting}

In this section, we consider run classification schemes for pointed graphs and generalizations of $\GNN$s and $\GMSC$ with global readout and counting global modality. A global readout in a $\GNN$ permits each node to scan all the nodes in each round on top of scanning itself and its neighbours. A counting global modality allows logic formulae of the form $\langle e \rangle_{\geq k} \varphi$, which states that $\varphi$ is true in at least $k$ nodes.

We consider classifications on three levels of generality, moving from most to least general. 
\begin{itemize}[itemsep=0.25em,topsep=0pt,parsep=0pt,partopsep=0pt]
\item 
In Section \ref{Graph classifications and semantics}, we examine the most general level with general graph-classifiers and classification systems with a semantics for labeled graphs. 
\item In Section \ref{Runs and run generators}, we examine the middle level with run-generators, which generate a run in each labeled graph, and run-classifications, which classify said runs; together, a run-generator and run classifier form a graph classifier. A run is simply a (possibly infinite) sequence of labelings over some domain and a run classifier maps each run to a class label.
\item In Section \ref{Extended GMSC, CMPAs and GNNs}, we examine the ground level, where we interpret $\GMSC$-programs, $\GNN$s and $\CMPA$s as run-generators in a natural way, which allows us to define general semantics for them.
\end{itemize}

When comparing classes of run-generators, we have two concepts of equivalence (see Section \ref{Runs and run generators} for the details): 
\begin{enumerate}[itemsep=0.25em,topsep=0pt,parsep=0pt,partopsep=0pt]
    \item run-equivalence, which means that the generated runs match one-to-one, and 
    \item attention-equivalence, which means that the significant rounds of the generated runs match one-to-one, i.e., runs may contain auxiliary computation rounds.
\end{enumerate}
In fact, run-equivalence and attention-equivalence are always defined w.r.t. a color similarity which simply is an equivalence relation between labels used in runs, telling which labels are ``similar''.
The number of auxiliary rounds in our translations is always modest and consistent, multiplying the slowness of the runs by some constant. The notion of attention-equivalence helps in translations from $\GMSC$ to $\GNN$s and $\CMPA$s, and is already used as such in, e.g., \cite{dist_circ_mfcs}. 

In this section, we consider the following modifications to the syntax of $\GMSC$-programs, $\GNN$s and $\CMPA$s which are run-generators. 
First, we let $\GMSCN$ refer to the class of $\GMSC$-programs in normal form, i.e., each terminal clause has modal depth $0$ and each iteration clause has modal depth at most~$1$.
Now, let $s$ be either the syntax of $\GMSC$, of $\GNN$s or of $\CMPA$s.
\begin{itemize}[itemsep=0.25em,topsep=0pt,parsep=0pt,partopsep=0pt]
    \item $s\mathrm{+A}$ refers to adding \textbf{attention} to $s$, i.e. the syntax obtained from $s$ by adding to each syntactical object (i.e. run-generator) of $s$ an additional syntactic component that marks the recorded significant rounds of each run generated by the object.
    \item $s\mathrm{+G}$ refers to adding \textbf{globality} to $s$, i.e. the syntax obtained from $s$ by adding
    \begin{itemize}[itemsep=0.25em,topsep=0pt,parsep=0pt,partopsep=0pt]
        \item a counting global modality, if $s$ is the syntax of $\GMSC$,
        \item a global readout, if $s$ is the syntax of $\GNN$s or of $\CMPA$s.
    \end{itemize}
    \item $s\mathrm{+AG}$ refers to adding both attention and globality to $s$.
\end{itemize}
For example, a program of $\GMSCAGN$ is a $\GMSC$-program in normal form with attention where a counting global modality can appear in the program. On the other hand, $\GNNFG$ refers to a $\GNNF$ that has global readout, but no attention.

\begin{table}[t]

\caption{A summary of the expressivity results in this section. Here $x \equiv y$ means that $x$ and $y$ have the same expressive power with respect to any equivalent run classifiers. Respectively, $x \equiv_a y$ is defined analogously to $x \equiv y$, but the equivalent run classifiers are attention-invariant in the following sense: the run classifiers cannot take into account the auxiliary computation rounds. ``AC'' stands for ``aggregate-combine''. ``Reduction'' means that the result can be proven simply by specifying the run classifiers that induce semantics. Graph-size semantics (as also studied in \cite{Pfluger_Tena_Cucala_Kostylev_2024}) means that a node is accepted if it is in an accepting state in a specific round determined by the size of the graph. Fixed-point semantics (as also studied in \cite{Pfluger_Tena_Cucala_Kostylev_2024}) means that a node is accepted if starting from some round, it is always in an accepting state. Büchi semantics means that a node is accepted if it visits an accepting state an infinite number of times. 
Convergence-based fixed-point semantics (analogous to the accepting condition in the seminal papers \cite{gori, scarcelli}) means that a node is accepted if it is in an accepting state after all nodes in the graph have stopped changing states, i.e., all nodes in the graph reach a fixed-point state.}
\resizebox{0.98\columnwidth}{!}{\begin{tabular}{|c|c|}
    \hline
    \multirow{4}{7em}{\textbf{Theorem~\ref{theorem: equal expressive power with run-similar-equivalent run classifiers}}} &$\GMSCAGN \equiv \GNNFAG$,
    \\
    & $\GMSCAN \equiv \GNNFA$, 
    \\ 
    & $\GMSCGN \equiv \GNNFG$, 
    \\
    &$\GMSCN \equiv \GNNF$  
    \\
    \hline
    \multirow{2}{7em}{\textbf{Theorem~\ref{thrm: equal expressive power with attention-equivalent}}} & $\GMSCAG \equiv_a \GNNFAG \equiv_a \text{R-simple\ AC-} \GNNFAG$,
    \\
    &$\GMSCA \equiv_a \GNNFA \equiv_a \text{R-simple\ AC-} \GNNFA$
    \\
    \hline
    \textbf{Corollary \ref{corollary: graph-size GMSC = graph-size GNN[F] = graph-size R-simple GNN[F]}} & Reduction of Theorem \ref{thrm: equal expressive power with attention-equivalent} to graph-size semantics
    \\
    \hline
    \textbf{Theorem \ref{theorem: büchi and fixed-point}} & Reduction of Theorem \ref{thrm: equal expressive power with attention-equivalent} to fixed-point or Büchi semantics
    \\
    \hline
    \textbf{Theorem \ref{theorem: convergence fixed-point}} & Theorem \ref{thrm: equal expressive power with attention-equivalent} in restriction to convergence-based fixed-point semantics
    \\
    \hline
\end{tabular}}
\label{table: results}
\end{table}

The main results of this section are as follows:
\begin{itemize}[itemsep=0.25em,topsep=0pt,parsep=0pt,partopsep=0pt]
    \item In Theorem \ref{theorem: GMSC[1] + global readout = GNN[F] + global readout}, we show that programs of $\GMSCAGN$ and $\GNNFAG$s are run-equivalent. We can simultaneously remove attention, globality or both, and the result still holds.
    \item In Theorem \ref{theorem: GMSCAG + global readout = GNNF + global readout = R-simple GNNF + global readout}, we show that programs of $\GMSCAG$, $\GNNFAG$s and R-simple aggregate-combine $\GNNFAG$s are attention-equivalent. We can simultaneously remove globality and the result still holds.
\end{itemize}
The main reason why Theorem \ref{theorem: GMSCAG + global readout = GNNF + global readout = R-simple GNNF + global readout} is only obtained w.r.t. attention-equivalence 
is that unlike $\GNNFAG$s, programs of $\GMSCAG$ can scan beyond immediate neighbours in a single round, while $\GNNFAG$s cannot; ergo, a $\GNNFAG$ might need auxiliary rounds to compute a single round of a program of $\GMSCAG$.
From the above theorems, we show that the mentioned classes have the same expressive power with respect to any equivalent run classifiers, including commonly used and natural classifiers; these results can be seen in Table \ref{table: results}.
In Remark \ref{remark:GNNR+Gs = VGMLG}, we also briefly explain how GNNs over reals with global readout can be characterized in terms of graded modal logic with counting global modality.

\subsection{Graph classifications and semantics}\label{Graph classifications and semantics}

In this section, we define highly general concepts of graph classification functions, classification systems and semantics. A graph classifier is a function classifying labeled graphs, and a class of graph classifiers creates a \emph{classification system}.

Let $\Pi$ be a set of node label symbols and let $T$ and $\cC$ be non-empty sets, and let $\fG(\Pi)$ and $\fG'(\Pi)$ denote the classes of $\Pi$-labeled graphs and pointed $\Pi$-labeled graphs respectively. Intuitively, semantics are given with respect to a (purely syntactic) set $\cC$ of graph classifiers, each of which classifies the same class of graphs. Two pointed labeled graphs $(G, v)$ and $(H,u)$ are \textbf{isomorphic} if $f$ is an isomorphism between $G$ and $H$, and $f(v) = u$.

Let $\fC$ be a class of non-empty distinct classes. A \textbf{syntax} w.r.t. $\fC$ is a function of the form $s \colon \cP(\mathrm{LAB}) \to \fC$. A \textbf{pointed semantics over $(\Pi,\cC)$} is a function $g \colon \fG'(\Pi) \times \cC \to T$ such that $g((G,v), C) = g((H,u), C)$ for all $C \in \cC$ when $(G,v)$ and $(H,u)$ are isomorphic. Let $s$ be a syntax w.r.t. $\fC$ and let $\cS$ denote the class of pointed semantics over $(\Pi, \cC)$ for any $\Pi \subseteq \mathrm{LAB}$ and $\cC \in \fC$. A \textbf{pointed semantics over the syntax $s$} is a function $h \colon \cP(\mathrm{LAB}) \to \cS$, that gives for $\Pi \subseteq \mathrm{LAB}$ a pointed semantics over $(\Pi, s(\Pi))$.

A \textbf{pointed graph classifier over $\Pi$} is a function\footnote{Strictly speaking a pointed graph classifier over $\Pi$ is not a function in the sense of being set, but is a class function (or sometimes called a mapping). However, for simplicity, class functions will be referred as functions.} of the form $C \colon \fG'(\Pi) \to T$ such that $C(G,v) = C(H,u)$ when $(G,v)$ and $(H,u)$ are isomorphic. A class $\cC$ of pointed graph classifiers over $\Pi$ is called a \textbf{pointed classification system over $\Pi$}, and it defines the pointed semantics $g$ over $(\Pi, \cC)$ such that $g((G,v), C) = C(G,v)$. Let $\fC$ denote the class of all pointed classification systems over $\Pi$ for all $\Pi$. A \textbf{pointed classification system} is the syntax $\cL \colon \cP(\mathrm{LAB}) \to \fC, \cL(\Pi) = \cC$, where $\cC$ is a pointed classification system over~$\Pi$, and it defines the pointed semantics $h$ over itself such that $h(\Pi) = g$, where $g$ is the pointed semantics defined by $\cL(\Pi)$.

We then define highly general concepts of equivalence between graph classifiers and classification systems. We say that two pointed graph classifiers $C$ and $C'$ over $\Pi$ are \textbf{equivalent} if $C(G,v) = C'(G,v)$ for all $(G,v) \in \fG'(\Pi)$. We say that two classification systems $\cC$ and~$\cC'$ over $\Pi$ \textbf{have the same expressive power} if for each pointed graph classifier $C \in \cC$ there is an equivalent pointed graph classifier $C' \in \cC'$ and vice versa. Finally, two pointed classification systems $\cL$ and $\cL'$ \textbf{have the same expressive power} if $\cL(\Pi)$ and $\cL'(\Pi)$ have the same expressive power for all $\Pi$.

\subsection{Runs and run generators}\label{Runs and run generators}

In this section, we define runs, run classifications and run-generators. Runs are sequences of labelings of the elements of some domain, and run classifications classify runs with respect to some set of classification values. Run-generators are objects that generate a run for each labeled graph over its domain. Intuitively, $\GMSC$-programs, $\CMPA$s and $\GNN$s can all be seen as run-generators; a $\GMSC$-program labels nodes with the sets of head predicates that are true in them, a $\CMPA$ labels nodes with states, and a $\GNN$ labels nodes with feature vectors. This gives us a flexible framework for defining different classifications.

\subsubsection*{Runs and run classifications}

A set of \textbf{configuration labels} is any non-empty set $S$. A \textbf{configuration labeling} over a domain $V$ with respect to $S$ is a partial function $V \rightharpoonup S$. A \textbf{run-skeleton} over $V$ w.r.t. $S$ is then a function $r \colon \N \to C(V, S)$ where $C(V, S)$ is the set of configuration labelings for~$V$ w.r.t. $S$, such that $\dom(r(i+1)) \subseteq \dom(r(i))$, i.e., once the run-skeleton stops labeling some node in the domain, it does not resume labeling that same node. We call $r(i)$ the \textbf{$i$th configuration labeling} of $r$ and write $r_{i}(v) = r(i)(v)$. Let $\cA$ be a set of \textbf{colors}. An \textbf{assignment} over $S$ w.r.t. $\cA$ is a function $f \colon S \to \cA$ that assigns a color to each configuration label in $S$.

We examine four different types of runs. Informally, a run assigns colors to the configuration labels of a run-skeleton. An attention run pairs a run up with a run-skeleton that indicates which rounds are important (labeled $1$) and which ones are unimportant (labeled $0$). A pointed run specifies a node in the domain of a run, and a pointed attention run does the same for an attention run.
More formally, let $r$ be a run-skeleton over $V$ w.r.t. $S$, $f$ an assignment over $S$ w.r.t. $\cA$, $v \in V$ a node and $\alpha$ a run-skeleton over $V$ w.r.t. $\{0,1\}$. We call $(r, f)$ a \textbf{run}, $(r, f, v)$ a \textbf{pointed run} (or a p-run), $(r, f, \alpha)$ an \textbf{attention run} (or an a-run) and $(r, f, \alpha, v)$ a \textbf{pointed attention run} (or a pa-run) over $V$ w.r.t. $(S, \cA)$.
We let $\textit{runs}(V, S, \cA)$ denote the set of all runs over $V$ w.r.t. $\cA$, and $\textit{runs}(S, \cA)$ denote the union of sets $\textit{runs}(V, S, \cA)$ over all domains $V$. We define analogous sets for p-runs, a-runs and pa-runs. For a-runs and pa-runs, we call each natural number for which $\alpha(u) = 1$ an \textbf{attention round} of the a-run or pa-run at $u$. 

While two runs over the same domain can use different sets of colors, the runs may still be identical when the colors are interpreted in a certain way. If, say, one run uses the color cyan and the other uses the colors turquoise and aquamarine, these can all be interpreted as blue and thus perceived to be similar. Thus, we introduce the concept of color similarity that pairs up configuration labels that are assigned similar colors. Let $f_{1} \colon S_{1} \to \cA_{1}$ and $f_{2} \colon S_{2} \to \cA_{2}$ be assignments where $S_{1}$ and $S_{2}$ are disjoint, as are $\cA_{1}$ and $\cA_{2}$.
Now, let $S \colonequals S_{1} \cup S_{2}$ and $f \colonequals f_{1} \cup f_{2}$.
Let $\fP_{1}$ and $\fP_{2}$ be partitions of $\cA_{1}$ and $\cA_{2}$ respectively, and let $p \colon \fP_{1} \to \fP_{2}$ be a bijection. 
Let $\fP \colonequals \{\, P \cup p(P) \mid P \in \fP_{1} \,\}$ be a set that unites all color sets connected by $p$.
Let $\sim_{p}$ be the equivalence relation defined by
\[
    \sim_{p} \colonequals \{\, (s, s') \in S \times S \mid \exists P \in \fP : f(s), f(s') \in P \,\}.
\]
We call $\sim_{p}$ a \textbf{color similarity} w.r.t. $\cA_{1}$ and $\cA_{2}$.

Now that we have a notion of similarity between configuration labels, we can define a notion of run-similarity, which means that two runs are essentially the same, where only the names of the domain elements and configuration labels are different according to a color similarity. More formally, we say that two pa-runs $(r, f, \alpha, v) \in \textit{pa-runs}(V, S, \cA)$ and $(r', f', \alpha', v') \in \textit{pa-runs}(V', S', \cA')$ are \textbf{run-similar w.r.t. $\sim_{p}$} if there is a bijection $g \colon V \to V'$ such that
\begin{enumerate}
    \item $g(v) = v'$,
    \item $r_{i}(u) \sim_{p} r'_{i}(g(u))$ for all $u \in V$ and $i \in \N$,
    \item $\alpha_{i}(u) = \alpha'_{i}(g(u))$ for all $u \in V$ and $i \in \N$.
\end{enumerate}
Moreover, they are \textbf{domain-isomorphic} if only the names of the domain elements differ, i.e., if $S = S'$, $\cA = \cA'$, $f = f'$ and $r_{i}(u) = r'_{i}(g(u))$ for all $u \in V$ and $i \in \N$.

For pa-runs, we define a method for removing the unimportant rounds of the run at each node, and a notion of similarity based on this process; the below concepts are defined analogously for a-runs. First, for any pa-run $(r, f, \alpha, v)$, if $\alpha_{i}(u) = 1$ and also $\abs{\{\, j \in \N \mid j < i, \alpha_{j}(u) = 1 \,\}} = k-1$, then we say that $i$ is the \textbf{$k$th attention round at $u$}. An \textbf{attention-processing function} for pa-runs w.r.t. $(S, \cA)$ is the function $h \colon \textit{pa-runs}(S, \cA) \to \textit{p-runs}(S, \cA)$ where for all $(r, f, \alpha, v) \in \textit{pa-runs}(S, \cA)$ we have that $h(r, f, \alpha, v) = (r', f, v)$ such that
\begin{itemize}
    \item $r'_{k}(v) = r_{i}(u)$ if $i$ is the $(k+1)$th attention round of $(r, f, \alpha, v)$ at $u$, and
    \item $u \notin \dom(r'_{k})$ if there is no $(k+1)$th attention round of $(r, f, \alpha, v)$ at $u$.
\end{itemize}
If $h$ and $h'$ are the attention-processing functions for $(S, \cA)$ and $(S', \cA')$ respectively, then we say that $\rho \in \textit{pa-runs}(V, S, \cA)$ and $\rho' \in \textit{pa-runs}(V', S', \cA')$ are \textbf{attention-similar w.r.t.~$\sim_{p}$} if $h(\rho)$ and $h'(\rho')$ are run-similar w.r.t.~$\sim_{p}$.

\begin{proposition}
    If two a-runs or pa-runs are run-similar, they are also attention-similar.
\end{proposition}

Now that we have defined runs, we need a general way to classify them w.r.t. some non-empty set $T$ of classification labels, such as $\{0,1\}$. A \textbf{pa-run classification} w.r.t. $(S, \cA)$ is a function $C \colon \textit{pa-runs}(S, \cA) \to T$ such that $C(\rho) = C(\rho')$ whenever the pa-runs $\rho$ and $\rho'$ are domain-isomorphic. Classifications for p-runs, a-runs and runs are defined analogously. If $C(\rho) = C(\rho')$ whenever $h(\rho)$ and $h(\rho')$ are domain-isomorphic, where $h$ is the attention-processing function for pa-runs w.r.t. $(S, \cA)$, then we call $C$ \textbf{attention-invariant}.

Finally, we define notions of similarity for comparing pa-run classifications, which extend for a-runs, p-runs and runs analogously. Two pa-run classifications $C, K \colon \textit{pa-runs}(S, \cA) \to~T$ are \textbf{equivalent} if $C(\rho) = K(\rho)$ for all pa-runs $\rho$. Given a color similarity $\sim_{p}$ w.r.t. $\cA$ and~$\cA'$, then we say that $C \colon \textit{pa-runs}(S, \cA) \to T$ and $C' \colon \textit{pa-runs}(S', \cA') \to T$ are \textbf{equivalent w.r.t.~$\sim_{p}$} if $C(\rho) = C'(\rho')$ whenever $\rho$ and $\rho'$ are run-similar w.r.t. $\sim_{p}$.

\begin{proposition}
    If $C$ and $C'$ are attention-invariant pa-run classifications and equivalent w.r.t. $\sim_{p}$, then $C(\rho) = C'(\rho')$ when $\rho$ and $\rho'$ are attention-similar w.r.t. $\sim_{p}$.
\end{proposition}

\subsubsection*{Run-generators}

Next, we define run-generators which, as the name implies, generate runs over labeled graphs. We also define semantics for run-generators using run classifiers, and notions of equivalence for comparing run-generators using different configuration labels.

A \textbf{run-generator} \textbf{over $\Pi$} w.r.t. $(S, \cA)$ is a pair $(R, f)$ where $f \colon S \to \cA$ is an assignment and $R \colon \fG(\Pi) \to \textit{runs}(S, \cA)$ is a function that for each $\Pi$-labeled graph $G = (V, E, \lambda)$ gives a run $(r, f) \in \textit{runs}(V, S, \cA)$, and for any two isomorphic $\Pi$-labeled graphs $G$ and $H$, the runs $R(G)$ and $R(H)$ are domain-isomorphic. 
Each run-generator $(R, f)$ over $\Pi$ induces a pointed run $(R(G),v)$ in each $\Pi$-labeled graph $G$.
An \textbf{attention run-generator} is defined analogously such that $R$ is a function $\fG(\Pi) \to \textit{a-runs}(S, \cA)$ that likewise respects isomorphisms.
Note that run-generators are uniform in the sense that the assignment $f$ does not depend on the domain of the graph.
We could define non-uniform run-generators, where the assignment can depend on the graph and the pointed run can depend on the node, but we do not need them.

We call a class $\cR$ of run-generators (or attention run-generators) a \textbf{pre-classification system over $\Pi$}. Likewise, we call a function $\cS \colon \cP(\mathrm{LAB}) \to \fR$, $\cS(\Pi) = \cR$, where $\cR$ is a pre-classification system over $\Pi$, a \textbf{pre-classification system}.
We can pair a run-generator $(R,f)$ over~$\Pi$ up with a run-classifier $C$ that uses the same configuration labels, in which case we obtain a pointed graph classifier over $\Pi$ that maps each pointed $\Pi$-labeled graph $(G,v)$ to $C(R(G),v)$. When pairing up a run-generator and run-classifier, we naturally always assume that they use the same configuration labels. 
By pairing up each run-generator in a pre-classification system~$\cR$ up with a run classifier, we obtain a classification system. If the run classifier is the same~$C$ for all run-generators, we call the classification system \textbf{uniform} and denote it by $(\cR, C)$.

Next, we define similarity between run-generators based on the notion of color-similarity. Run-similar run-generators give similar runs in every graph. More formally, two run generators (or attention run-generators) $(R,f)$ and $(R',f')$ over $\Pi$ w.r.t. $(S, \cA)$ and $(S', \cA')$ are \textbf{run-similar w.r.t. $\sim_p$} if for all $\Pi$-labeled graphs $G$, $R(G)$ and $R'(G)$ are run-similar w.r.t.~$\sim_p$. We say that two attention run-generators $(R,f)$ and $(R',f')$ are \textbf{attention-similar w.r.t. $\sim_p$} if $R(G)$ and $R'(G)$ are attention-similar w.r.t.~$\sim_p$.

\begin{proposition}\label{proposition: run-similar run-generators are attention-similar}
    If two attention run-generators are run-similar, then they are also attention-similar.
\end{proposition}

Let $\cR$ and $\cR'$ be classes of run-generators (or attention run generators) w.r.t. $(S, \cA)$ and $(S', \cA')$.
We say that $\cR$ and $\cR'$ are \textbf{run-equivalent w.r.t.} $\sim_p$ if for each $(R,f) \in \cR$ there exists a run-similar $(R', f') \in \cG'$ w.r.t. $\sim_p$ and vice versa. If $\cR$ and $\cR'$ are classes of attention run-generators, then we say that they are \textbf{attention-equivalent w.r.t. $\sim_p$} if for each $(R,f) \in \cR$ there exists an attention-similar $(R',f') \in \cR'$ w.r.t. $\sim_p$ and vice versa.

\begin{lemma}
    If two classes of attention run-generators are run-equivalent, then they are also attention-equivalent.
\end{lemma}

With uniform classification systems, run-equivalence between run-generators and run classifications is enough to guarantee equivalent expressive power. Let $\cR_{1}$ and $\cR_{2}$ be classes of run-generators or attention run-generators over $\Pi$, and let $(\cR_1, C_1)$ and $(\cR_2, C_2)$ be uniform (pointed) classification systems over $\Pi$. 
\begin{proposition}\label{proposition: run-similar + run-similar equivalent = equivalent expressive power}
    $(\cR_1, C_1)$ and $(\cR_2, C_2)$ have the same expressive power if there is a color similarity $\sim_p$ such that either of the following conditions holds:
    \begin{itemize}
        \item $\cR_1$ and $\cR_2$ are run-equivalent w.r.t. $\sim_p$, and $C_1$ and $C_2$ are equivalent w.r.t. $\sim_p$.
        \item $\cR_1$ and $\cR_2$ are attention-equivalent w.r.t. $\sim_p$, $C_1$ and $C_2$ are equivalent w.r.t. $\sim_p$ and $C_1$ and $C_2$ are attention invariant.
    \end{itemize}
\end{proposition}

\subsection{Extended GMSC, CMPAs and GNNs}\label{Extended GMSC, CMPAs and GNNs}

Now we take a step further down in generality and examine $\GMSC$-programs, $\CMPA$s and $\GNN$s as run-generators. We introduce a generalization of $\GMSC$ with a counting global modality $\langle e \rangle_{\geq k}$ and a generalization of $\CMPA$s and $\GNN$s with global readout, which allows nodes to receive messages from all nodes in the graph. We also generalize all three by giving them attention functions that turn them into attention run-generators.

In the literature, constant-iteration graph neural networks with global readout have been studied before in, e.g., \cite{DBLP:conf/iclr/BarceloKM0RS20, DBLP:conf/lics/Grohe21, DBLP:conf/lics/Grohe23}.
For example, Barceló et al. \cite{DBLP:conf/iclr/BarceloKM0RS20} showed that constant-iteration $\GNN[\R]$s with global readout can define any property definable in two-variable first-order logic with counting quantifiers. In Section \ref{section: Capturing semantics} we will expand our logical characterizations and show that $\GMSC$ with counting global modality has the same expressive power as $\GNNF$s with global readout, which should not be too surprising. 

$\MSC$ and distributed computing based on circuits with attention was first studied in \cite{dist_circ_mfcs}. Also, the diamond-free fragment of $\MSC$ and recurrent neural networks with attention were studied in \cite{ahvonen_et_al:LIPIcs.CSL.2024.9}.

\subsubsection*{Extending GMSC, CMPAs and GNNs with global diamonds, global readout and attention}

A \textbf{$\Pi$-formula of $\GML$ with counting global modality} (or $\GMLG$) is the same as an ordinary $\Pi$-formula of $\GML$, but we also allow construction of formulae of the type $\langle e \rangle_{\geq k} \varphi$, where the semantics are defined as follows:
\[
    G, v \models \langle e \rangle_{\geq k} \varphi \text{ if and only if } \abs{\{\, u \in V \mid G, u \models \varphi\,\}} \geq k.
\]
$(\Pi, \cT)$-schemata of $\GMSC$ with counting global modality ($\GMSCG$) are defined analogously. A $(\Pi, \cT)$-program of $\GMSCG$ is otherwise the same as that of $\GMSC$, but the terminal clauses are $\Pi$-formulae of $\GMLG$ and the iteration clauses are $(\Pi, \cT)$-schemata of $\GMSCG$. The \textbf{modal depth} of a $\GMLG$-formula or $\GMSCG$-schema is then the maximum number of nested modalities in the formula; for example, the formula $\Diamond_{\geq 2} \langle e \rangle_{\geq 3} p$ has modal depth $2$.

A \textbf{$\CMPA$ over $\Pi$ with global readout} ($\CMPAG$) is otherwise the same as a $\CMPA$, except that the transition function is of the form $\delta \colon Q \times \cM(Q) \times \cM(Q) \to Q$ where a node's new state depends also on the states of all nodes in the graph. More formally, if $x^{t}_{v}$ denotes the state of $v$ in round $t$ in the $\Pi$-labeled graph $G = (V, E, \lambda)$, then
\[
    x^{t+1}_{v} = \delta\left(x^{t}_{v}, \{\{\, x^{t}_{u} \mid (v,u) \in E \,\}\}, \{\{\, x^{t}_{u} \mid u \in V \,\}\}\right).
\]
We define $\FCMPAG$s and bounded $\CMPAG$s in the natural way.
Given $k \in \N$, we say that the transition function $\delta$ is \textbf{invariant w.r.t. $k$}, if $\delta(q, M, N) = \delta(q, M_{|k}, N_{|k})$ for all multisets $M, N \in \cM(Q)$ and $q \in Q$, i.e., $\delta$ can be equivalently defined as a function $Q \times \cM_k(Q) \times \cM_k(Q) \to Q$.
Therefore, a $k$-$\CMPAG$ is a $\CMPAG$, where the transition function is invariant w.r.t. $k$, and a bounded $\CMPAG$ is a $k$-$\CMPAG$ for some $k$.

A \textbf{graph neural network (for $(\Pi, d)$) with global readout} ($\GNNG$) is otherwise the same as a $\GNN$, but the transition function is $\delta \colon \R^{d} \times \cM(\R^{d}) \times \cM(\R^{d}) \to \R^{d}$ where $\delta(x, y, z) = \COM(x, \AGG(y), \READ(z))$ and $\READ \colon$ $\cM(\R^{d}) \to \R^{d}$ is a \textbf{readout function}. If $x^{t-1}_{u}$ denotes the feature vector of $u$ in round $t-1$ in the $\Pi$-labeled graph $(W, R, V)$, then
\[
    x^{t}_{v} = \COM\left(x^{t-1}_{v}, \AGG\left(\{\{\, x^{t-1}_{u} \mid (v,u) \in E \,\}\}\right), \READ\left(\{\{\, x^{t-1}_{u} \mid u \in V \,\}\}\right)\right).
\]
$\GNNFG$s, R-simple aggregate-combine $\GNNG$s and R-simple aggregate-combine $\GNNFG$s follow in the natural way.

Now, we define the generalization w.r.t.~attention. We call a function $A \colon \fG(\Pi) \to S_{\{0,1\}}$ where $S_{\{0,1\}}$ is the set of run-skeletons w.r.t.~$\{0,1\}$ an \textbf{attention function}. A $(\Pi, \cT)$-program of $\GMSC$ \textbf{with attention} ($\GMSCA$) is then a $\GMSC$-program with an attention function, and likewise for $\CMPAA$s and $\GNNA$s. We can combine counting global modality or global readout with attention and obtain $\GMSCAG$, $\CMPAAG$ and $\GNNAG$.

We also define a fragment of normal form programs for all variants of $\GMSC$. Intuitively, $\GMSC$-programs differ from $\GNN$s in their ability to obtain information from multiple steps away in a single iteration, so we place some restrictions on the programs to bring them closer to $\GNN$s. A $(\Pi, \cT)$-program of $\GMSC$ is in \textbf{normal form} ($\GMSCN$) if the modal depth of the body of each terminal clause is $0$, and the modal depth of the body of each iteration clause is at most $1$. We obtain the classes of $\GMSCGN$, $\GMSCAN$ and $\GMSCAGN$ analogously. Normal form programs of modal substitution calculus $\MSC$ ($\MSC$ being the fragment of $\GMSC$ where only diamonds $\Diamond_{\geq 1}$ are allowed) have been studied before in \cite{dist_circ_mfcs}.

\subsubsection*{Generalized semantics for extensions}

Now, we define a semantics for $\GMSCAG$-programs, $\CMPAAG$s and $\GNNAG$s. Informally, all three can be seen as attention-run generators $(R, f)$ (or run-generators in the case of $\GMSCA$, $\CMPAA$s and $\GNNA$s) that generate an a-run for every graph that tells for each round and for each node which head predicates are true in the case of $\GMSC$, and the node's state or feature vector in the case of $\CMPAAG$s and $\GNNAG$s. The assignment tells whether a predicate is appointed or whether a state or feature vector is accepting.

More formally,
let $\Lambda$ be a $(\Pi, \cT)$-program of standard $\GMSC$ and let $\cA \subseteq \cT$ be the set of appointed predicates of $\Lambda$. Now, $\Lambda$ is simply a run-generator $(R, f)$ defined as follows. For each $\Pi$-labeled graph $G = (V, E, \lambda)$, $R(G)$ is a run $(r, f) \in \textit{runs}(V, \cP(\cT), \{1\})$ such that for each $u \in V$
\[
    r_{i}(u) = \{\, X \in \cT \mid G,u \models X^{i} \,\}
\]
and
$f(\cX) = 1$ if $\cA \cap \cX \neq \emptyset$ and $\cX \notin \dom(f)$ otherwise. 
If $\Lambda$ is a $\GMSCAG$-program, then $\Lambda$ is an attention run-generator $(R, f)$ where $R(G) = (r, f, \alpha) \in \textit{a-runs}(V, \cP(\cT), \{1\})$ where $\alpha = A(G)$, where $A$ is the attention function of the program. 

Analogously $\CMPAAG$s (resp. $\GNNAG$s) are attention run-generators: in each round, the run skeleton $r$ labels each node with the state (resp., feature vector) where the automaton (resp., the $\GNN$) is at the node in that round, $f(q) = 1$ if $q$ is an accepting state (resp., accepting feature vector) and $\alpha= A(G)$ where $A$ is the attention function specified by the automaton (resp., the $\GNN$).

While $\GMSCAG$-programs, $\CMPAAG$s and $\GNNAG$s only use one color, it is possible to generalize them to use multiple colors instead by omitting appointed predicates, accepting states and feature vectors, and replacing them with an assignment $f \colon S \to \cA$ directly, where $\cA$ is a set of colors. More formally, a \textbf{generalized} $\GMSC$-program over $(\Pi, \cT, \cA)$ is a pair $(\Lambda, f)$ where $f \colon \cP(\cT) \to \cA$ is an assignment, and appointed predicates are not included in $\Lambda$. We identify $(\Lambda, f)$ with the run-generator $(R, f)$, where $R$ is induced by $\Lambda$ as before, but the assignment is $f$ in every generated run. Generalized $\CMPAAG$s and generalized $\GNNAG$s, and associated subclasses, are defined analogously. 
We also define \textbf{simple} programs of $\GMSCAG$ as those generalized $\GMSCAG$-programs $(\Lambda, f)$, where there is a bijection $b \colon \cA \to \cB$, where the elements of $\cB$ are sets, such that for each $\cX \in \cP(\cT)$, $b(f(\cX)) = \bigcup \{\, b(f(\{X\})) \mid X \in \cX \,\}$; intuitively, this means that we color head predicates, and the color of a set of head predicates is the union of the colors of its elements.

\begin{example}\label{example: GNN reals fixed point semantics}
    It is easy to see that the standard semantics for $\GMSC$ is definable via simple $\GMSC$-programs over $(\Pi, \{1\})$, where a set of head predicates is colored $1$ only if it contains an appointed predicate, and a pointed run is accepted if a node visits a set of head predicates colored $1$. Likewise, we can see that the standard semantics for $\GNN[\R]$s is definable via $\GNN[\R]$s for $(\Pi, d, \{1\})$, where accepting feature vectors are colored $1$ and a pointed run is accepted if the node visits a state colored $1$. The fixed-point semantics for $\GNN[\R]$s in \cite{Pfluger_Tena_Cucala_Kostylev_2024} is definable via a pointed run classifier that instead accepts a pointed run if there is a round $k \in \N$ such that the state of the node is colored $1$ in each round $k' \geq k$.
\end{example}

\begin{remark}
    As noted already in \cite{ahvonen_et_al:LIPIcs.CSL.2024.9} where different attention mechanisms were studied, instead of an attention function, we could define attention predicates for $\GMSC$, attention states for $\CMPA$s and attention feature vectors for $\GNN$s. The function $\alpha$ in the induced attention function $(r, f, \alpha)$ would then be defined such that $\alpha_{i}(v) = 1$ if and only if $r_{i}(v)$ contains an attention predicate, or is an attention state or attention feature vector.
\end{remark}

\subsection{Equivalences between extensions of GMSC and GNNs with run-based semantics}\label{section: Capturing semantics}

In this section, we show equivalence between $\GMSC$-programs and $\GNNF$s in the setting with global readout and run classifications. The main results of this section are theorems~\ref{theorem: GMSC[1] + global readout = GNN[F] + global readout} and \ref{theorem: equal expressive power with run-similar-equivalent run classifiers}, which show that the fragment of generalized normal form $\GMSCAG$ is run-equivalent to generalized $\GNNFAG$s and has the same expressive power with 
equivalent run classifications.
We obtain similar results in theorems \ref{theorem: GMSCAG + global readout = GNNF + global readout = R-simple GNNF + global readout} and \ref{thrm: equal expressive power with attention-equivalent} with \emph{attention}-equivalence which extends the result to include generalized R-simple aggregate-combine $\GNNFAG$s and $\GMSCAG$ not in normal form.
As reductions of Theorem \ref{thrm: equal expressive power with attention-equivalent} we show in Corollary \ref{corollary: graph-size GMSC = graph-size GNN[F] = graph-size R-simple GNN[F]} that the theorem also holds for so-called graph-size graph neural networks (as defined in \cite{Pfluger_Tena_Cucala_Kostylev_2024}) and in Theorem \ref{theorem: büchi and fixed-point} that $\GMSC$, $\GNNF$s and R-simple aggregate-combine $\GNNF$s (without global modality or global readout) have the same expressive power with fixed-point semantics (as defined in \cite{Pfluger_Tena_Cucala_Kostylev_2024}) as well as Büchi semantics. Finally in Theorem \ref{theorem: convergence fixed-point}, we show the same for a semantics based on global fixed-points, which in this case corresponds to the acceptance condition in \cite{scarcelli}.

Throughout this section, $\Pi$ is a set of node label symbols, $\cT$ is a finite set of head predicates, $\cB$ and $\cB'$ 
are sets of colors, $\fP$ and $\fP'$ are partitions of $\cB$ abd $\cB'$ respectively, and $p \colon \fP \to \fP'$ is a bijection.

\subsubsection{Normal form GMSC+AG and GNN[F]+AGs are run-equivalent}\label{Normal form GMSC+AG and floating-point GNN[F]+AGs are equivalent}

In this section, we show that the normal form fragment of generalized $\GMSCAG$ and generalized $\GNNFAG$s are run-equivalent with the same expressive power. The former is stated in Theorem \ref{theorem: GMSC[1] + global readout = GNN[F] + global readout} and the latter in Theorem \ref{theorem: equal expressive power with run-similar-equivalent run classifiers}. The translations we construct are robust, meaning that we can omit global modality/readout, attention, or both, and the characterization still holds.

First, we show how to translate a generalized normal form $\GMSCAG$-program to a run-similar generalized $\GNNFAG$.
\begin{proposition}\label{proposition: GMSC[1] + global modality < k-FCMPA + global readout}
    For each $(\Pi, \cB)$-program of $\GMSCAGN$, we can construct a run-similar $\GNNFAG$ over $(\Pi, \cB')$ w.r.t. any color similarity $\sim_p$.
\end{proposition}
\begin{proof}
    Let $(\Lambda, f_{\Lambda})$ be a $(\Pi, \cT, \cB)$-program of $\GMSCAGN$. We construct a run-similar bounded $\FCMPAAG$ $(\cA, f_{\cA})$, which is trivial to translate to a run-similar $\GNNFAG$ over $(\Pi, \cB')$.

    The construction is similar to that in the proof of Lemma \ref{lemma: GMSC to k-FCMPA}. For the set $Q$ of states, we create a state $q_{\cX, P}$ for each set $\cX \subseteq \cT$ of head predicates and each set $P \subseteq \Pi$ of node label symbols that appear in $\Lambda$. For the initialization function $\pi$, we define that $\pi(P) = q_{\cX, P}$ with $\cX$ the set of head predicates whose terminal clause is true when exactly the node label symbols in $P$ are true. For the transition function $\delta$, we define that $\delta(q_{\cX, P}, N, N') = q_{\cX', P}$, where $\cX'$ is the set of head predicates that become true when exactly the head predicates $\cX$ and node label symbols $P$ are true in the node, the sets of true head predicates and node label symbols of neighbours are exactly as in $N$ and those of nodes in the graph in general are exactly as in $N'$. The attention function of $\cA$ is identical to that of $\Lambda$. Finally, for the assignment $f_{\cA}$, if $f_{\Lambda}(\cX) = b \in \cP \in \fP$, then we define for all $P \subseteq \Pi$ that $f_{\cA}(q_{\cX, P}) = b'$ for some $b' \in p(\cP) \in \fP'$. If $\cX \notin \dom(f_{\Lambda})$, then $q_{\cX, P} \notin \dom(f_{\cA})$ for all $P \subseteq \Pi$. It is easy to verify that $(\Lambda, f_{\Lambda})$ and $(\cA, f_{\cA})$ are run-similar w.r.t. $\sim_{p}$.
\end{proof}

We then show how to translate a generalized $\GNNFAG$ into a run-similar generalized normal form $\GMSCAG$-program that is simple.
\begin{proposition}\label{proposition: GMSC + global modality > k-FCMPA + global readout}
    For each $\GNNFAG$ over $(\Pi, \cB)$, we can construct a run-similar simple $(\Pi, \cB')$-program of $\GMSCAGN$ w.r.t. any color similarity $\sim_p$.
\end{proposition}
\begin{proof}
    For each $\GNNFAG$ over $(\Pi, \cB)$, it is trivial to construct a run-similar bounded $\FCMPAAG$ $(\cA, f_{\cA})$ over $(\Pi, \cB)$. It is left to construct a run-similar simple $(\Pi, \cB')$-program $(\Lambda, f_{\Lambda})$ of $\GMSCAGN$.

    The construction is similar to the one in the proof of Lemma \ref{lemma: k-FCMPA to GMSC}, and the construction of head predicates and terminal clauses is in fact identical. For iteration clauses, we define for all $q, q' \in Q$ the finite set $M_k(q, q') = \{\, (S_{1}, S_{2}) \in \cM_k(Q) \times \cM_{k}(Q) \mid \delta(q, S_{1}, S_{2}) = q' \,\}$, where $k$ is the bound of $\cA$. For each $(S_{1}, S_{2}) \in M_k(q, q')$ we define
    \[
        \begin{aligned}
            \varphi_{(S_{1}, S_{2})} \colonequals &\bigwedge_{q \in Q,\, S_{1}(q) = n,\, n < k} \Diamond_{=n} X_{q} & &\land & &\bigwedge_{q \in Q,\, S_{1}(q) = k} \Diamond_{\geq k} X_{q} \\
            \land &\bigwedge_{q \in Q,\, S_{2}(q) = n,\, n < k} \langle e \rangle_{=n} X_{q} & &\land & &\bigwedge_{q \in Q,\, S_{2}(q) = k} \langle e \rangle_{\geq k} X_{q}.
        \end{aligned} 
    \]
    Then the iteration clause for $X_q$ is defined by
    \[
        X_q \colonminus \bigwedge_{q' \in Q} \Big( X_{q'} \rightarrow \bigvee_{(S_{1}, S_{2}) \in M_k(q', q)} \varphi_{(S_{1}, S_{2})} \Big).
    \]
    The attention function of $\Lambda$ is the attention function of $\cA$. For the assignment $f_{\Lambda}$, we define that if $f_{\cA}(q) = b \in \cP \in \fP$, then $f_{\Lambda}(\{X\}) = b' \in p(\cP) \in \fP'$, and if $q \notin \dom(f_{\Lambda})$, then $\{X_{q}\} \notin \dom(f_{\Lambda})$. For all other $\cX$, we define $f_{\Lambda}(\cX)$ arbitrarily such that the program is simple. It is easy to verify that $(\cA, f_{\cA})$ and $(\Lambda, f_{\Lambda})$ are run-similar w.r.t. $\sim_{p}$.
\end{proof}

In particular, Propositions \ref{proposition: GMSC[1] + global modality < k-FCMPA + global readout} and \ref{proposition: GMSC + global modality > k-FCMPA + global readout} also hold when we simultaneously drop attention, globality or both from each class.
Thus, we obtain the following theorems.

\begin{theorem}\label{theorem: GMSC[1] + global readout = GNN[F] + global readout}
    Over any $\Pi$, the following are run-equivalent w.r.t. any color similarity $\sim_p$:
    \begin{itemize}[itemsep=0.25em,topsep=0pt,parsep=0pt,partopsep=0pt]
        \item generalized $\GMSCAGN$ and generalized $\GNNFAG$s,
        \item generalized $\GMSCAN$ and generalized $\GNNFA$s,
        \item generalized $\GMSCGN$ and generalized $\GNNFG$s,
        \item generalized $\GMSCN$ and generalized $\GNNF$s.
    \end{itemize}
\end{theorem}

\begin{theorem}\label{theorem: equal expressive power with run-similar-equivalent run classifiers}
    The following have the same expressive power with any pa-run classifiers that are equivalent w.r.t. any color similarity $\sim_p$:
    \begin{itemize}[itemsep=0.25em,topsep=0pt,parsep=0pt,partopsep=0pt]
        \item generalized $\GMSCAGN$ and generalized $\GNNFAG$s,
        \item generalized $\GMSCAN$ and generalized $\GNNFA$s,
        \item generalized $\GMSCGN$ and generalized $\GNNFG$s,
        \item generalized $\GMSCN$ and generalized $\GNNF$s.
    \end{itemize}
\end{theorem}

\subsubsection{GMSC+AG and GNN[F]+AGs are attention-equivalent}\label{GMSC+AG and GNN[F]+AGs are equivalent}

In this section, we show that $\GMSCAG$, $\GNNFAG$s and R-simple aggregate-combine $\GNNFAG$s are attention-similar and have the same expressive power. The former is stated in Theorem \ref{theorem: GMSCAG + global readout = GNNF + global readout = R-simple GNNF + global readout} and the latter in Theorem \ref{thrm: equal expressive power with attention-equivalent}. We also obtain a reduction of Theorem \ref{thrm: equal expressive power with attention-equivalent} for the case where the number of iterations is determined by the size of the input graph. The translations are robust, meaning that counting global modality and global readout can be removed such that the characterization still holds.

First, we show that generalized $\GMSCAG$ and generalized $\GMSCAGN$ are attention-equivalent.
\begin{theorem}\label{theorem: GMSC + global modality = GMSC[1] + global readout}
    Over any $\Pi$, generalized $\GMSCAG$ and generalized $\GMSCAGN$ are attention-equivalent w.r.t. any color similarity $\sim_p$.
\end{theorem}
\begin{proof}
    It is trivial to translate a $(\Pi, \cT, \cB)$-program of $\GMSCAGN$ into an attention-similar $(\Pi, \cT, \cB')$-program of $\GMSCAG$ by simply changing the assignment. We show the other direction, which is similar to the proof of Lemma \ref{lemma: GMSC to GMSC[1]}. Let $(\Lambda, f_{\Lambda})$ be a $(\Pi, \cT, \cB)$-program of $\GMSCAG$. We construct an attention-similar $(\Pi, \cT', \cB)$-program $(\Gamma, f_{\Gamma})$ of $\GMSCAG$ (where $\cT \subseteq \cT'$) by modifying the proof of Lemma \ref{lemma: GMSC to GMSC[1]} as follows. We replace the clock in the construction with a similar clock $T_{1}, \dots, T_{n-1}$ with the terminal clauses $T_{\ell}(0) \colonminus \bot$ for all $\ell \in [n-1]$ and the iteration clauses $T_{1} \colonminus \bigwedge_{i = 1}^{n-1} \neg T_{i}$ and $T_{\ell} \colonminus T_{\ell-1}$ for all $\ell \in \{2, \dots, n-1\}$. Then we modify the iteration clauses $X \colonminus \psi$ of all other head predicates in $\cT' \setminus \cT$ into $X \colonminus \neg T_{n-1} \land \psi_{X}$. This ensures that all head predicates in $\cT' \setminus \cT$ are false in every round $n(i+1)$ where $i \in \N$, which happen to be the rounds where $\Gamma$ has finished simulating the $(i-1)$th iteration of $\Lambda$, meaning that if a node $u$ in a graph~$G$ has the label $\cX$ in round $i$ of $\Lambda$, then it is also the label of $u$ in round $n(i+1)$ of $\Gamma$. Thus for the attention function, we define $A_{\Gamma}$ such that for all graphs $G$ and nodes~$u$ of~$G$, $i$ is an attention round of $A_{\Lambda}(G)$ at $u$ if and only if $n(i+1)$ is an attention round of $A_{\Gamma}(G)$ at $u$. We define the assignment $f_{\Gamma} \colon \cP(\cT') \to \cB'$ as follows. 
    For all $\cX \in \cP(\cT)$, we define that if $f_{\Lambda}(\cX) = b \in \cP \in \fP$, then $f_{\Gamma}(\cX) = b'$ for some $b' \in p(\cP) \in \fP'$, and if $\cX \notin \dom(f_{\Lambda})$, then $\cX \notin \dom(f_{\Gamma})$. For all $\cX \in \cP(\cT') \setminus \cP(\cT)$, we define $f_{\Gamma}(\cX)$ arbitrarily (including the possibility that $\cX \notin \dom(f_{\Gamma})$).
    It is easy to verify that $(\Lambda, f_{\Lambda})$ and $(\Gamma, f_{\Gamma})$ are attention-equivalent w.r.t. $\sim_{p}$.
\end{proof}

Next, as a special case, we show how to translate a generalized $\GMSCAG$-program into an attention-similar generalized R-simple aggregate-combine $\GNNFAG$.
\begin{lemma}\label{lemma: GMSC + global modality < R-simple ACGNN[F] + global readout}
    For each $(\Pi, \cB)$-program of $\GMSCAG$, we can construct an attention-similar R-simple aggregate-combine $\GNNFAG$ over $(\Pi, \cB')$ w.r.t. any color similarity $\sim_p$.
\end{lemma}
\begin{proof}
    Let $(\Lambda, f_{\Lambda})$ be a $\GMSCAG$-program over $(\Pi, \cT, \cB)$. We amend $\Gamma$ and $\cG_{\Gamma}$ (henceforth just $\cG$) in the proof of Lemma \ref{lemma: GMSC to simple GNN} as follows. For $\Gamma$, we define $A_{\Gamma}$ such that for all graphs $G$, if $\alpha = A_{\Lambda}(G)$ and $\alpha' = A_{\Gamma}(G)$, then $\alpha'_{0}(u) = 0$ for all nodes $u$ and $\alpha'_{i+1}(u) = \alpha_{i}(u)$, because $\Gamma$ stalls $\Lambda$ by one round. The assignment $f_{\Gamma} \colon \cP(\cT) \to \cB$ is simply $f_{\Lambda}$. It is clear that $(\Lambda, f_{\Lambda})$ and $(\Gamma, f_{\Gamma})$ are run-equivalent w.r.t. $\sim_{p}$, including the case where the partitions consist of singletons and $p$ is the identity function. For $\cG$, we add a third matrix $R$ or global readouts, and the transition function is thus of the form
    \[
        \COM(x,y,z) = \sigma(x \cdot C + y \cdot A + z \cdot R + b).
    \]
    We introduce the following rule for subschemata with counting global modality:
    \begin{itemize}
        \item If $\varphi_{\ell} = \langle e \rangle_{\geq K} \varphi_{k}$, then $R_{k, \ell} = 1$ and $b_{\ell} = -K + 1$. 
    \end{itemize}
    For the attention function, we define $A_{\cG}$ such that for all graphs $G$ and nodes $u$ of $G$, $i$ is an attention round of $A_{\Gamma}(G)$ at $u$ if and only if $(L+1)i$ is an attention round of $A_{\cG}(G)$ at~$u$. Lastly, we define the assignment $f_{\cG}$ as follows. In restriction to feature vectors whose last element is $1$ (i.e., $\bv_{D+L+1} = 1$), we define that if $f_{\Gamma}(\cX) = b \in \cP \in \fP$, then $f_{\cG}(\bv) = b'$ for some $b' \in p(\cP) \in \fP'$ for all $\bv$ where $\bv_{\ell} = 1$ if $\varphi_{\ell}$ is a head predicate in $\cX$ and $\bv_{\ell} = 0$ if $\varphi_{\ell}$ is a head predicate not in $\cX$ (and if $\cX \notin \dom(f_{\Gamma})$, then $\bv \notin \dom(f_{\cG})$). In restriction to feature vectors whose last element is not $1$ (i.e., $\bv_{D+L+1} \neq 1$), we may define $f_{\cG}(\bv)$ arbitrarily (including the possibility that $\bv \notin \dom(f_{\cG})$). It is simple to verify that $(\Lambda, f_{\Lambda})$ and $(\cG, f_{\cG})$ are attention-equivalent w.r.t. $\sim_{p}$.
\end{proof}

In particular, Theorem \ref{theorem: GMSC + global modality = GMSC[1] + global readout} and Lemma \ref{lemma: GMSC + global modality < R-simple ACGNN[F] + global readout} also hold when we simultaneously drop globality from each class. Thus, we obtain the following theorems.

\begin{theorem}\label{theorem: GMSCAG + global readout = GNNF + global readout = R-simple GNNF + global readout}
    Over any $\Pi$, the following are attention-equivalent w.r.t. any color similarity $\sim_p$:
    \begin{itemize}[itemsep=0.25em,topsep=0pt,parsep=0pt,partopsep=0pt]
        \item generalized $\GNNFAG$s, generalized R-simple aggregate-combine $\GNNFAG$s and generalized $\GMSCAG$,
        \item generalized $\GNNFA$s, generalized R-simple aggregate-combine $\GNNFA$s and generalized $\GMSCA$.
    \end{itemize}
\end{theorem}

\begin{theorem}\label{thrm: equal expressive power with attention-equivalent}
    The following have the same expressive power with any attention-invariant pa-run classifiers that are equivalent w.r.t. any color similarity $\sim_p$:
    \begin{itemize}[itemsep=0.25em,topsep=0pt,parsep=0pt,partopsep=0pt]
        \item generalized $\GNNFAG$s, generalized R-simple aggregate-combine $\GNNFAG$s and generalized $\GMSCAG$,
        \item generalized $\GNNFA$s, generalized R-simple aggregate-combine $\GNNFA$s and generalized $\GMSCA$.
    \end{itemize}
\end{theorem}

We next show a reduction of the above theorem for so-called graph-size $\GNN$s studied in~\cite{Pfluger_Tena_Cucala_Kostylev_2024}, where the $\GNN$ is given a function $\mathrm{Iter} \colon \N \to \N$ that tells how many times the $\GNN$ is iterated depending on the size of the graph. A \textbf{graph-size $\GNN[\R]$} over $\Pi$ is a generalized $\GNN[\R]$ over $(\Pi, \{1\})$ (where $\{1\}$ is the set of colors) where the attention function $A$ is defined such that there is a function $\mathrm{Iter} \colon \N \to \N$ such that $A(V, E, \lambda) = \alpha$ where for all nodes $v$, $\alpha_{i}(v) = 1$ if and only if $i = \mathrm{Iter}(\abs{V})$. \textbf{Graph-size $\GMSC$-programs} and \textbf{graph-size $\CMPA$s} are defined analogously, as well as variations such as graph-size $\GNNF$s.

The classification used for graph-size $\GNN$s in \cite{Pfluger_Tena_Cucala_Kostylev_2024} classifies a node according to its label after the last iteration. We define an equivalent pa-run classification that only looks at a node's color in the attention round in the case where there is only one attention round. A \textbf{one-round run classification} is a pa-run classification $C \colon \textit{pa-runs}(S, \{1\}) \to \{0,1\}$ such that for all pointed attention runs $(r, f, \alpha, v)$ where $\alpha$ has exactly one attention round for each node, we have that $C(r, f, \alpha, v) = 1$ if and only if $f(r_{i}(v)) = 1$ where $i$ is the lone attention round of $\alpha$ at $v$.

\begin{corollary}\label{corollary: graph-size GMSC = graph-size GNN[F] = graph-size R-simple GNN[F]}
    The following have the same expressive power with one-round run classifiers:
    \begin{itemize}[itemsep=0.25em,topsep=0pt,parsep=0pt,partopsep=0pt]
        \item generalized $\GNNFAG$s, generalized R-simple aggregate-combine $\GNNFAG$s and generalized $\GMSCAG$,
        \item generalized $\GNNFA$s, generalized R-simple aggregate-combine $\GNNFA$s and generalized $\GMSCA$.
    \end{itemize}
\end{corollary}

\begin{remark}\label{remark:GNNR+Gs = VGMLG}
    In addition to Theorem \ref{theorem: k-GNN[F] = k-FCMPA = GMSC}, Theorems \ref{omega-GML = GNN = CMPAs} and \ref{constant iteration GNN reals} also trivially generalize for global readouts with simple modifications to the proofs by defining full types of $\GMLG$. In other words, $\GNNG$s and infinite disjunctions of $\GMLG$-formulae have the same expressive power. Analogously, constant-iteration $\GNNG$s and depth-bounded infinite disjunctions of $\GMLG$-formulae have the same expressive power.
\end{remark}

\subsubsection{Reductions for fixed-point and Büchi semantics}\label{Reductions for fixed-point and Büchi semantics}

In this section, we show in Theorem \ref{theorem: büchi and fixed-point} that $\GMSC$, $\GNNF$s and R-simple aggregate-combine $\GNNF$s with fixed point semantics (resp. Büchi semantics) have the same expressive power. To show this, we construct a reduction by applying Theorem \ref{thrm: equal expressive power with attention-equivalent}.  

A \textbf{fixed-point run classifier} is a Boolean p-run classifier $C \colon \textit{p-runs}(S, \{1\}) \to \{0,1\}$ defined as follows. For each $(r, f, v) \in \textit{p-runs}(S, \{1\})$ we have $C(r, f, v) = 1$, if there exists a $k \in \N$ such that $f(r_{k'}(v)) = 1$ for all $k' \geq k$, and otherwise $C(r, f , v) = 0$. 
Note that $\GNN[\R]$s with fixed-point run classifiers induce the same pointed fixed-point semantics as in the Example \ref{example: GNN reals fixed point semantics}.
Respectively, 
a \textbf{Büchi run classifier} is a Boolean p-run classifier $C \colon \textit{p-runs}(S, \{1\}) \to \{0,1\}$ defined as follows. For each $(r, f, v) \in \textit{p-runs}(S, \{1\})$ we have $C(r, f, v) = 1$ if the cardinality of the set $\{\, i \in \N \mid f(r_i(v)) = 1 \,\}$ is $\abs{\N}$, and otherwise $C(r, f, v) = 0$.

\begin{theorem}\label{theorem: büchi and fixed-point}
    $\GMSC$, $\GNNF$s and R-simple aggregate-combine $\GNNF$s have the same expressive power with fixed-point run classifiers and also with Büchi run classifiers.
\end{theorem} 
\begin{proof}
    Firstly, each $\GMSC$-program, $\GNNF$ and R-simple aggregate-combine $\GNNF$ is easy to translate into an equivalent generalized one with attention by defining that each round is an attention round and defining an assignment that is consistent with the set of appointed predicates, accepting states or accepting feature vectors. Now, we can apply Theorem \ref{thrm: equal expressive power with attention-equivalent} to these models.

    Let $\Lambda$ be a $\GMSC$-program over $\Pi$. Let $\Lambda^*$ be a generalized $\GMSCA$-program over $(\Pi, \cB)$ obtained from $\Lambda$ by associating it with an attention function that for every graph induces an attention run where every round is an attention round in every node. We assume that $\cB = \cB' = \{1\}$, $\fP = \fP' = \{\{1\}\}$ and $p \colon \fP \to \fP'$ is the identity function. We modify the proofs of Theorem \ref{theorem: GMSC + global modality = GMSC[1] + global readout} and Lemma \ref{lemma: GMSC + global modality < R-simple ACGNN[F] + global readout} by modifying the assignments as follows. In Theorem \ref{theorem: GMSC + global modality = GMSC[1] + global readout}, recall that we defined $f_{\Gamma}$ arbitrarily in restriction to $\cP(\cT') \setminus \cP(\cT)$. To capture fixed-point run classifiers, we define that $f_{\Gamma}(\cX) = 1$ if and only if $\cX \in \cP(\cT') \setminus \cP(\cT)$; this ensures that $\Gamma$ (and subsequently $\cA$) accepts in every auxiliary round. To capture Büchi run classifiers, we define instead that $\cX \notin \dom(f_{\Gamma})$ if $\cX \in \cP(\cT') \setminus \cP(\cT)$; this ensures that in each auxiliary round, $\Gamma$ (and subsequently the automaton $\cA$) does not accept (i.e., does not give the color $1$). In Lemma \ref{lemma: GMSC + global modality < R-simple ACGNN[F] + global readout}, recall that we defined $f_{\cG}$ arbitrarily for feature vectors with a non-zero last element (i.e., when $\bv_{D+L+1} \neq 1$). To capture fixed-point run classifiers, we define that $f_{\cG}(\bv) = 1$ if $\bv_{D+L+1} \neq 1$. To capture Büchi run classifiers, we define instead that $\bv \notin \dom(f_{\cG})$ if $\bv_{D+L+1} \neq 1$. The proofs of propositions \ref{proposition: GMSC[1] + global modality < k-FCMPA + global readout} and \ref{proposition: GMSC + global modality > k-FCMPA + global readout} do not require any modification. 
    
    Thus, by omitting the attention functions and considering instead fixed-point run classifiers (resp. Büchi run classifiers), we see that $\GMSC$, $\GNNF$s and R-simple aggregate-combine $\GNNF$s all have the same expressive power.
\end{proof}

\subsubsection{Characterization with convergence-based fixed-point semantics}\label{section:convergence-based_fixed-point_semantics}

A \textbf{convergence-based fixed-point run classifier} is a Boolean p-run classifier of the form $C \colon \textit{p-runs}(S, \{1\}) \to \{0,1\}$ defined as follows. For each $(r, f, v) \in \textit{p-runs}(S, \{1\})$ we have $C(r,f,v) = 1$ if there exists a $k \in \N$ such that $f(r_k(v)) = 1$ and for all $k' \geq k$ and for all $u \in V$ we have $r_{k'}(u) = r_k(u)$.

First of all it is easy to show that Lemma \ref{lem: simplify GMSC-program} holds also for programs with convergence-based fixed-point semantics.
\begin{lemma}\label{lem: simplify GMSC-program convergence-based}
    With respect to the convergence-based fixed-point semantics,
    for each $\Pi$-program $\Lambda$ of $\GMSC$ 
    with formula depth $D$, we can construct an equivalent $\Pi$-program $\Gamma$ of $\GMSC$ which has the following properties.
    \begin{enumerate}
    \item Each terminal clause is of the form $X (0) \colonminus \bot$.
    \item Each iteration clause has the same formula depth $\max(3, D + 2)$.
    \item If $\varphi \land \theta$ is a subschema of $\Lambda$ such that neither $\varphi$ nor $\theta$ is $\top$, then $\varphi$ and $\theta$ have the same formula depth.
\end{enumerate}
\end{lemma}

Now, it is easy to modify the proof of Lemma \ref{lemma: GMSC to simple GNN} for programs and R-simple GNNs with convergence-based fixed-point semantics.

\begin{lemma}\label{lemma: convergence-based GMSC to simple GNN}
    With respect to convergence-based fixed-point semantics,
    for each $\Pi$-program of $\GMSC$ we can construct an equivalent R-simple aggregate-combine $\GNNF$ over $\Pi$. 
\end{lemma}
\begin{proof}
Let $\Lambda$ be a $\Pi$-program of $\GMSC$ with convergence-based fixed-point semantics.
Informally, an equivalent R-simple aggregate-combine $\GNNF$ (with convergence-based fixed-point semantics) for $\Lambda$ is constructed as follows.
First, from $\Lambda$ we construct an equivalent $\Pi$-program $\Gamma$ of $\GMSC$ with Lemma \ref{lem: simplify GMSC-program convergence-based}. Let $D$ be the formula depth of $\Gamma$.

The rest of the construction is essentially the same as in the proof of Lemma \ref{lemma: GMSC to simple GNN}, but the feature vectors of the constructed GNN do not track the current formula depth that is being evaluated, since otherwise the feature vector would not reach a fixed-point. Instead, at the start of the computation, the GNN is synchronized correctly such that the feature vector reaches the corresponding fixed-point if $\Gamma$ does. 

Now, we explain construction formally.
An equivalent R-simple aggregate-combine $\GNNF$ $\cG_{\Gamma}$ for $\Gamma$ periodically computes a single iteration round of $\Gamma$ in $D+1$ rounds, except the round zero of $\Gamma$ which is simulated repeatedly for the first $D+1$ rounds to synchronize the feature vectors of the GNN. 
The feature vectors used by $\cG_{\Gamma}$ are \emph{binary} vectors $\bv = \bu\bw$ where: 
\begin{itemize}
    \item $\bu$ is the same as in the proof Lemma \ref{lemma: GMSC to simple GNN}, i.e., $\bu$ has one bit per each (distinct) subschema of a body of an iteration clause in $\Gamma$ as well as for each head predicate, and $\bv_{1}$ keeps track of their truth values, and
    \item $\bw$ has $D+1$ bits that keep track of how many rounds of the round zero of $\Gamma$ have been simulated.
\end{itemize}
Note that $\cG_{\Gamma}$ is a $\GNNF$ over $(\Pi, N + D+1)$, where $N$ is the number of (distinct) subschemata of the bodies of the iteration clauses of $\Gamma$ as well as head predicates, and $D$ is the maximum formula depth of the bodies of the iteration clauses. We assume an arbitrary enumeration $\varphi_1, \ldots, \varphi_N$ for the distinct subschemata and head predicates of $\Gamma$.
As in the proof of Lemma \ref{lemma: GMSC to simple GNN}, the floating-point system $S$ for $\cG_{\Gamma}$ is chosen to be large enough.

Recall that $\mathrm{ReLU}^*(x)$ denotes the truncated rectified linear unit.
Next, we define the functions $\pi$ and $\COM(x,y) = \mathrm{ReLU}^*(\bx \cdot C + \by \cdot A + \bb)$ and the set $F$ of accepting states for $\cG_{\Gamma}$. 
The initialization function $\pi$ of $\cG_{\Gamma}$ is the same as in the proof of Lemma \ref{lemma: GMSC to simple GNN}.
For $k, \ell \leq N+D+1$, $C_{k, \ell}$ and $A_{k, \ell}$
are defined in the same way as in the proof of Lemma~\ref{lemma: GMSC to simple GNN}, but we modify the matrix $C$ for head predicates as follows.
For all $\ell \leq N$, if $\varphi_\ell$ is a head predicate $X$ with the iteration clause $X \colonminus \varphi_{k}$, then $C_{k, \ell} = 1$ and $C_{N + D + 1, k} = 1$.
The bottom-right $(D+1) \times (D+1)$ submatrix of $C$ is the $(D+1) \times (D+1)$ matrix that is obtained from the identity matrix by replacing the zero in the second last column in the last row by one, removing the last column, and inserting a zero column in the front of the matrix. More formally, for all $N + 1 \leq \ell < N+D+1$ we have $C_{\ell, \ell+1} = 1$, and we also have $C_{N+D+1, N+D+1} = 1$.
Lastly, we define that all other elements in $C$, $A$ and $\bb$ are $0$s.
The set of accepting feature vectors is the same as in the proof of Lemma \ref{lemma: GMSC to simple GNN}, i.e., $\bv \in F$ if and only if $\bv_\ell = 1$ (i.e., the $\ell$th value of $\bv$ is $1$) for some appointed predicate $\varphi_{\ell}$ and also $\bv_{N+D+1} = 1$.

Recall that the formula depth of $\Gamma$ is $D$. Let $\bv(w)_i^t$ denote the value of the $i$th component of the feature vector in round $t$ at node $w$. 

It is easy to show by induction that for all pointed $\Pi$-labeled graphs $(G, w)$, for all $t \in [0;D]$, and for every schema $\psi_{\ell}$ in $\mathrm{SUB}(\Gamma)$ of formula depth $0$, we have
\[
\bv(w)_\ell^{t} = 1 \text{ if } G, w \models \psi_{\ell}^0 \text{ and $\bv(w)_\ell^{t} = 0$ if $G, w \not\models \psi_{\ell}^0$}.
\]
Also, for all $m > N$ it holds that
\[
\bv(w)_m^{t} = 1 \text{ if } m = N + t + 1 \text{ and } \bv(w)_m^{t} = 0 \text{ if } m \neq N + t + 1
\]
and for every schema $\theta_{\ell}$ in $\mathrm{SUB}(\Gamma)$ of formula depth $d \neq 0$ and for all $t' \in \{d, \ldots,  D \}$, we have
\[
\bv(w)_\ell^{t'} = 1 \text{ if } G, w \models \theta_{\ell}^1 \text{ and $\bv(w)_\ell^{t'} = 0$ if $G, w \not\models \theta_{\ell}^1$}.
\]
Intuitively, the claim shows that the feature vectors remember the initial truth values for the head predicates for the first $D + 1$ rounds. Moreover, during those rounds, the R-simple GNN updates the truth values of other subschemata of $\Gamma$ and those truth values correspond to the truth values in round $1$ w.r.t. $\Gamma$.

Then it is easy to show by the above that for all $n \in \N$, for all pointed $\Pi$-labeled graphs $(G, w)$, and for every schema $\varphi_\ell$ in $\mathrm{SUB}(\Gamma)$ of formula depth $0$ and for all $t \in [0; D]$, we have
\[
\bv(w)_\ell^{(n+1) (D+1) + t} = 1 \text{ if } G, w \models \varphi_\ell^{n+1} \text{ and $\bv(w)_\ell^{(n + 1) (D+1) + t} = 0$ if $G, w \not\models \varphi_\ell^{n+1}$}
\]
and for all $N + D + 1 > m > N$ and $t' > D$ it holds that $\bv(w)_m^{t'} = 0$ and $\bv(w)_{N+D+1}^{t'} = 1$.

Moreover, it is easy to show by induction that for all $n \in \N$, for all pointed $\Pi$-labeled graphs $(G, w)$, for every schema $\psi_\ell$ in $\mathrm{SUB}(\Gamma)$ of formula depth $d \neq 0$ and for all $t \in [0;d-1]$, we have
\[
\bv(w)_\ell^{(n+1) (D+1) + t} = 1 \text{ if } G, w \models \psi_\ell^{n+1} \text{ and $\bv(w)_\ell^{(n + 1) (D+1) + t} = 0$ if $G, w \not\models \psi_\ell^{n+1}$}
\]
and for all $t' \in \{d, \ldots,  D +1\}$, we have
\[
\bv(w)_\ell^{(n+1) (D+1) + t'} = 1 \text{ if } G, w \models \psi_\ell^{n+2} \text{ and $\bv(w)_\ell^{(n + 1) (D+1) + t'} = 0$ if $G, w \not\models \psi_\ell^{n+2}$}.
\]

Now, if $\Gamma$ reaches a fixed-point, then there is round $n \in \N$ such that for all pointed $\Pi$-labeled graphs $(G, w)$ and for all subschemata $\varphi_\ell$ of $\Gamma$, we have $G, w \models \varphi_\ell^n$ iff $G, w \models \varphi_\ell^{n+1}$, i.e., each subschema also reaches a fixed-point.
By the above, this means that if $\Gamma$ reaches a fixed-point then so does $\cG_\Gamma$. 
That is, there exists an $n \in \N$ such that for all pointed $\Pi$-labeled graphs $(G, w)$, and for every schema $\varphi_\ell$ in $\mathrm{SUB}(\Gamma)$, we have
\[
\bv(w)_\ell^{n} = 1 \text{ if } G, w \models \varphi_\ell^{n} \text{ and $\bv(w)_\ell^{n} = 0$ if $G, w \not\models \varphi_\ell^{n}$}
\]
and for all $N + D +1 > m > N$ it holds that
$\bv(w)_m^{n} = 0$ and $\bv(w)_{N + D + 1}^{n} = 1$. 
\end{proof}

With the above Lemma, we are now ready to prove our final result.

\begin{theorem}\label{theorem: convergence fixed-point}
    $\GMSC$, $\GNNF$s and R-simple aggregate-combine $\GNNF$s have the same expressive power with convergence-based fixed-point run classifiers.
\end{theorem}
\begin{proof}
    Trivially, each R-simple aggregate-combine $\GNNF$ is a $\GNNF$. 
    Theorem \ref{theorem: equal expressive power with run-similar-equivalent run classifiers} proves that $\GNNF$s and $\GMSC[1]$-programs have the same expressive power with any pa-run classifiers, i.e., for each $\GNNF$ we can construct an equivalent $\GMSC$-program with respect to convergence-based fixed point semantics. 
    For the translation from $\GMSC$ to R-simple aggregate-combine $\GNNF$s, we use Lemma \ref{lemma: convergence-based GMSC to simple GNN}.
\end{proof}

Theorem \ref{theorem: convergence fixed-point} also holds when restricted to the classes of $\GMSC$-programs, $\GNNF$s, and R-simple aggregate-combine $\GNNF$s that reach a fixed-point in each node with every input, though we do not provide a syntactic characterization for these classes here.
Note that, it is easy to define syntactical characterizations for some subclasses of such $\GMSC$-programs, $\GNNF$s and R-simple $\GNNF$s.
However, we leave the study of these syntactic restrictions for future work.

\end{document}